\documentclass{amsart}
\usepackage[latin1]{inputenc}
\usepackage{latexsym}
\usepackage{amsmath}
\usepackage{amssymb}
\usepackage{amsfonts}
\usepackage{mathrsfs}               % \mathscr
\usepackage{bbm}                    % \mathbbm 
\usepackage{bbold}                  % \mathbb 
\usepackage{textcomp}               
\usepackage{amstext}
\usepackage{enumerate}
\usepackage{units}
\usepackage{stmaryrd}
\usepackage{multicol}
\usepackage{enumitem}
 \newtheorem{thm}{Theorem}[section]
 \newtheorem{cor}[thm]{Corollary}
 \newtheorem{lem}[thm]{Lemma}
 \newtheorem{prop}[thm]{Proposition}
 \theoremstyle{definition}
 
 \theoremstyle{remark}
 \newtheorem{rem}[thm]{Remark}
 \newtheorem{ex}[thm]{Example}
 \newtheorem{hyp}[thm]{Hypothesis}

 \numberwithin{equation}{section}
\newcommand{\CC}{\mathbb{C}}
\newcommand{\EE}{\mathbb{E}}

\newcommand{\NN}{\mathbb{N}}
\newcommand{\PP}{\mathbb{P}}

\newcommand{\RR}{\mathbb{R}}

\newcommand{\supp}{\mathrm{supp}}
\newcommand{\dist}{\mathrm{dist}}

\newcommand{\loc}{\mathrm{loc}}
\newcommand{\Div}{{\mathrm{div}}}
\newcommand{\Id}{\mathrm{d}}

\newcommand{\SPn}[2]{\langle #1|#2\rangle} 
\newcommand{\SPb}[2]{\big\langle #1\big|#2\big\rangle} 
\newcommand{\SPB}[2]{\Big\langle #1\Big|#2\Big\rangle}
\newcommand{\ol}[1]{\overline{#1}} 
 
\newcommand{\mr}[1]{\mathring{#1}}
\newcommand{\wh}[1]{\widehat{#1}}
\newcommand{\wt}[1]{\widetilde{#1}}
\newcommand{\nf}[2]{\nicefrac{#1}{#2}}
\newcommand{\eh}{{\nf{1}{2}}}
\newcommand{\mh}{{-\nf{1}{2}}}
\newcommand{\cA}{\mathcal{A}}

\newcommand{\cD}{\mathcal{D}} 
\newcommand{\cQ}{\mathcal{Q}}

\newcommand{\cK}{\mathcal{K}}
 
\newcommand{\cM}{\mathcal{M}}       
 
\newcommand{\sC}{\mathscr{C}}
\newcommand{\sD}{\mathscr{D}} 
\newcommand{\sE}{\mathscr{E}}
\newcommand{\sF}{\mathscr{F}}
\newcommand{\sS}{\mathscr{S}}

\newcommand{\sU}{\mathscr{U}}

\newcommand{\sK}{\mathscr{K}}\newcommand{\sW}{\mathscr{W}}
         
\newcommand{\sM}{\mathscr{M}}       
 
\newcommand{\fA}{\mathfrak{A}}
\newcommand{\fB}{\mathfrak{B}} 
\newcommand{\fC}{\mathfrak{C}}

\newcommand{\fF}{\mathfrak{F}}
\newcommand{\fS}{\mathfrak{S}}
\newcommand{\fH}{\mathfrak{H}}

\newcommand{\fJ}{\mathfrak{J}}
\newcommand{\fK}{\mathfrak{K}}

\newcommand{\V}[1]{\boldsymbol{#1}}
\newcommand{\ve}{\varepsilon}
\newcommand{\vp}{\varphi}
\newcommand{\vo}{\varpi}
\newcommand{\vk}{\varkappa}
\newcommand{\vr}{\varrho}
\newcommand{\vt}{\vartheta}

\newcommand{\id}{\mathbbm{1}}                   % Identity 
\newcommand{\dom}{\cD}                          % domain of definition
\newcommand{\fdom}{\cQ}                         % form domain 
\newcommand{\HR}{\mathscr{H}}                   % Hilbert space
\newcommand{\HP}{\mathfrak{k}}                  % photon Hilbert space 
\newcommand{\LO}{\mathcal{B}}                   % set of bounded, linear operators
\newcommand{\ad}{a^\dagger}                     % creation op.

\newcommand{\ee}{\mathfrak{e}} 
\newcommand{\Geb}{\Lambda}

\newcommand{\restr}{\mathord{\upharpoonright}}

\renewcommand{\Re}{\mathrm{Re}}
\renewcommand{\Im}{\mathrm{Im}}
\renewcommand{\le}{\leqslant}        
\renewcommand{\ge}{\geqslant}  
\begin{document}

\title[Feynman-Kac formulas]{Feynman-Kac formulas for Dirichlet-Pauli-Fierz operators 
with singular coefficients}

\author[O.~Matte]{Oliver Matte}

\address{Oliver Matte, Institut for Matematiske Fag, Aalborg Universitet,
Skjernvej 4A, DK-9220 Aalborg, Denmark}
\email{oliver@math.aau.dk}

\begin{abstract}
We derive Feynman-Kac formulas for Dirichlet realizations of Pauli-Fierz operators 
generating the dynamics of nonrelativistic quantum mechanical matter particles, which are minimally
coupled to both classical and quantized radiation fields and confined to an arbitrary open subset of 
the Euclidean space. Thanks to a suitable interpretation of the involved Stratonovich integrals, we are
able to retain familiar formulas for the Feynman-Kac integrands merely assuming local
square-integrability of the classical vector potential and the coupling function in the quantized vector 
potential. Allowing for fairly general coupling functions becomes relevant when the matter-radiation
system is confined to cavities with inward pointing boundary singularities.
\end{abstract}
\maketitle
%
%\setcounter{tocdepth}{1}
%\tableofcontents
%
%%%%%%%%%%%%%%%%%%%%%%%%%%%%%%%%%%%%%%%%%%%%%%%%%
%%%%%%%%%%%%%%%%%%%%%%%%%%%%%%%%%%%%%%%%%%%%%%%%%
%%%%%%%%%%%%%%%%%%%%%%%%%%%%%%%%%%%%%%%%%%%%%%%%%

\section{Introduction and main results}\label{introsec}

\subsection{General introduction}

\noindent
The main objective of this article is to derive Feynman-Kac formulas for Dirichlet realizations on arbitrary
open subsets $\Geb\subset\RR^\nu$ of Pauli-Fierz operators with possibly quite singular coefficients. 
Pauli-Fierz operators are selfadjoint operators generating the dynamics of nonrelativistic quantum 
mechanical matter particles confined to $\Geb$ and interacting with a quantized radiation field. 
Let $\sF$ denote the bosonic Fock space modelled over the one-boson Hilbert space
\begin{align}\label{defHP}
\HP:=L^2(\cK,\fK,\mu).
\end{align}
We assume the measure space $(\cK,\fK,\mu)$ to be $\sigma$-finite and countably generated,
which entails separability of $\HP$. Define 
\begin{align}\label{defsDGeb}
\sD(\Geb):=C_0^\infty(\Geb),\qquad
\sD(\Geb,\sE)&:=\mathrm{span}_{\CC}\big\{f\psi\big|\,f\in\sD(\Geb),\,\psi\in\sE\big\},
\end{align}
for any complex vector space $\sE$. Then the Dirichlet-Pauli-Fierz operator investigated here
-- we denote it by $H_\Geb$ --
acts in the Hilbert space $L^2(\Geb,\sF)$ and represents the closure of the quadratic form given by
\begin{align}\nonumber
\mr{\mathfrak{h}}_{\Geb}[\Psi]&:=\frac{1}{2}\sum_{j=1}^\nu
\int_\Geb\|(\partial_{x_j}-i{A_j}(\V{x})-i\vp({G}_{j,\V{x}}))\Psi(\V{x})\|_{\sF}^2\Id\V{x}
\\\label{introPFform}
&\quad+\int_\Geb\|\Id\Gamma(\omega)^\eh\Psi(\V{x})\|_{\sF}^2\Id\V{x}
+\int_\Geb V(\V{x})\|\Psi(\V{x})\|_{\sF}^2\Id\V{x},
\end{align}
for all
\begin{align}\label{dommrh}
\Psi\in\dom(\mr{\mathfrak{h}}_{\Geb}):=\sD(\Geb,\fdom(\Id\Gamma(1\vee\omega))).
\end{align} 
In the above expressions, 
$\omega\ge0$, the boson dispersion relation, is a multiplication operator in $\HP$, and 
$\Id\Gamma(\omega)$, the radiation field energy, is its differential second quantization; $\dom(\cdot)$ 
denotes domains and $\fdom(\cdot)$ form domains. 
By {\em coefficients} in \eqref{introPFform} we mean the triple comprised of the electrostatic 
potential\footnote{A negative part will be subtracted from $V$ only in Cor.~\ref{cornegpart} and 
Rem.~\ref{remappl}.} 
\begin{align}\label{Vlocint}
V\in L^1_\loc(\Geb,\RR),\quad V\ge0,
\end{align} 
the classical vector potential 
\begin{align}\label{introA}
\V{A}=(A_1,\ldots,A_\nu)\in L^2_\loc(\Geb,\RR^\nu),
\end{align}
and the coupling function 
\begin{align}\label{introG}
\V{G}=(G_1,\ldots,G_\nu)\in L_\loc^2(\Geb,\HP^\nu),
\end{align}
that determines the interaction between the matter particles and the radiation field. 
As usual $\vp({G}_{j,\V{x}})$ stands for the field operator corresponding to ${G}_{j,\V{x}}:=G_j(\V{x})$.

The present article actually continues our earlier study \cite{Matte2017} of Dirichlet-Pauli-Fierz operators
with singular coefficients where we determined the domain and found natural operator cores of 
these operators. While many technical
results of \cite{Matte2017} hold in greater generality, these main results were obtained
under the assumption that $\V{G}\in L^\infty(\Geb,\fdom(\omega^{-1}+\omega)^\nu)$ 
with a weak divergence $\Div\V{G}\in L^\infty(\Geb,\fdom(\omega^{-1}))$.
This is more than enough to cover the standard model of nonrelativistic quantum electrodynamics
on Euclidean space with an ultraviolet cutoff or ultraviolet regularized models of quantum optics in 
bounded cavities with {\em smooth} boundaries. 
Recall that, according to the general quantization scheme for the
electromagnetic field found in physics textbooks (see, e.g., \cite{Dutra2005}), the coupling function has 
the following form in applications to quantum optics in bounded cavities with $\nu=3$:
\begin{align}\label{introGcav}
\V{G}_{\V{x}}&=\ee\sum_{n=1}^\infty\frac{\chi(\omega(n))}{\sqrt{\omega(n)}}\V{E}_n(\V{x}).
\end{align}
Here $0<\omega(1)\le\omega(2)\le\ldots$ are the strictly positive eigenfrequences of the Maxwell
operator on $\Geb$ with perfect electric conductor boundary conditions. The normalized function
$\V{E}_n$ is the electric component of the eigenvector of the Maxwell operator corresponding to the 
frequency $\omega(n)$. Furthermore, $\ee\in\RR$ is a combination of physical constants, and the 
auxiliary, sufficiently fast decaying function $\chi:[0,\infty)\to[0,1]$ implements the ultraviolet cutoff.

The boundary $\partial\Geb$ of a cavity $\Geb$ might, however, not always be smooth. If 
$\partial\Geb$ has singularities, like polyhedral structures with inward pointing edges and 
corners for instance, then the functions $\V{E}_n$ in \eqref{introGcav} are
singular as well at the inward pointing boundary singularities; 
see, e.g., \cite{CostabelDauge2000} and the references given there. 
In particular, the usual $L^\infty$-conditions imposed on $\V{G}$ in \cite{Matte2017} (and in almost 
all other articles on Pauli-Fierz type operators, dipole approximations being one exception) might not be 
fulfilled in the presence of boundary singularities. This motivates keeping the assumptions on $\V{G}$ 
more general while studying basic qualitative features of Dirichlet-Pauli-Fierz operators. 

In this article we choose to consider a situation where the individual terms in the quadratic form 
\eqref{introPFform} are well-defined and finite for every $\Psi$ as in \eqref{dommrh}. Since
$\Psi$ in \eqref{dommrh} can be the product of any function in $\sD(\Geb)$ and the Fock space 
vacuum, this necessitates \eqref{introA}, \eqref{introG}, and the first condition in \eqref{Vlocint}.
We assume the second condition in \eqref{Vlocint}, since it is often convenient to have it in our proofs
and our main results extend by standard arguments to suitable electrostatic potentials that are 
unbounded from below; see Cor.~\ref{cornegpart}.
(Making sufficient effort, magnetic Schr\"{o}dinger operators can actually be constructed even without 
assuming local square-integrability of the vector potential and local integrability of the electrostatic 
potential \cite{LiskevichManavi1997}.)

A good part of this article is made up of analyzing quadratic forms and diamagnetic inequalities and
here the condition \eqref{introG} is in fact sufficient. The Feynman-Kac formulas will, however, only
be valid when the operators $\vp(G_{j,\V{x}})$ admit the interpretation as position observables
of the radiation field. The latter is the case when
\begin{align}\label{introGR}
\V{G}=(G_1,\ldots,G_\nu)\in L_\loc^2(\Geb,\HP_{\RR}^\nu),
\end{align}
where $\HP_{\RR}$ is an arbitrary completely real subspace of $\HP$ satisfying
$e^{-t\omega}\HP_{\RR}\subset\HP_{\RR}$, for all $t>0$. As it turns out, it is possible under the 
conditions \eqref{Vlocint}, \eqref{introA}, and \eqref{introGR} to derive Feynman-Kac formulas for 
Dirichlet-Pauli-Fierz operators given by familiar expressions, provided that the Stratonovich integrals 
involving $\V{A}$ and $\V{G}$ in these formulas are defined as in \eqref{defSxintro} and
\eqref{defKxintro} below. 

%%%%%%%%%%%%%%%%%%%%%%%%%%%%%%%%%%%%%%%%%%%%%%%%%

\subsection{The main theorem}\label{ssecmainthm}

\noindent
In the whole article
$(\Omega,\mathfrak{F},(\mathfrak{F}_t)_{t\ge0},\PP)$
denotes a filtered probability space satisfying the usual assumptions of completeness and right-continuity 
of the filtration $(\mathfrak{F}_t)_{t\ge0}$. The letter $\EE$ denotes expectation with respect to 
$\PP$. Furthermore,  $\V{B}$ denotes a $\nu$-dimensional $(\fF_t)_{t\ge0}$-Brownian motion (starting in $0$) and 
we put $\V{B}^{\V{x}}:=\V{x}+\V{B}$, for all $\V{x}\in\RR^\nu$. Pick some $t>0$ and let 
\begin{align}\label{BMrev}
\V{B}^{t;\V{x}}:=(\V{B}_{t-s}^{\V{x}})_{s\in[0,t]}
\end{align}
denote the time-reversal of $\V{B}^{\V{x}}$ at $t$. This time-reversed process
is a semimartingale when the underlying probability space is equipped
with a suitable new filtration as explained in more detail in Subsect.~\ref{ssecBMBB}; see
\cite{HaussmannPardoux1986,PardouxLNM1204} for the general theory of time-reversed diffusion
processes. With this we define\footnote{Readers who are wondering about the signs in \eqref{defSxintro} should notice that the complex conjugate of $S_t(\V{x})$ appears in our Feynman-Kac
formula; see \eqref{defWWintroscal} and the first equality in \eqref{FKintroscal}.}
\begin{align}\label{defSxintro}
S_t(\V{x})&:=\int_0^t V(\V{B}_s^{\V{x}})\Id s
-\frac{i}{2}\int_0^t\V{A}(\V{B}_s^{\V{x}})\Id\V{B}_s^{\V{x}}
+\frac{i}{2}\int_0^t\V{A}(\V{B}^{t;\V{x}}_s)\Id\V{B}_s^{t;\V{x}},
\\\label{defKxintro}
K_t(\V{x})&:=\frac{1}{2}\int_0^tj_s\V{G}_{\V{B}_s^{\V{x}}}\Id\V{B}_s^{\V{x}}-
\frac{1}{2}\int_0^tj_{t-s}\V{G}_{\V{B}^{t;\V{x}}_s}\Id\V{B}_s^{t;\V{x}}.
\end{align}
In the second line, $\{j_s\}_{s\in\RR}$ is a strongly continuous family of isometries originally introduced 
by E.~Nelson \cite{Nelson1973}. These isometries are defined on $\HP$ 
and attain values in the new Hilbert space 
\begin{align*}
\hat{\HP}:=L^2(\RR\times\cK,\mathfrak{B}(\RR)\otimes\mathfrak{K},\lambda\otimes\mu),
\end{align*} 
with $\lambda$ denoting the one-dimensional Lebesgue measure and
$\fB(\RR)$ the Borel subsets of $\RR$. They are given by the formulas
\begin{align*}
(j_sf)(\kappa,k)&:=\frac{1}{\pi^\eh}\frac{\omega(k)^\eh}{(\kappa^2+\omega(k)^2)^\eh}
e^{-is\kappa}f(k),\quad\text{a.e. $(\kappa,k)\in\RR\times\cK$},
\end{align*}
for all $f\in\HP$ and $s\in\RR$. We apply $j_s$ componentwise to an element of $\HP^\nu$.
The construction of the four stochastic integrals above under the
conditions \eqref{introA} and \eqref{introGR} requires a few simple comments which are given in 
Subsects.~\ref{ssecexpathint} and~\ref{ssecFKIsingAG}; 
their existence is guaranteed for a.e. $\V{x}\in\Geb$ at least. Notice that 
the first and second stochastic integrals in both \eqref{defSxintro} 
and \eqref{defKxintro} are defined with respect to different filtrations; in each line the linear
combination of the two It\^{o} type integrals substitutes more common expressions for
Stratonovich integrals.

Next, let $\V{b}^{t;\V{y},\V{x}}$ be the semimartingale realization of a Brownian bridge
from $\V{y}\in\RR^\nu$ to $\V{x}\in\RR^\nu$ in time $t$ introduced in more detail in
Subsect.~\ref{ssecBMBB}. As verified in \cite[App.~4]{GMM2017}, the relevant results of
\cite{HaussmannPardoux1986,PardouxLNM1204} on time-reversed processes also apply to
Brownian bridges. Putting
\begin{align}\label{BBrev}
\smash{\hat{\V{b}}}^{t;\V{x},\V{y}}:=(\V{b}_{t-s}^{t;\V{y},\V{x}})_{s\in[0,t]},
\end{align} 
we thus obtain a semimartingale realization of a Brownian bridge
from $\V{x}\in\RR^\nu$ to $\V{y}\in\RR^\nu$ in time $t$, provided that the original filtration is
replaced by a suitable new one; see again Subsect.~\ref{ssecFKIsingAG}.
Analogously to \eqref{defSxintro} and \eqref{defKxintro} we define
\begin{align}\nonumber
S_t(\V{x},\V{y})&:=\int_0^t V(\V{b}_s^{t;\V{y},\V{x}})\Id s
\\\label{defSxyintro}
&\quad-\frac{i}{2}\int_0^t\V{A}(\V{b}_s^{t;\V{y},\V{x}})\Id\V{b}_s^{t;\V{y},\V{x}}
+\frac{i}{2}\int_0^t\V{A}(\smash{\hat{\V{b}}}^{t;\V{x},\V{y}}_s)
\Id\smash{\hat{\V{b}}}_s^{t;\V{x},\V{y}},
\\\label{defKxyintro}
K_t(\V{x},\V{y})&:=\frac{1}{2}\int_0^tj_s\V{G}_{\V{b}_s^{t;\V{y},\V{x}}}\Id\V{b}_s^{t;\V{y},\V{x}}-
\frac{1}{2}\int_0^tj_{t-s}\V{G}_{\smash{\hat{\V{b}}}^{t;\V{x},\V{y}}_s}\Id
\smash{\hat{\V{b}}}_s^{t;\V{x},\V{y}}.
\end{align}
Again the existence of the four stochastic integrals appearing here is ensured by \eqref{introA} and 
\eqref{introGR}, for a.e. $(\V{x},\V{y})\in\RR^{2\nu}$ at least; see 
Subsects.~\ref{ssecexpathint} and~\ref{ssecFKIsingAG}. 

We finally list all remaining notation needed to formulate our main theorem:
\begin{enumerate}[leftmargin=0.8cm]
\item[$\triangleright$] 
We abbreviate
\begin{align}\label{defWWintroscal}
W_t(\V{x})&:=e^{-{S}_t(\V{x})}\Gamma(j_t)^*e^{i\vp(K_t(\V{x}))}\Gamma(j_0),
\\\label{introdefWtxy}
W_t(\V{x},\V{y})&:=e^{-{S}_t(\V{x},\V{y})}\Gamma(j_t)^*e^{i\vp(K_t(\V{x},\V{y}))}\Gamma(j_0),
\end{align}
where $\Gamma(j_s)$ denotes the second quantization of the isometry $j_s$. 
\item[$\triangleright$] 
The first exit time of $\V{B}^{\V{x}}$ from $\Geb$ is denoted by
 \begin{align*}
\tau_{\Geb}(\V{x})&:=\inf\{s\ge0|\,\V{B}_s^{\V{x}}\notin\Geb\}.
\end{align*} 
We always employ the common convention $\inf\emptyset:=\infty$.
\item[$\triangleright$] 
The first exit time of $\V{b}^{t;\V{y},\V{x}}$ from $\Geb$ is denoted by
\begin{align*}
\tau_{\Geb}(t;\V{y},\V{x})&:=\inf\big\{s\in[0,t]\big|\,\V{b}_s^{t;\V{y},\V{x}}\notin\Geb\}.
\end{align*}
\item[$\triangleright$]
The symbol $1_\cA$ stands for the indicator function of a set $\cA$.
\item[$\triangleright$]
We denote the Euclidean heat kernel by
\begin{align}\label{Gausskern}
p_t(\V{x},\V{y})&:=(2\pi t)^{-\nf{\nu}{2}}e^{-|\V{x}-\V{y}|^2/2t},\quad \V{x},\V{y}\in\RR^\nu,\,t>0.
\end{align}
\end{enumerate}

\begin{thm}\label{thmFKintro1}
Assume \eqref{Vlocint}, \eqref{introA}, and \eqref{introGR}. Let $t>0$ and $\Psi\in L^2(\Geb,\sF)$.
Then we have the following Feynman-Kac formulas for the Dirichlet-Pauli-Fierz operator 
$H_\Geb$ representing the closure of the form given by \eqref{introPFform} and \eqref{dommrh},
\begin{align}\nonumber
(e^{-tH_\Geb}\Psi)(\V{x})&=\EE\Big[1_{\{\tau_\Geb(\V{x})>t\}}W_t(\V{x})^*\Psi(\V{B}_t^{\V{x}})\Big]
\\\label{FKintroscal}
&=\int_\Geb p_{t}(\V{x},\V{y})\EE\Big[1_{\{\tau_\Geb(t;\V{y},\V{x})=\infty\}}
W_t(\V{x},\V{y})\Psi(\V{y})\Big]\Id\V{y},\quad\text{a.e. $\V{x}\in\Geb$.} 
\end{align}
\end{thm}

\begin{proof}
This theorem is proven in Subsect.~\ref{ssecFKsingAGGeb}.
\end{proof}

\begin{rem}
Manifestly, $W_t(\V{x})^*$ and $W_t(\V{x},\V{y})$ 
are strongly measurable maps from $\Omega$ to $\LO(\sF)$. Furthermore,
\begin{align}\label{WWWbd}
\|W_t(\V{x})\|\le1,\quad\|W_t(\V{x},\V{y})\|\le1,
\end{align}
pointwise on $\Omega$. In particular, the $\sF$-valued
expectations in \eqref{FKintroscal} are well-defined.
\end{rem}

\begin{rem}\label{remkernig}
Write $\fdom(\omega^{-1})_{\RR}:=\fdom(\omega^{-1})\cap\mathfrak{k}_{\RR}$ and replace 
\eqref{introGR} by the stronger condition
\begin{align}\label{introGRomega}
\V{G}\in L_\loc^2(\Geb,\fdom(\omega^{-1})_{\RR}^\nu),
\end{align}
which is typically fulfilled in physically relevant examples with ultraviolet regularized interaction terms.
Pick some $t>0$ and $\V{x},\V{y}\in\RR^\nu$ such that all integrals in \eqref{defSxyintro} 
and \eqref{defKxyintro} exist.
According to \cite[Rem.~17.7]{GMM2017} we then have the alternative formula
\begin{align}\label{GMMformel}
W_t(\V{x},\V{y})&=e^{-S_t(\V{x},\V{y})-\|K_t(\V{x},\V{y})\|_{\HP}^2/2}
F_{\nf{t}{2}}(ij_t^*K_t(\V{x},\V{y}))F_{\nf{t}{2}}(-ij_0^*K_t(\V{x},\V{y}))^*,
\end{align}
where the Fock space operator-valued maps
\begin{align*}
\fdom(\omega^{-1})\ni g\longmapsto F_{\nf{t}{2}}(g):=\sum_{n=0}^\infty
\ad(g)^ne^{-t\Id\Gamma(\omega)/2}\in\LO(\sF),
\end{align*}
are analytic \cite[Lem.~17.4]{GMM2017}, thus separably valued as $\fdom(\omega^{-1})$ is separable. 
(Here $\ad(g)$ is the bosonic creation operator in $\sF$ associated with
$g$; see, e.g., \cite{Parthasarathy1992}.) In particular, $W_t(\V{x},\V{y}):\Omega\to\LO(\sF)$ is
measurable, separably valued, and bounded, whence the {\em $\LO(\sF)$-valued} expectation in
\begin{align}\label{kernig}
e^{-tH_\Geb}(\V{x},\V{y})&:=
p_t(\V{x},\V{y})\EE\big[1_{\{\tau_\Geb(t;\V{y},\V{x})=\infty\}}W_t(\V{x},\V{y})\big]\in\LO(\sF)
\end{align}
is well-defined. In view of \eqref{FKintroscal}, the operators in \eqref{kernig} thus define a
$\LO(\sF)$-valued integral kernel of $e^{-tH_{\Geb}}$.
The random function $W_t(\V{x})$ can be written in the form \eqref{GMMformel} as well, 
provided that we drop $\V{y}$ on the right hand side, of course.
\end{rem}

In the following corollary we subtract a negative part $U$ from $V$. The form
$\mr{\mathfrak{h}}_{\Geb}^U$ appearing in its statement is defined on 
$\dom(\mr{\mathfrak{h}}_\Geb)$ and  obtained upon putting $V-U$ in place
of $V$ in \eqref{introPFform}.

\begin{cor}\label{cornegpart}
Assume \eqref{Vlocint}, \eqref{introA}, \eqref{introGR}, and
let $U:\Geb\to[0,\infty)$ be form bounded with respect to one-half times the Dirichlet-Laplacian on 
$\Geb$ with relative form bound $b\le1$. Then $U\id_{\sF}$ is form bounded with respect to
$H_\Geb$ with relative form bound $\le b$ and, in particular, $\mr{\mathfrak{h}}_{\Geb}^U$
is semibounded. Assume in addition that $\mr{\mathfrak{h}}_{\Geb}^U$ is closable and denote the
selfadjoint operator representing its closure by $H_\Geb^U$. Then \eqref{FKintroscal}
remains true, when $H_\Geb$ is replaced
by $H_\Geb^U$ and $V-U$ is put in place of $V$ in \eqref{defSxintro} and \eqref{defSxyintro}. 
If \eqref{introGRomega} is satisfied, then Rem.~\ref{remkernig} is still valid under the same
replacements.
\end{cor}

Notice that the somewhat implicit assumption that $\mr{\mathfrak{h}}_{\Geb}^U$ be closable is
satisfied when $b<1$. It is also satisfied when $b\le1$, 
$\dom(\mr{\mathfrak{h}}_{\Geb})\subset\dom(H_\Geb)$, and $U$ is locally
{\em square}-integrable, in which case $H_\Geb^U$ is a Friedrichs extension.

In Schr\"{o}dinger operator theory even more singular $U$ than the ones considered here
have been treated; see \cite{BroderixLeschkeMueller2004,LiskevichManavi1997,Voigt1986} 
and the references given therein.

\begin{proof}
Cor.~\ref{cornegpart} is proven at the end of Subsect.~\ref{ssecFKsingAGGeb}.
\end{proof}

Our Feynman-Kac formulas have several immediate and by now well-known applications that we shall
mention only very briefly:

\begin{rem}\label{remappl} 
Assume \eqref{Vlocint}, \eqref{introA}, and \eqref{introGRomega}.

Adopting the notion of positivity on $\sF$ induced by its $\cQ$-space representation, we find that
the semigroup of $H_\Geb^U$ with $U$ as in Corollary~\ref{cornegpart} is ergodic;
compare \cite[\textsection10]{Matte2016}, \cite[\textsection8.1]{MatteMoeller2018}, and the
references therein. If $U$ has an extension to $\RR^\nu$ that belongs to the Kato class of $\RR^\nu$, 
then we obtain $L^p(\Geb,\sF)$ to 
$L^q(\Geb,\sF)$ estimates (with $1\le p\le q\le\infty$) for the semigroup of $H_\Geb^U$ and
Gaussian upper bounds on its operator-valued integral kernel; see \cite{Matte2016} for references and
further extensions in the case $\Geb=\RR^\nu$ with regular coefficients. 
If $U$ is Kato in the above sense and $\omega$ has a strictly positive lower bound,
then the semigroup is hypercontractive simultaneously in the $\V{x}$- and $\cQ$-space-variables; 
see \cite[Thm.~1.9 and \textsection3.1]{HiroshimaMatte2019} for an analogous bound in
the renormalized Nelson model. If the latter hypercontractivity bound is available and $\Geb$ is
bounded and connected, then the infimum of the spectrum of $H_\Geb^U$ is a
non-degenerate eigenvalue and the corresponding eigenvector can be chosen strictly positive; 
see again \cite[\textsection3.1]{HiroshimaMatte2019}.
\end{rem}

%%%%%%%%%%%%%%%%%%%%%%%%%%%%%%%%%%%%%%%%%%%%%%%%

\subsection{Brief remarks on earlier results}

\noindent
For $\Geb=\RR^\nu$, $\V{A}=0$, and under stronger assumptions on $\V{G}$,
the first identity in \eqref{FKintroscal} has been proven earlier by F.~Hiroshima
\cite{Hiroshima1997}, and the second equality in \eqref{FKintroscal} has been shown in 
\cite{GMM2017}. The idea to represent Feynman-Kac integrands in nonrelativistic
quantum field theory in the form \eqref{defWWintroscal} is originally due to E.~Nelson 
\cite{Nelson1973}, who considered scalar matter particles that are linearly coupled to quantized 
radiations fields. 

In \cite{GMM2017} and in \cite{HiroshimaLorinczi2008} different possibilities to account for spin
degrees of freedom in Feynman-Kac formulas for the Pauli-Fierz model are considered. 
An extension of Theorem~\ref{thmFKintro1} to a situation where the matter particles may have spin would, 
however, by no means be trivial and require extra conditions on the magnetic fields generated by 
the classical and quantized vector potentials.

As any meaningful survey of the extensive literature on Feynman-Kac formulas for magnetic 
Schr\"{o}dinger operators and their various generalizations and applications would go beyond the scope 
of the discussion, we kindly ask the interested reader to consult, e.g., the remarks and long reference 
lists in the relatively recent article \cite{Hinz2015} and the books \cite{Gueneysu2017,LHB2011} 
for a start. Explicitly, we mention only a few articles dealing 
with possibly very singular classical vector potentials on open subsets of the Euclidean space: 

In \cite{BHL2000} local Kato class assumptions are 
imposed on $\V{A}^2$ and $\Div\V{A}$ to derive Feynman-Kac formulas. The most singular case
where quadratic forms still make sense on $\sD(\Geb)$, that is,
$\V{A}\in L_\loc^2(\Geb,\RR^\nu)$, is treated in \cite{PerelmuterSemenov1981} in the special case
where $\Geb^c$ has zero Lebesgue measure. Since the Feynman-Kac integrands are constructed with 
the help of compactness arguments in \cite{PerelmuterSemenov1981}, they are, however, not given by 
explicit formulas there. 

For every $\V{A}\in L_\loc^2(\Geb,\RR^\nu)$, we actually find some
$\V{A}_{{\textsc{c}}}\in L_\loc^2(\Geb,\RR^\nu)$, 
having the same curl in distribution sense as $\V{A}$
and satisfying the Coulomb gauge condition $\Div\V{A}_{{\textsc{c}}}=0$ in the weak sense,
as well as some gauge potential $\gamma\in W^{1,2}_\loc(\Geb)$ such that
$\V{A}=\V{A}_{{\textsc{c}}}+\nabla\gamma$; see \cite[Lem.~1.1]{Leinfelder1983}. Exploiting the gauge
invariance of Schr\"{o}dinger operators \cite[(Proof of) Thm.~1.2]{Leinfelder1983}, we can thus
derive a Feynman-Kac formula for the Schr\"{o}dinger operator with vector potential 
$\V{A}_{{\textsc{c}}}$
containing only one stochastic integral in It\^{o}'s sense, and obtain a Feynman-Kac type formula for 
$\V{A}$ by adding a $\gamma$-dependent term to the complex action. This strategy to 
find Feynman-Kac formulas for Dirichlet realizations of Schr\"{o}dinger operators with highly
singular vector potentials is treated as 
well-known in the more recent literature at least in the case where $\V{A}$ has a locally 
square-integrable extension to the whole $\RR^\nu$ (see, e.g., \cite{HKNSV2006}), 
and probably also in greater generality.

%%%%%%%%%%%%%%%%%%%%%%%%%%%%%%%%%%%%%%%%%%%%%%%%

\subsection{Organization, proof strategies, and further results}

\begin{enumerate}[leftmargin=*]
\item[$\triangleright$]
In Sect.~\ref{secbasicdefns} we recall some Fock space calculus and provide
precise definitions of the most important quadratic forms and operators considered in this article.
\item[$\triangleright$] 
Our general strategy is to infer Feynman-Kac formulas for proper open subsets 
$\Geb\subset\RR^\nu$ from corresponding formulas in the case $\Geb=\RR^\nu$. To this end
we employ a procedure originally used for Schr\"{o}dinger semigroups in \cite{Simon1978Adv} and 
later on for 
magnetic Schr\"{o}dinger semigroups in \cite{BHL2000}. In Sect.~\ref{secFKforD} we recall this 
procedure in a suitably abstracted version that applies to the quantum field theoretic models we are 
interested in here and in the recent work \cite{HiroshimaMatte2019}. 
\item[$\triangleright$] A crucial ingredient for the proof procedure alluded to in the 
previous item are results on approximations with respect to the form norms of certain
{\em maximal} Pauli-Fierz forms. (The closure of the form defined in \eqref{introPFform} and
\eqref{dommrh} is the {\em minimal} Pauli-Fierz form.)
These approximation results, which are non-trivial and possibly of independent interest,
 are obtained in Sect.~\ref{secapprox}.
A Leibniz rule for vector-valued weak derivatives needed here is derived first in Sect.~\ref{secLeibniz}.
As a byproduct we shall also see that the maximal and minimal Pauli-Fierz forms agree when
$\Geb=\RR^\nu$, as it is the case for Schr\"{o}dinger operators \cite{SimonJOT1979}.
\item[$\triangleright$] Also in the case $\Geb=\RR^\nu$ our Feynman-Kac formulas 
are obtained by approximation. Here we depart from Feynman-Kac formulas
for Pauli-Fierz operators with regularized coefficients. In Sect.~\ref{secstrresconv} we therefore
study strong resolvent convergence of Pauli-Fierz operators on $\RR^\nu$ when $\V{A}$ and $\V{G}$ 
are approximated in $L^2_\loc$ by more regular quantities. In doing so we employ a diamagnetic
inequality for resolvents of Pauli-Fierz operators that we derive first in Sect.~\ref{secmoredia}, 
more generally for Dirichlet-Pauli-Fierz operators on general open $\Geb\subset\RR^\nu$.
In its full generality this diamagnetic inequality is new even when $\Geb=\RR^\nu$.
\item[$\triangleright$] For regular coefficients and $\Geb=\RR^\nu$, we derive our Feynman-Kac
formulas in Sect.~\ref{secstochana}, employing the stochastic differential
equations associated with the Pauli-Fierz model analyzed in \cite{GMM2017}.
We shall push forward some results of \cite{GMM2017} to non-vanishing $\V{A}$.
Eventually, we  prove an associated strong Markov property (employing a ``useful rule'' for vector-valued 
conditional expectations verified in App.~\ref{appusefulrule}) and show that the ``probabilistic'' right 
hand sides of the Feynman-Kac formulas give rise to a strongly continuous semigroup of bounded
selfadjoint operators. The Pauli-Fierz operator finally turns out to be the generator of this semigroup,
which proves the Feynman-Kac formulas for regular coefficients.
\item[$\triangleright$] The only technical obstacle remaining after the above preliminary 
results is to show convergence of the probabilistic sides of the Feynman-Kac formulas
for $\Geb=\RR^\nu$, when singular coefficients are approximated by regular ones.
This is done in Sect.~\ref{secFKsing}. Apart from that, we give a detailed discussion
of the Feynman-Kac integrands for singular coefficients and eventually complete the  
proofs of Thm.~\ref{thmFKintro1} and Cor.~\ref{cornegpart} in this final section.
\end{enumerate}

%%%%%%%%%%%%%%%%%%%%%%%%%%%%%%%%%%%%%%%%%%%%%%%%%
%%%%%%%%%%%%%%%%%%%%%%%%%%%%%%%%%%%%%%%%%%%%%%%%%
%%%%%%%%%%%%%%%%%%%%%%%%%%%%%%%%%%%%%%%%%%%%%%%%%

\section{Basic definitions}\label{secbasicdefns}

\noindent
In this section we collect the most important functional analytic definitions employed throughout the
article. In the following subsections we shall, respectively, recall some Fock space calculus,
define vector-valued weak derivatives, covariant derivatives, and finally introduce our 
Dirichlet-Pauli-Fierz operators.

In the whole article $\Geb$ denotes an arbitrary open subset of $\RR^\nu$; variables in
$\Geb$ will most of the time be denoted by $\V{x}=(x_1,\ldots,x_\nu)$ or
$\V{y}=(y_1,\ldots,y_\nu)$.

If $T$ is a linear operator in some Hilbert space then its domain $\dom(T)$ is equipped with the
graph norm
\begin{align*}
\|\phi\|_{\dom(T)}&:=(\|\phi\|^2+\|T\phi\|^2)^\eh,\quad \phi\in\dom(T).
\end{align*}
If $T$ is nonnegative and selfadjoint, then its form domain $\fdom(T)$ is equipped with the form norm
\begin{align*}
\|\phi\|_{\fdom(T)}&:=(\|\phi\|^2+\|T^\eh\phi\|^2)^\eh,\quad \phi\in\fdom(T).
\end{align*}

%%%%%%%%%%%%%%%%%%%%%%%%%%%%%%%%%%%%%%%%%%%%%%%

\subsection{Operators in the bosonic Fock space}\label{ssecFock}

\noindent
Here we briefly recall some standard facts on the Weyl representation on bosonic Fock space. 
For a detailed textbook exposition of these matters we recommend \cite{Parthasarathy1992}.

Recall that the by assumption separable $L^2$-space $\HP$ has been introduced in \eqref{defHP}.
The bosonic Fock space modelled over $\HP$ is given by the direct orthogonal sum
\begin{align*}
\sF&:=\CC\oplus\bigoplus_{n=1}^\infty L^2(\cK^n,\fK^n,\mu^n),
\end{align*}
where $\fK^n$ is the $n$-fold product of the $\sigma$-algebra $\fK$ with itself and
$\mu^n$ is the $n$-fold product of the measure $\mu$ with itself. A total subset of $\sF$ is
given by the set of exponential vectors $\epsilon(f)\in\sF$ with $f\in\sF$, 
\begin{align*}
\epsilon(f)&:=(1,f,\ldots,(n!)^{\mh}f^{\otimes_n},\ldots\:),
\end{align*}
with $f^{\otimes_n}(k_1,\ldots,k_n):=f(k_1)\dots f(k_n)$, $\mu^n$-a.e. Let $\sU(\sK)$ denote the
set of unitary operators on some Hilbert space $\sK$ equipped with the topology associated with the
strong convergence of bounded operators on $\sK$. Given $f\in\HP$ and $U\in\sU(\HP)$, we let
$\sW(f,U)\in\sU(\sF)$ denote the corresponding Weyl operator. We recall that it is determined
by the prescription
\begin{align*}
\sW(f,U)\epsilon(g)&:=e^{-\|f\|^2/2-\SPn{f}{Ug}}\epsilon(f+Ug),\quad g\in\HP,
\end{align*}
followed by linear and isometric extensions. The so obtained Weyl representation
\begin{align*}
\sW:\HP\times\sU(\HP)\longrightarrow\sU(\sF),\quad(f,U)\longmapsto\sW(f,U),
\end{align*}
is a strongly continuous projective representation of the semi-direct product of $\HP$ and $\sU(\HP)$.
More precisely, we have the Weyl relations
\begin{align*}
\sW(f_1,U_1)\sW(f_2,U_2)&=e^{-i\Im\SPn{f_1}{U_1f_2}}\sW(f_1+U_1f_2,U_1U_2),
\end{align*}
for all $ f_1,f_2\in\HP$ and $U_1,U_2\in\sU(\HP)$. As usual we abbreviate
\begin{align}\label{defGammasW}
\sW(f)&:=\sW(f,\id),\quad\Gamma(U):=\sW(0,U),\quad f\in\HP,\,U\in\sU(\HP).
\end{align}

Let $f\in\HP$. Then the above remarks imply that $\RR\ni t\mapsto\sW(-itf)$ is a strongly
continuous unitary group on $\sF$. Its selfadjoint generator is called the field operator associated
with $f$. It is denoted by $\vp(f)$, so that
\begin{align*}
\sW(-itf)=e^{-it\vp(f)},\quad t\in\RR.
\end{align*}

In the whole article, 
\begin{align*}
\text{$\omega:\cK\to\RR$ is a measurable function that is $\mu$-a.e. strictly positive.}
\end{align*}
 It has the physical interpretation of a boson dispersion relation. 
 We shall use the same symbol $\omega$ to denote the associated selfadjoint
multiplication operator in $\HP$. Then our remarks on the Weyl representation further imply that
$\RR\ni t\mapsto\Gamma(e^{-it\omega})$ is a strongly continuous unitary group on $\sF$.
Therefore, there exists a selfadjoint operator $\Id\Gamma(\omega)$ in $\sF$ such that
\begin{align*}
\Gamma(e^{-it\omega})=e^{-it\Id\Gamma(\omega)},\quad t\in\RR.
\end{align*}
It is called the differential second quantization of $\omega$ and
interpreted as the energy of the quantized radiation field.

Since the Nelson isometries $j_s:\HP\to\hat{\HP}$ introduced in Subsect.~\ref{ssecmainthm}
map into a Hilbert space different from $\HP$, the symbol $\Gamma(j_s)$ actually has to be
understood in a sense generalizing \eqref{defGammasW}. In fact, $\Gamma(j_s):\sF\to\hat{\sF}$ 
is obtained by linear and isometric extension of the prescription 
$\Gamma(j_s)\epsilon(g):=\epsilon(j_sg)\in\hat{\sF}$,
$g\in\HP$, where $\hat{\sF}$ is the bosonic Fock space modelled over $\hat{\HP}$.

We conclude this subsection by recalling the following standard relative bounds,
where $\vk:\cK\to\RR$ has the same properties as $\omega$ above,
\begin{align}\label{rbvp}
\|\vp(f)\psi\|_{\sF}&\le2^\eh\|f\|_{\fdom(\vk^{-1})}\|\psi\|_{\fdom(\Id\Gamma(\vk))},
\\\label{rbvpvp}
\|\vp(f)\vp(g)\phi\|_{\sF}&\le8\|f\|_{\fdom(\vk^{-1})}\|g\|_{\fdom(\vk^{-1})}
\|\phi\|_{\dom(\Id\Gamma(\vk))},
\end{align}
for all $f,g,\phi,\psi$ in the vectors spaces indicated by the respective subscripts; see, e.g.,
\cite[Rem.~2.10]{Matte2017} for the second bound.

%%%%%%%%%%%%%%%%%%%%%%%%%%%%%%%%%%%%%%%%%%%%%%%

\subsection{Vector-valued weak derivatives}

\noindent
To deal with singular classical and quantized vector potentials it is most helpful to mimic
the distributional techniques used in the study of magnetic Schr\"{o}dinger operators in a
vector-valued setting \cite{Matte2017}. For the convenience of the reader we therefore recall
the following fundamental definition:

Let $\sK$ be a separable Hilbert space, $j\in\{1,\ldots,\nu\}$, and
$\Psi\in L^1_\loc(\Geb,\sK)$. Then $\Psi$ is said to have a
weak partial derivative with respect to $x_j$, iff there exists some (necessarily unique)
vector $\partial_{x_j}\Psi\in L_\loc^1(\Geb,\sK)$ such that
\begin{align}\label{defwpd}
\int_{\Geb}\SPn{\partial_{x_j}\eta(\V{x})}{\Psi(\V{x})}_\sK\Id\V{x}
=-\int_\Geb\SPn{\eta(\V{x})}{\partial_{x_j}\Psi(\V{x})}_\sK\Id\V{x},\quad \eta\in\sD(\Geb,\sK).
\end{align}

%%%%%%%%%%%%%%%%%%%%%%%%%%%%%%%%%%%%%%%%%%%%%%%

\subsection{Covariant derivatives}

\noindent
Pick $j\in\{1,\ldots,\nu\}$, $A_j\in L^2_\loc(\Geb,\RR)$ and let $G_j:\Geb\to\HP$,
$\V{x}\mapsto G_{j,\V{x}}$ be in $L^2_\loc(\Geb,\HP)$. 
For every $\Psi\in L^2(\Geb,\fdom(\Id\Gamma(1)))$,
we define $\vp(G_j)\Psi\in L^1_\loc(\Geb,\sF)$ by
\begin{align*}
(\vp(G_j)\Psi)(\V{x})&:=\vp(G_{j,\V{x}})\Psi(\V{x}),\quad\text{a.e. $\V{x}\in\Geb$.}
\end{align*}
With this we define a symmetric operator $v_{\Geb,j}$ in $L^2(\Geb,\sF)$ by
\begin{align*}
\dom(v_{\Geb,j})&:=\sD(\Geb,\fdom(\Id\Gamma(1))),
\\
v_{\Geb,j}\Psi&:=-i\partial_{x_j}\Psi-A_j\Psi-\vp(G_j)\Psi,\quad\Psi\in\dom(v_{\Geb,j}).
\end{align*}
Its adjoint $v_{\Geb,j}^*$ will play the role of a covariant derivative in the $j$-th coordinate direction
in our Pauli-Fierz forms.

The approximation results proven in Sect.~\ref{secapprox} depend crucially on
the following theorem \cite[Thm.~3.5]{Matte2017} where, for any separable Hilbert space
$\sK$ and any representative $\Psi(\cdot)$ of $\Psi\in L_\loc^1(\Geb,\sK)$, we define
\begin{align}\label{defsgn}
{\fS_{\Psi}}(\V{x})&:=\left\{\begin{array}{ll}
\|\Psi(\V{x})\|_\sK^{-1}\Psi(\V{x}),&\V{x}\in\{\Psi(\cdot)\not=0\},
\\
0,& \V{x}\in\{\Psi(\cdot)=0\}.
\end{array}\right.
\end{align}

\begin{thm}\label{thm-dia0}
Let $\Psi\in\dom(v_{\Geb,j}^*)$. 
Then $\|\Psi\|_\sF\in L^2(\Geb)$ has a weak partial derivative with respect to $x_j$ which is given by
\begin{align}\label{dia0}
\partial_{{x}_j}\|\Psi\|_{\sF}=\Re\SPn{{\fS_\Psi}}{i{v}_{\Geb,j}^*\Psi}_{\sF}\in L^2(\Geb).
\end{align}
\end{thm}

%%%%%%%%%%%%%%%%%%%%%%%%%%%%%%%%%%%%%%%%%%%%%%%

\subsection{Pauli-Fierz forms and Dirichlet-Pauli-Fierz operators}\label{ssecDPFform}

\noindent
Assuming \eqref{Vlocint}, \eqref{introA}, and \eqref{introG}
we first define a {\em maximal Pauli-Fierz form},
\begin{align}\nonumber
\dom(\mathfrak{h}_{\Geb,\mathrm{N}})&:=L^2(\Geb,\fdom(\Id\Gamma(\omega)))\cap
\fdom(V\id_{\sF})\cap\bigcap_{j=1}^\nu\dom(v_{\Geb,j}^*),
\\\label{adam11}
\mathfrak{h}_{\Geb,\mathrm{N}}[\Psi]&:=
\frac{1}{2}\sum_{j=1}^\nu\|v_{\Geb,j}^*\Psi\|^2+\int_{\Geb}V(\V{x})\|\Psi(\V{x})\|^2_{\sF}\Id\V{x}
\\\nonumber
&\quad+\int_{\Geb}\|\Id\Gamma(\omega)^\eh\Psi(\V{x})\|_{\sF}^2\Id\V{x},
\quad\Psi\in\dom(\mathfrak{h}_{\Geb,\mathrm{N}}).
\end{align}
As a sum of nonnegative closed forms, $\mathfrak{h}_{\Geb,\mathrm{N}}$ is itself closed and
nonnegative. We further define a {\em minimal Pauli-Fierz form},
\begin{align}\label{defminPFf}
\mathfrak{h}_{\Geb,\mathrm{D}}
&:=\ol{\mathfrak{h}_{\Geb,\mathrm{N}}\restr_{\sD(\Geb,\fdom(\Id\Gamma(1\vee\omega)))}}
=\ol{\mr{\mathfrak{h}}_{\Geb}},
\end{align}
where in the second identity we used notation introduced in \eqref{introPFform} and \eqref{dommrh} 
of the introduction. In analogy to the Schr\"{o}dinger case,
the selfadjoint operator representing $\mathfrak{h}_{\Geb,\mathrm{D}}$, we shall simply call it
$H_{\Geb}$ dropping the subscript ``$\mathrm{D}$'', can be interpreted as the Dirichlet
realization of the Pauli-Fierz operator on $\Geb$. (The subscript ``$\mathrm{N}$'' is also borrowed
from the Schr\"{o}dinger theory where it stands for ``Neumann''.)

%%%%%%%%%%%%%%%%%%%%%%%%%%%%%%%%%%%%%%%%%%%%%%%%%
%%%%%%%%%%%%%%%%%%%%%%%%%%%%%%%%%%%%%%%%%%%%%%%%%
%%%%%%%%%%%%%%%%%%%%%%%%%%%%%%%%%%%%%%%%%%%%%%%%%

\section{Deriving Feynman-Kac formulas for Dirichlet realizations}\label{secFKforD}

\noindent
In this section we explain how to derive Feynman-Kac formulas for proper open subsets
$\Geb\subset\RR^\nu$ departing from known formulas in the case $\Geb=\RR^\nu$. 
This is done by a procedure which is standard for Schr\"{o}dinger operators and
originates from \cite{Simon1978Adv}; see also \cite[App.~B]{BHL2000} for a helpful exposition 
treating Schr\"{o}dinger operators with classical magnetic fields. All we do in this section is to 
carry through this procedure in a slightly abstracted setting covering the various nonrelativistic quantum 
field theoretic models we are interested in. The results of this section are, for instance, applied to the 
renormalized Nelson model in \cite{HiroshimaMatte2019}.

Let $\sK\not=\{0\}$ be a separable Hilbert space. 
Suppose that $Q_{\RR^\nu}$ and $Q_\Geb$ are selfadjoint operators in $\HR:=L^2(\RR^\nu,\sK)$ 
and its subspace $\HR_\Geb:=1_\Geb L^2(\RR^\nu,\sK)$, respectively, which are semi-bounded from 
below. Denote the corresponding
quadratic forms by $\mathfrak{q}_{\RR^\nu}$ and $\mathfrak{q}_\Geb$, respectively.
We assume that these two quadratic forms are related as follows:

We pick compact subsets $K_\ell$, $\ell\in\NN$, of $\Geb$ with
$$
K_\ell\subset\mr{K}_{\ell+1},\quad\ell\in\NN,\quad\text{and}\quad
\bigcup_{\ell\in\NN}K_\ell=\Geb.
$$
Furthermore, we pick $\vt_\ell\in C_0^\infty(\RR^\nu)$ with $\vt_\ell=1$ on $K_\ell$,
$\vt_\ell=0$ on $K_{\ell+1}^c$, and $0\le\vt_\ell\le1$, for all $\ell\in\NN$. As in 
\cite{Simon1978Adv} we finally define a numerical function $Y^\Geb_\infty:\RR^\nu\to[0,\infty]$ by
\begin{align}\label{defYsGinfty}
Y_\infty^\Geb(\V{x})
&:=\left\{\begin{array}{ll}
\dist(\V{x},\Geb^c)^{-3}+\sum_{\ell=1}^\infty|\nabla\vt_\ell|^2(\V{x}),&\V{x}\in\Geb,
\\
\infty,&\V{x}\in\Geb^c;
\end{array}\right.
\end{align}
observe that the series appearing here actually has at most one non-vanishing term, for every fixed
$\V{x}\in\Geb$. This function defines a closed form in $\HR$ with domain
$$
\fdom(Y_\infty^\Geb)=\bigg\{\Psi\in L^2(\RR^\nu,\sK)\,\bigg|\:
\int_{\RR^\nu}Y_\infty^\Geb(\V{x})\|\Psi(\V{x})\|^2\Id\V{x}<\infty\bigg\}\subset\HR_{\Geb},
$$
which is not dense in general. We further set 
\begin{align}\label{defD1infty}
\dom(\mathfrak{q}_{\RR^\nu}^{1,\infty})&
:=\dom(\mathfrak{q}_{\RR^\nu})\cap\fdom(Y_\infty^\Geb)\subset\HR_{\Geb}.
\end{align}
We now assume that the following:

\begin{hyp}\label{hypabD}
For at least one function $Y_\infty^\Geb$ defined in the above fashion, 
Statements~(a) and~(b) hold, where:
\begin{enumerate}[leftmargin=0.8cm]
\item[(a)] 
$\dom(\mathfrak{q}_{\RR^\nu}^{1,\infty})\subset\dom(\mathfrak{q}_\Geb)$ and the
closure of $\dom(\mathfrak{q}_{\RR^\nu}^{1,\infty})$ with respect to the form norm on
$\dom(\mathfrak{q}_\Geb)$ is equal to $\dom(\mathfrak{q}_\Geb)$.
\item[(b)] $\mathfrak{q}_\Geb[\Psi]=\mathfrak{q}_{\RR^\nu}[\Psi]$, 
for all $\Psi\in\dom(\mathfrak{q}_{\RR^\nu}^{1,\infty})$.
\end{enumerate}
\end{hyp}

The next remark explains the choice of the power $-3$ in \eqref{defYsGinfty}. Any power strictly less
than $-2$ would actually be sufficient for our applications in the later sections.

\begin{rem}\label{remeindrittel}
Let $t>0$, let $I\subset\RR$ be an interval containing $[0,t]$, 
and suppose that $\gamma:I\to\RR^\nu$ is H\"{o}lder continuous of order $1/3$. Set
\begin{align*}
Y_n^\Geb:=n\wedge Y_\infty^\Geb,\quad n\in\NN.
\end{align*}
Then we have the following equivalence, where the limit to the left always exists in $[0,\infty]$
by monotone convergence,
\begin{align}\label{eindrittelequiv}
\lim_{n\to\infty}\int_0^t Y^\Geb_n(\gamma(s))\Id s<\infty\quad\Leftrightarrow\quad
\inf\big\{s\ge0\big|\,s\in I,\,\gamma(s)\in\Geb^c\big\}>t,
\end{align}
with the common convention $\inf\emptyset=\infty$.

In fact, let $\tau\in[0,\infty]$ denote the infimum in \eqref{eindrittelequiv}.
Assume first that $\tau>t$. Then $\gamma([0,t])\subset\Geb$. Thus
$(Y_\infty^\Geb\circ\gamma)\restr_{[0,t]}$ is a real-valued continuous function
on the compact interval $[0,t]$. It is then clear that the limit as $n\to\infty$ of
the integral to the left in \eqref{eindrittelequiv} is finite.
Next, assume that $\tau\le t$. Since $\gamma$ is continuous and $\Geb^c$ is closed,
we then have $\gamma(\tau)\in\Geb^c$. The H\"{o}lder continuity of $\gamma$ thus implies
$$
Y^\Geb_\infty(\gamma(s))\ge|\gamma(s)-\gamma(\tau)|^{-3}
\ge\frac{1}{C}|s-\tau|^{-1},\quad s\in[0,t],
$$
for some $C>0$. Consequently,
\begin{align*}
\lim_{n\to\infty}\int_0^t Y^\Geb_n(\gamma(s))\Id s\ge\frac{1}{C}\int_0^t\frac{\Id s}{|s-\tau|}=\infty.
\end{align*}
\end{rem}

%%%%%%%%%%%%%%%%%%%%%%%%%%%%%%%%%%%%%%%%%%%%%%%%%

\subsection{Feynman-Kac formulas for Dirichlet realizations}\label{ssecDFK}

\noindent
Throughout this subsection we fix some $t>0$. We work under the assumptions of the preceding 
subsection and the following hypothesis:

\begin{hyp}\label{hypDFKgen}
There exist a probability space $(\Omega,\mathfrak{F},\PP)$ and, for every $\V{x}\in\RR^\nu$,
\begin{enumerate}[leftmargin=0.7cm]
\item[$\triangleright$] a strongly measurable map $M_t(\V{x}):\Omega\to\LO(\sK)$;
\item[$\triangleright$] some pathwise continuous $\RR^\nu$-valued stochastic process 
$\V{X}^{\V{x}}$ which $\PP$-a.s. starts at $0$ and whose paths are 
$\PP$-a.s. H\"older continuous of order $1/3$; 
\end{enumerate} 
such that the following holds:
\begin{enumerate}[leftmargin=0.7cm]
\item[$\triangleright$] For all $\Psi\in\HR$ and $\V{x}\in\RR^\nu$,
\begin{align}\label{domMt}
\|M_t(\V{x})\|\|\Psi(\V{X}^{\V{x}}_t)\|\in L^1(\PP).
\end{align}
\item[$\triangleright$] For all bounded and continuous functions $v:\RR^\nu\to\RR$,
the following Feynman-Kac type formula holds for all $\Psi\in\HR$,
\begin{align}\label{FKabsRRd}
(e^{-t(Q_{\RR^\nu}+v)}\Psi)(\V{x})=\EE\big[e^{-\int_0^tv(\V{X}^{\V{x}}_s)\Id s}
M_t(\V{x})\Psi(\V{X}^{\V{x}}_t)\big],\quad\text{a.e. $\V{x}\in\RR^\nu$.}
\end{align}
\end{enumerate}
\end{hyp}

We further let $\tau_\Geb(\V{x}):\Omega\to[0,\infty]$ denote the first
exit time of $\V{X}^{\V{x}}$ from $\Geb$, i.e.,
\begin{align*}
\tau_{\Geb}(\V{x})&:=\inf\{s\ge0|\,\V{X}^{\V{x}}_s\in\Geb^c\},
\end{align*}
with $\inf\emptyset=\infty$. Since $\V{X}^{\V{x}}$ is pathwise continuous and $\Geb^c$ is closed, 
$\tau_\Geb(\V{x})$ is a stopping time with respect to the filtration generated by $\V{X}^{\V{x}}$.
In particular, 
$$
\{\tau_\Geb(\V{x})>t\}\in\fF.
$$

\begin{lem}\label{lemFKGeb}
In the situation described above, let $\Psi\in\HR_\Geb$. Then
\begin{align}\label{FKabs}
(e^{-tQ_\Geb}\Psi)(\V{x})&=\EE\big[1_{\{\tau_\Geb(\V{x})>t\}}
M_t(\V{x})\Psi(\V{X}^{\V{x}}_t)\big],\quad\text{a.e. $\V{x}\in\RR^\nu$.}
\end{align}
\end{lem}

\begin{proof}
Before we comment on the various steps of this proof we have to introduce some more notation:

For every $\kappa>0$, we define
$\dom(\mathfrak{q}_{\RR^\nu}^{\kappa,\infty}):=\dom(\mathfrak{q}_{\RR^\nu}^{1,\infty})$ 
(recall \eqref{defD1infty}) and
\begin{align*}
\mathfrak{q}_{\RR^\nu}^{\kappa,\infty}[\Psi]
:=\mathfrak{q}_{\RR^\nu}[\Psi]+\kappa\int_{\RR^\nu}Y_\infty^\Geb(\V{x})\|\Psi(\V{x})\|^2\Id\V{x},
\quad\Psi\in\dom(\mathfrak{q}_{\RR^\nu}^{\kappa,\infty}).
\end{align*}
Then $\mathfrak{q}_{\RR^\nu}^{\kappa,\infty}$ is closed as a sum of closed semi-bounded forms.
As remarked above, it is in general not densely defined as a form in $\HR$. By assumption
(a) it is, however, a densely defined semi-bounded closed form on the sub-Hilbert space $\HR_{\Geb}$.
Therefore, there exists a unique selfadjoint operator in $\HR_{\Geb}$ representing
$\dom(\mathfrak{q}_{\RR^\nu}^{\kappa,\infty})$ that we denote by $Q^{\kappa,\infty}_{\RR^\nu}$. 
We further define the Hamiltonians 
$$
Q_{\RR^\nu}^{\kappa,n}:=Q_{\RR^\nu}+\kappa Y_n^\Geb,\quad n\in\NN,\,\kappa>0,
$$
and denote the associated quadratic forms by $\mathfrak{q}_{\RR^\nu}^{\kappa,n}$.

\smallskip

{\em Step 1.} Let $\kappa>0$. We shall show that
\begin{align}\label{gustav1}
\big\|e^{-tQ_{\RR^\nu}^{\kappa,n}}\Psi-e^{-tQ_{\RR^\nu}^{\kappa,\infty}}1_{\Geb}\Psi\big\|
\xrightarrow{\;\;n\to\infty\;\;}0,\quad\Psi\in\HR.
\end{align}
We know that the form domain of $Q_{\RR^\nu}^{\kappa,n}$ is $\dom(\mathfrak{q}_{\RR^\nu})$,
which contains $\dom(\mathfrak{q}_{\RR^\nu}^{\kappa,\infty})$.
The monotone convergence theorem further shows that
\begin{align*}
\mathfrak{q}_{\RR^\nu}^{\kappa,n}[\Psi]\uparrow\mathfrak{q}_{\RR^\nu}^{\kappa,\infty}[\Psi],
\quad\Psi\in\dom(\mathfrak{q}_{\RR^\nu}^{\kappa,\infty})=
\Big\{\Phi\in\dom(\mathfrak{q}_{\RR^\nu})\,\Big|\:
\sup_{n\in\NN}\mathfrak{q}_{\RR^\nu}^{\kappa,n}[\Phi]<\infty\Big\}.
\end{align*}
The convergence \eqref{gustav1} now follows from a monotone convergence theorem for not 
necessarily densely defined quadratic forms \cite[Thm.~4.1\&Thm.~4.2]{Simon1978}.

{\em Step 2.} Let $\kappa>0$ and $\Psi\in\HR_{\Geb}$. We next show that
\begin{align}\label{gustav2}
\big(e^{-tQ_{\RR^\nu}^{\kappa,\infty}}\Psi\big)(\V{x})&=\EE\Big[1_{\{\tau_\Geb({\V{x}})>t\}}
e^{-\kappa\int_0^tY_\infty^\Geb(\V{X}_s^{\V{x}})\Id s}M_t(\V{x})\Psi(\V{X}^{\V{x}}_t)\Big],
\end{align}
for a.e. $\V{x}\in\RR^\nu$, where $e^{-\infty}:=0$.
Owing to Step~1 we find natural numbers $n_1<n_2<\ldots$ such that, a.e. on $\RR^\nu$,
the sequence $(e^{-tQ_{\RR^\nu}^{\kappa,n_\ell}}\Psi)_{\ell\in\NN}$ converges to the vector
$e^{-tQ_{\RR^\nu}^{\kappa,\infty}}\Psi$. Furthermore, since the potentials $\kappa Y_n^\Geb$, 
$n\in\NN$, $\kappa>0$, are bounded and continuous, the Feynman-Kac type formula 
\eqref{FKabsRRd} applies to $Q_{\RR^\nu}^{\kappa,n}$. We thus have
\begin{align*}
\big(e^{-tQ_{\RR^\nu}^{\kappa,n}}\Psi\big)(\V{x})&=\EE\Big[
e^{-\kappa\int_0^tY_n^\Geb(\V{X}_s^{\V{x}})\Id s}M_t(\V{x})\Psi(\V{X}^{\V{x}}_t)\Big],
\quad\text{a.e. $\V{x}\in\RR^\nu$},\,n\in\NN,
\end{align*}
as well as the domination
$$
\big\|e^{-\kappa\int_0^tY_n^\Geb(\V{X}^{\V{x}}_s)\Id s}M_t(\V{x})\Psi(\V{X}_t^{\V{x}})\big\|
\le\|M_t(\V{x})\|\|\Psi(\V{X}_t^{\V{x}})\|\in L^1(\PP),\quad n\in\NN.
$$
Therefore, it remains to prove that, for every $\V{x}\in\RR^\nu$,
$$
e^{-\kappa\int_0^tY_n^\Geb(\V{X}^{\V{x}}_s)\Id s}\xrightarrow{\;\;n\to\infty\;\;} 
1_{\{\tau_\Geb({\V{x}})>t\}}
e^{-\kappa\int_0^tY_\infty^\Geb(\V{X}^{\V{x}}_s)\Id s},\quad\text{$\PP$-a.s.}
$$
This follows, however, immediately from Rem.~\ref{remeindrittel} and the postulated 
$\PP$-a.s. $1/3$-H\"{o}lder continuity of $\V{X}^{\V{x}}$.

{\em Step~3.} We now claim that
\begin{align}\label{gustav1000}
\big\|e^{-tQ_{\RR^\nu}^{\kappa,\infty}}\Psi-e^{-tQ_{\Geb}}\Psi\big\|
\xrightarrow{\;\;\kappa\downarrow0\;\;}0,\quad\Psi\in\HR_\Geb.
\end{align}
In fact, our assumption (a) ensures that
$\dom(\mathfrak{q}_{\RR^\nu}^{\kappa,\infty})\subset\dom(\mathfrak{q}_{\Geb})$,
and using (b) we further observe that
\begin{align*}
\mathfrak{q}_{\RR^\nu}^{\kappa',\infty}[\Psi]&\ge\mathfrak{q}_{\RR^\nu}^{\kappa,\infty}[\Psi]
\ge\mathfrak{q}_{\RR^\nu}[\Psi]=\mathfrak{q}_{\Geb}[\Psi],\quad\kappa'\ge\kappa>0,
\\
\mathfrak{q}_{\RR^\nu}^{\kappa,\infty}[\Psi]&\downarrow\mathfrak{q}_{\Geb}[\Psi],\quad
\kappa\downarrow0,
\end{align*} 
for all $\Psi\in\bigcup_{\kappa>0}\dom(\mathfrak{q}_{\RR^\nu}^{\kappa,\infty})
=\dom(\mathfrak{q}_{\RR^\nu}^{1,\infty})$.
Thanks to the density requirement in (a), the convergence \eqref{gustav1000} 
now follows from a monotone 
convergence theorem for quadratic forms \cite[Thm.~VIII.3.11]{Kato}.

{\em Step~4.} Finally, let $\Psi\in\HR_\Geb$. We shall verify \eqref{FKabs}.
By virtue of Step~3 we find $\kappa_n>0$, $n\in\NN$, with $\kappa_n\downarrow0$, $n\to\infty$,
such that, a.e. on $\RR^\nu$, the sequence $(e^{-tQ_{\RR^\nu}^{\kappa_n,\infty}}\Psi)_{n\in\NN}$ 
converges to the left hand side of \eqref{FKabs}. Thanks to Step~2 we further know that
\begin{equation}\label{ariel}
(e^{-tQ_{\RR^\nu}^{\kappa_n,\infty}}\Psi)(\V{x})=\EE\Big[1_{\{\tau_\Geb({\V{x}})>t\}}
e^{-\kappa_n\int_0^tY_{\infty}^\Geb(\V{X}_s^{\V{x}})\Id s}M_t(\V{x})\Psi(\V{X}^{\V{x}}_t)\Big],
\end{equation}
for a.e. $\V{x}\in\RR^\nu$ and all $n\in\NN$. Since we also have the domination
$$
e^{-\int_0^t\kappa_nY_\infty^\Geb(\V{X}_s^{\V{x}})\Id s}
\|M_{t}(\V{x})\Psi(\V{B}_t^{\V{x}})\|
\le\|M_t(\V{x})\|\|\Psi(\V{X}_t^{\V{x}})\|\in L^1(\PP),
\quad n\in\NN,\,\V{x}\in\RR^\nu,
$$
the dominated convergence theorem implies that, for all $\V{x}\in\RR^\nu$, 
the right hand side of \eqref{ariel} converges, as $n\to\infty$, to the right hand side of \eqref{FKabs}.
\end{proof}

%%%%%%%%%%%%%%%%%%%%%%%%%%%%%%%%%%%%%%%%%%%%%%%%%

\subsection{Feynman-Kac formulas for semigroup kernels of Dirichlet realizations}

\noindent
Again we fix $t>0$ and we assume:

\begin{hyp}\label{hypDFKbridgegen}
There exists a probability space $(\Omega,\mathfrak{F},\PP)$ and, for all $\V{x},\V{y}\in\RR^\nu$,
\begin{enumerate}[leftmargin=0.7cm]
\item[$\triangleright$] an operator-valued map $M_t(\V{x},\V{y}):\Omega\to\LO(\sK)$;
\item[$\triangleright$] a continuous $\RR^\nu$-valued stochastic process
$(\V{X}_s^{\V{y},\V{x}})_{s\in[0,t]}$ which $\PP$-a.s. starts at $\V{y}$ and 
whose paths are $\PP$-a.s. H\"older continuous of order $1/3$; 
\end{enumerate} 
such that the following holds:
\begin{enumerate}[leftmargin=0.7cm]
\item[$\triangleright$] For every $\V{x}\in\RR^\nu$, the following map is measurable,
\begin{align*}
[0,t]\times\RR^\nu\times\Omega\ni(s,\V{y},\vo)\longmapsto\V{X}_s^{\V{y},\V{x}}(\vo)\in\RR^\nu.
\end{align*}
\item[$\triangleright$] For every $\V{x}\in\RR^\nu$, the following map is strongly measurable,
\begin{align*}
\RR^\nu\times\Omega\ni(\V{y},\vo)\longmapsto M_t(\V{x},\V{y})(\vo)\in\LO(\sK).
\end{align*}
\item[$\triangleright$] For all $\V{x}\in\RR^\nu$ and $\Psi\in\HR$,
\begin{align}\label{domMtb2}
\int_{\RR^\nu}\EE[\|M_t(\V{x},\V{y})\|]\|\Psi(\V{y})\|\Id\V{y}<\infty
\end{align}
\item[$\triangleright$] For all bounded and continuous functions $v:\RR^\nu\to\RR$, the relation
\begin{align}\label{FKabsRRdxy}
(e^{-t(Q_{\RR^\nu}+v)}\Psi)(\V{x})=\int_{\RR^\nu}\EE\big[e^{-\int_0^tv(\V{X}^{\V{y},\V{x}}_s)\Id s}
M_t(\V{x},\V{y})\Psi(\V{y})\big]\Id\V{y},\quad\text{a.e. $\V{x}\in\RR^\nu$,}
\end{align}
holds for all $\Psi\in\HR$.
\end{enumerate}
\end{hyp}

It might make sense to give the following remark, where $\tau_\Geb(\V{y},\V{x}):\Omega\to[0,\infty]$ 
denotes the first exit time of $\V{X}^{\V{y},\V{x}}$ from $\Geb$, i.e.,
\begin{align*}
\tau_{\Geb}(\V{y},\V{x})&:=\inf\big\{s\in[0,t]\big|\,\V{X}^{\V{y},\V{x}}_s\in\Geb^c\big\}.
\end{align*}

\begin{rem}\label{remtauGebxy}
Let $\V{x}\in\RR^\nu$. Then 
\begin{align}\label{ingrid1999}
\big\{(\V{y},\vo)\in\RR^\nu\times\Omega\big|\,\tau_\Geb(\V{y},\V{x})(\vo)=\infty\big\}
\in\fB(\RR^\nu)\otimes\fF.
\end{align}
In fact, set $\V{Y}_s(\V{y},\vo):=\V{X}_s^{\V{y},\V{x}}(\vo)$, for all $s\in[0,t]$ 
and $(\V{y},\vo)\in\RR^\nu\times\Omega$. Then $(\V{Y}_s)_{s\in[0,t]}$ is a continuous stochastic 
process on $(\RR^\nu\times\Omega,\fB(\RR^\nu)\otimes\fF,\beta\otimes\PP)$, where
$\beta$ is an arbitrary Borel probability measure on $\RR^\nu$. Then its first exit time from
$\Geb$, i.e., $\tilde{\tau}_\Geb:=\inf\{s\in[0,t]|\V{Y}_s\in\Geb^c\}$ is a stopping time with respect
to the filtration generated by $\V{Y}$. In particular,
$\fB(\RR^\nu)\otimes\fF\ni\{\tilde{\tau}_\Geb>t\}=\{\tilde{\tau}_\Geb=\infty\}$ and by inspecting 
definitions we see that $\{\tilde{\tau}_\Geb=\infty\}$ is equal to the set in \eqref{ingrid1999}.
\end{rem}

\begin{lem}\label{lemFKGebb}
In the situation described above, let $\Psi\in\HR_\Geb$. Then
\begin{align}\label{FKabsxy}
(e^{-tQ_\Geb}\Psi)(\V{x})&=\int_{\Geb}\EE\big[1_{\{\tau_\Geb(\V{y},\V{x})=\infty\}}
M_t(\V{x},\V{y})\Psi(\V{y})\big]\Id\V{y},\quad\text{a.e. $\V{x}\in\RR^\nu$.}
\end{align}
\end{lem}

\begin{proof}
The proof parallels the one of Lem.~\ref{lemFKGeb} and we shall again use some notation used there.
Steps~1 and~3, dealing with the left hand sides of the Feynman-Kac formulas, are identical. Therefore,
we only comment on the remaining two steps.

{\em Step~2.} We pick $\kappa>0$ and $\Psi\in\HR$ and propose to show that, 
for a.e. $\V{x}\in\RR^\nu$,
\begin{align}\label{ingrid1}
(e^{-tQ_{\RR^\nu}^{\kappa,\infty}}\Psi)(\V{x})&=\int_{\RR^\nu}
\EE\Big[1_{\{\tau_\Geb(\V{y},\V{x})=\infty\}}
e^{-\kappa\int_0^tY_\infty^\Geb(\V{X}_s^{\V{y},\V{x}})\Id s}M_t(\V{x},\V{y})\Psi(\V{y})\Big]\Id\V{y}.
\end{align}
By assumption the following special cases of \eqref{FKabsRRdxy} hold, for a.e. $\V{x}\in\RR^\nu$,
\begin{align}\label{ingrid2}
(e^{-tQ_{\RR^\nu}^{\kappa,n}}\Psi)(\V{x})
&=\int_{\RR^\nu}\EE\Big[e^{-\kappa\int_0^tY^\Geb_n(\V{X}^{\V{y},\V{x}}_s)\Id s}
M_t(\V{x},\V{y})\Psi(\V{y})\Big]\Id\V{y},\quad n\in\NN.
\end{align}
Fix $\V{x}\in\RR^\nu$. Then $\EE[\|M_t(\V{x},\V{y})\|]$ is finite for a.e. $\V{y}\in\RR^\nu$
and, for every $\V{y}$ for which this is the case, Rem.~\ref{remeindrittel}
and the dominated convergence theorem imply that the expectation under the $\Id\V{y}$-integration
in \eqref{ingrid2} converges to the expectation under the integral in \eqref{ingrid1}, as $n\to\infty$.
Hence, \eqref{ingrid1} follows from Step~1 in the proof of Lem.~\ref{lemFKGeb},
the bound \eqref{domMtb2}, and another application of the dominated convergence theorem.

It is now obvious how to formulate the analogue of Step~4 in the proof of Lem.~\ref{lemFKGeb}.
\end{proof}

%%%%%%%%%%%%%%%%%%%%%%%%%%%%%%%%%%%%%%%%%%%%%%%%%
%%%%%%%%%%%%%%%%%%%%%%%%%%%%%%%%%%%%%%%%%%%%%%%%%
%%%%%%%%%%%%%%%%%%%%%%%%%%%%%%%%%%%%%%%%%%%%%%%%%

\section{A Leibniz rule for vector-valued weak derivatives}\label{secLeibniz}

\noindent
Our goal in this section is to extend a version of the Leibniz rule for Sobolev functions we learned from
\cite[Lem.~2.3(i)]{HundertmarkSimon} to the vector-valued case. This is done in Thm.~\ref{LeibnizL1} 
below. While most of the time Leibniz rules for Sobolev functions are derived for a product of functions 
in $W^{1,p}$ and $W^{1,p'}$, respectively, with $p'$ denoting the conjugate exponent of $p$,
the point about Thm.~\ref{LeibnizL1} is that it applies to two $W^{1,1}$ functions and merely the three
{\em products} showing up in the Leibniz rule are assumed to be locally integrable. Similarly as in 
\cite{HundertmarkSimon} we benefit from this generality in \eqref{sissel2}, \eqref{dersgndelta}, 
and \eqref{dialog} below.

The proof of Thm.~\ref{LeibnizL1} is slightly different from the one in \cite{HundertmarkSimon}, 
also in the case where all involved Hilbert spaces are one-dimensional.

First, however, we recall a standard mollifying procedure and prove a lemma:
In the following paragraphs and the next lemma
$\sK$ is a separable Hilbert space. Let $p\in[1,\infty]$, $j\in\{1,\ldots,\nu\}$, and
$\Psi\in L^p_\loc(\Geb,\sK)$ have a weak partial derivative with respect to
$x_j$ such that $\partial_{x_j}\Psi\in L^p_\loc(\Lambda,\sK)$. Pick a cutoff function $\rho$ such that
\begin{align}\label{klausi1}
\rho\in C_0^\infty(\RR^\nu,\RR),\quad\rho\ge0,\quad\text{$\rho(\V{x})=0$, for $|\V{x}|\ge1$,}
\quad\|\rho\|_1=1.
\end{align} 
Furthermore, set
\begin{align}\nonumber
\Geb_n&:=\Big\{\V{y}\in\Geb\Big|\dist(\V{y},\partial\Geb)>\frac{1}{n}\Big\},
\\\label{klausi2}
\rho_n(\V{x})&:=n^{\nu}\rho(n\V{x}),\;\; \V{x}\in\RR^\nu,\,n\in\NN.
\end{align} 
Finally, define the mollified functions
\begin{align}\label{klaus1}
\Psi_n(\V{x}):=\int_{\Geb}\rho_n(\V{x}-\V{y})\Psi(\V{y})\Id\V{y},\quad\V{x}\in\Geb_n,\,n\in\NN.
\end{align}
Then $\Psi_n\in C^\infty(\Lambda_n,\sK)$, if $\Lambda_n\not=\emptyset$, 
and, for every compact subset $K\subset\Geb$,
\begin{align}\label{klausi3}
\|\Psi_n-\Psi\|_{L^p(K,\sK)}+
\|\partial_{x_j}\Psi_n-\partial_{x_j}\Psi\|_{L^p(K,\sK)}\xrightarrow{\;\;n\to\infty\;\;}0,
\quad\text{if $p<\infty$.}
\end{align}
If $p=\infty$, then $\Psi_n\to\Psi$ and $\partial_{x_j}\Psi_n\to\partial_{x_j}\Psi$ a.e. on $\Geb$. 
As remarked in \cite[Rem.~2.4]{Matte2017} these assertions can be proved in virtually the same
way as in the scalar case.

The next lemma will be used to compute weak derivatives of certain cutoff versions of
vector-valued functions. In its statement and henceforth we abbreviate
\begin{align}\label{defZdelta}
Z_\delta(\Psi)&:=(\delta^2+\|\Psi\|_{\sK}^2)^\eh,\quad{\fS_{\delta,\Psi}}:=Z_\delta(\Psi)^{-1}\Psi,
\end{align}
for all $\Psi\in L_\loc^1(\Omega,\sK)$ and $\delta>0$.
We also use the notation $\fS_{\Psi}$ introduced in \eqref{defsgn}.

\begin{lem}\label{lem-abs-val}
Let $j\in\{1,\ldots,\nu\}$, $p\in[1,\infty]$, and $\delta>0$. 
Assume that $\Psi\in L_\loc^{p}(\Geb,\sK)$ has a weak partial derivative with respect to $x_j$ 
satisfying $\partial_{x_j}\Psi\in L_\loc^{p}(\Geb,\sK)$. Then $\|\Psi\|_\sK\in L_{\loc}^{p}(\Geb)$ and 
$Z_\delta(\Psi)\in L_{\loc}^{p}(\Geb)$ have weak partial derivatives
\begin{align}\label{dia00}
\partial_{x_j}\|\Psi\|_\sK&=\Re\SPn{{\fS_\Psi}}{\partial_{x_j}\Psi}_{\sK}\in L_{\loc}^{p}(\Geb),
\\\label{dia00delta}
\partial_{x_j}Z_\delta(\Psi)&=\Re\SPn{{\fS_{\delta,\Psi}}}{\partial_{x_j}\Psi}_{\sK}\in L_{\loc}^{p}(\Geb).
\end{align}
Furthermore, let $m\in\NN$, $\vr\in C^\infty(\RR,\RR)$ such that $0\le\vr\le1$,
$\vr=1$ on $(-\infty,1]$, and $\vr=0$ on $[2,\infty)$, and set $\vr_m(t):=\vr(m^{-1}\ln(t))$, $t>0$, 
so that $|\vr_m'(t)|\le\|\vr'\|_\infty/mt$. Put  $\beta_{m}:=\vr_m(Z_1(\Psi))$.
Then $\beta_{m}\Psi\in L^\infty(\Lambda,\sK)$ has a weak partial derivative
with respect to $x_j$ satisfying $\partial_{x_j}(\beta_{m}\Psi)\in L_\loc^p(\Lambda,\sK)$ and
\begin{align}\label{miles1}
\partial_{x_j}(\beta_{m}\Psi)&=\vr_m'(Z_1(\Psi))
\Re\SPn{{\fS_{1,\Psi}}}{\partial_{x_j}\Psi}_{\sK}\Psi+\beta_m\partial_{x_j}\Psi.
\end{align}
\end{lem}

\begin{proof} 
The relations \eqref{dia00} and  \eqref{dia00delta} are derived in \cite[Lem.~2.5]{Matte2017}, whence 
we only need to prove \eqref{miles1}. With $\Psi_n$ as in \eqref{klaus1} we define
$\beta_{m,n}:=\vr_m(Z_1(\Psi_n))\in C^\infty(\Geb_n)$, so that
$\beta_{m,n}\Psi_n\in L^\infty(\Geb_n,\sK)\cap C^\infty(\Geb_n,\sK)$, for all $n\in\NN$. Then
\begin{align}\label{sissel22}
\partial_{x_j}\beta_{m,n}&=\vr_m'(Z_1(\Psi_n))\Re\SPn{\fS_{1,\Psi_n}}{\partial_{x_j}\Psi_n}_{\sK}
\quad\text{on $\Lambda_n$, $n\in\NN$.}
\end{align}
Let $\eta\in\sD(\Lambda,\sK)$ and pick some compact $K\subset\Geb$
with $\supp(\eta)\subset\mr{K}$ as well as some $n_0\in\NN$ such that $K\subset\Geb_{n_0}$.
For all $n\ge n_0$, we then have
\begin{align}\label{sissel23}
\int_K\SPn{\partial_{x_j}\eta}{\beta_{m,n}\Psi_{n}}_{\sK}\Id\V{x}&=-\int_K
\SPn{\eta}{(\partial_{x_j}\beta_{m,n})\Psi_{n}+\beta_{m,n}\partial_{x_j}\Psi_{n}}_{\sK}\Id\V{x}.
\end{align}
By virtue of the Riesz-Fischer theorem for $L^1(K,\sK)$ we find integers
$n_0\le n_1< n_2<\ldots$ and dominating functions $R,{R'}\in L^1(K)$ such that 
$\Psi_{n_\ell}\to\Psi$ and $\partial_{x_j}\Psi_{n_\ell}\to\partial_{x_j}\Psi$, 
a.e. on $K$ as $\ell\to\infty$, 
and such that $\|\Psi_{n_\ell}\|_{\sK}\le{R}$, $\|\partial_{x_j}\Psi_{n_\ell}\|_{\sK}\le R'$, 
a.e. on $K$, for every $\ell\in\NN$.  On account of the bound $|\vr_m'(t)|\le\|\vr'\|_\infty/mt$, $t>0$, 
and \eqref{sissel22}, $\|(\partial_{x_j}\beta_{m,n_\ell})\Psi_{n_\ell}\|_{\sK}\le(\|\vr'\|_\infty/{m})R'$,
$\ell\in\NN$. By dominated convergence,
both sides of \eqref{sissel23} thus converge, along the same subsequence, to the respective side of
\begin{align*}
\int_\Geb\SPn{\partial_{x_j}\eta}{\beta_{m}\Psi}_{\sK}\Id\V{x}&=-\int_\Geb
\SPn{\eta}{(\vr_m'(Z_1(\Psi))\Re\SPn{\fS_{1,\Psi}}{\partial_{x_j}\Psi}_{\sK})\Psi
+\beta_{m}\partial_{x_j}\Psi}_{\sK}\Id\V{x}.
\end{align*}
These remarks prove \eqref{miles1}.
\end{proof}

We are now in a position to prove the promised Leibniz rule:

\begin{thm}\label{LeibnizL1}
Let $\sK_1,\sK_2,\sK_3$ be real or complex separable Hilbert spaces and 
$$
b:\sK_1\times\sK_2\longrightarrow\sK_3
$$ 
be real bilinear and continuous. Let $j\in\{1,\ldots,\nu\}$ and 
$\Psi_i\in L^1_\loc(\Geb,\sK_i)$, $i\in\{1,2\}$, have weak partial derivatives 
$\partial_{x_j}\Psi_i\in L^1_\loc(\Geb,\sK_i)$ such that
$$
\|\Psi_1\|_{\sK_1}\|\Psi_2\|_{\sK_2}+\|\partial_{x_j}\Psi_1\|_{\sK_1}\|\Psi_2\|_{\sK_2}+
\|\Psi_1\|_{\sK_1}\|\partial_{x_j}\Psi_2\|_{\sK_2}\in L^1_\loc(\Geb).
$$
Then $b(\Psi_1,\Psi_2)\in L^1_\loc(\Geb,\sK_3)$ has a weak partial derivative with respect to
$x_j$ and
\begin{align}\label{Leibnizb}
\partial_{x_j}b(\Psi_1,\Psi_2)=b(\partial_{x_j}\Psi_1,\Psi_2)+b(\Psi_1,\partial_{x_j}\Psi_2).
\end{align}
\end{thm}

\begin{proof}
{\em Step~1.} To start with we suppose in addition that 
$\Psi_i\in L^\infty(\Geb,\sK_i)$, $i\in\{1,2\}$. Putting $\Psi_i$ in place of $\Psi$ in
\eqref{klaus1} we construct mollified functions
$\Psi_{i,n},\in C^\infty(\Geb_n,\sK_i)$, $n\in\NN$, $i\in\{1,2\}$, such that
$\Psi_{i,n}\to\Psi_i$ and $\partial_{x_j}\Psi_{i,n}\to\partial_{x_j}\Psi_i$
in $L^1(K,\sK_i)$ for every compact $K\subset\Geb$.
Since $\rho_n$ in \eqref{klaus1} is a probability density, we further have the dominations 
$\|\Psi_{i,n}\|_{\sK_i}\le\|\Psi_i\|_\infty:=\|\Psi_i\|_{L^\infty(\Geb,\sK_i)}$.

Now fix some compact $K\subset\Lambda$ and $n_0\in\NN$ with
$K\subset\Lambda_{n_0}$.
Employing the Riesz-Fischer theorem for $L^1(K,\sK_i)$ we can find integers
$n_0\le n_1< n_2<\ldots$ such that $\Psi_{i,n_\ell}\to\Psi_i$ and 
$\partial_{x_j}\Psi_{i,n_\ell}\to\partial_{x_j}\Psi_i$, a.e. on $K$ as $\ell\to\infty$, for $i\in\{1,2\}$.
The Riesz-Fischer theorem further implies the existence of $R_i\in L^1(K)$ such that 
$\|\partial_{x_j}\Psi_{i,n_\ell}\|_{\sK_i}\le R_i$, a.e. on $K$, for all $\ell\in\NN$ and
$i\in\{1,2\}$. Now the continuity of $b$ implies
\begin{align*}
\partial_{x_j}b(\Psi_{1,n_\ell},\Psi_{2,n_{\ell}})&=b(\partial_{x_j}\Psi_{1,n_\ell},\Psi_{2,n_{\ell}})
+b(\Psi_{1,n_\ell},\partial_{x_j}\Psi_{2,n_{\ell}})\quad\text{on $\Geb_{n_\ell}$,}\,\ell\in\NN,
\end{align*}
where the right hand side converges a.e. on $K$ to the right hand side of \eqref{Leibnizb}
and is dominated by $\|b\|(R_1\|\Psi_2\|_\infty+\|\Psi_1\|_\infty R_2)\in L^1(K)$.
Furthermore, $b(\Psi_{1,n_\ell},\Psi_{2,n_{\ell}})\to b(\Psi_{1},\Psi_{2})$, $\ell\to\infty$, 
and $\|b(\Psi_{1,n_\ell},\Psi_{2,n_{\ell}})\|_{\sK_3}\le\|b\|\|\Psi_1\|_\infty\|\Psi_2\|_\infty$,
a.e. on $K$. Since $K\Subset\Geb$ was an arbitrary compact subset, 
this proves \eqref{Leibnizb} under the present extra assumptions.

{\em Step~2.} 
Next, we treat the general case with $\Psi_i$ as in the statement.
According to Step~1 and the last statement of Lem.~\ref{lem-abs-val}
we may already apply \eqref{Leibnizb} to $\Phi_{i,n}:=\beta_{i,n}\Psi_i\in L^\infty(\Geb,\sK_i)$, 
where $\beta_{i,n}:=\vr_n(Z_1(\Psi_{i}))$, $n\in\NN$, $i\in\{1,2\}$, and $\vr_n$ is defined as
in the statement of Lem.~\ref{lem-abs-val}. These remarks entail
\begin{align}\nonumber
\partial_{x_j}b(\Phi_{1,n},\Phi_{2,n})&=\beta_{1,n}\beta_{2,n}
\big(b(\partial_{x_j}\Psi_1,\Psi_2)+b(\Psi_1,\partial_{x_j}\Psi_2)\big)
\\\nonumber
&\quad+\beta_{2,n}\vr_n'(Z_1(\Psi_1))
\Re\SPn{{\fS_{1,\Psi_1}}}{\partial_{x_j}\Psi_1}_{\sK_1}b(\Psi_1,\Psi_2)
\\\label{harald111}
&\quad+\beta_{1,n}\vr_n'(Z_1(\Psi_2))
\Re\SPn{{\fS_{1,\Psi_2}}}{\partial_{x_j}\Psi_2}_{\sK_2}b(\Psi_1,\Psi_2).
\end{align}
Since $\beta_{i,n}\to1$, $n\to\infty$, on $\Geb$ and 
$|\vr_n'(Z_1(\Psi_i))|\|\Psi_i\|_{\sK_i}\le\|\vr'\|_\infty/n$,
the right hand side of \eqref{harald111} converges to the right hand side of \eqref{Leibnizb} in
$L^1_\loc(\Geb,\sK_3)$, as $n\to\infty$, by the dominated convergence theorem, 
the boundedness of $b$, and the assumptions 
$\|\partial_{x_j}\Psi_1\|_{\sK_1}\|\Psi_2\|_{\sK_2}\in L^1_\loc(\Geb)$ and 
$\|\Psi_1\|_{\sK_1}\|\partial_{x_j}\Psi_2\|_{\sK_2}\in L^1_\loc(\Geb)$.
Since also $b(\Phi_{1,n},\Phi_{2,n})\to b(\Psi_1,\Psi_2)$ in $L^1_\loc(\Geb,\sK_3)$ 
by dominated convergence, boundedness of $b$, and the assumption 
$\|\Psi_1\|_{\sK_1}\|\Psi_2\|_{\sK_2}\in L^1_\loc(\Geb)$,
this concludes the proof of \eqref{Leibnizb} in full generality.
\end{proof}

%%%%%%%%%%%%%%%%%%%%%%%%%%%%%%%%%%%%%%%%%%%%%%%%%
%%%%%%%%%%%%%%%%%%%%%%%%%%%%%%%%%%%%%%%%%%%%%%%%%
%%%%%%%%%%%%%%%%%%%%%%%%%%%%%%%%%%%%%%%%%%%%%%%%%

\section{Approximation with respect to Pauli-Fierz forms}\label{secapprox}

\noindent
In this section we collect several fairly technical but crucial results on convergence and 
approximation with respect to the norm associated with the maximal Pauli-Fierz form 
$\mathfrak{h}_{\Geb,\mathrm{N}}$ defined in \eqref{adam11}. 
{\em In the whole section we will always assume \eqref{Vlocint}, \eqref{introA}, and \eqref{introG}.}
As prerequisites we shall need some more results of \cite{Matte2017} 
which are collected in the first two of the following remarks:

\begin{rem}\label{remapprox1}
Let $j\in\{1,\ldots,\nu\}$ and $\Psi\in\dom(v_{\Geb,j}^*)$. Consider the vectors
\begin{align}\label{defPsive}
\Psi_\ve:=N_\ve^\mh\Psi\in\fdom(\Id\Gamma(1)),\quad\ve>0, 
\end{align}
where
\begin{align}\label{defNve}
N_\ve&:=1+\ve\Id\Gamma(1).
\end{align}
Introduce densely defined operators in $\sF$ by
\begin{align*}
C_\ve(G_{j,\V{x}})\psi&:=N_\ve^\mh\vp(G_{j,\V{x}})\psi-\vp(G_{j,\V{x}})N_\ve^\mh\psi,
\quad\psi\in\fdom(\Id\Gamma(1)),
\end{align*}
for all $\V{x}\in\Geb$. Then
\begin{align}\label{convCve1}
\|C_\ve(G_{j,\V{x}})\|&\le(4/\pi)\ve^\eh\|G_{j,\V{x}}\|_{\HP},\quad\V{x}\in\Geb.
\end{align}
Moreover, $\Geb\ni\V{x}\mapsto C_\ve(G_{j,\V{x}})^*\in\LO(\sF)$ is strongly measurable and
the densely defined operators $C_\ve(G_{j,\V{x}})^*N_\ve^\eh$ are bounded with
\begin{align}\label{convCve2}
\|C_\ve(G_{j,\V{x}})^*N_\ve^\eh\|&\le2\ve^\eh\|G_{j,\V{x}}\|_{\HP},\quad\V{x}\in\Geb.
\end{align}
(To obtain \eqref{convCve1} we choose the constant dispersion relation $1$ in
Lem.~2.9(1) of \cite{Matte2017}. The bound \eqref{convCve2} follows upon choosing $\ve$ as 
dispersion relation in \cite[Lem.~2.9(2)]{Matte2017}. The asserted strong measurability is observed
prior to Lem.~3.2 in \cite{Matte2017}.)

Now \cite[Lem.~3.2]{Matte2017} says that $\Psi_\ve$ has a weak partial 
derivative with respect to $x_j$ given by 
\begin{align}\label{partialPsive}
\partial_{x_j}\Psi_\ve&=iN_\ve^\mh v_{\Geb,j}^*\Psi+iA_j\Psi_\ve+i\vp(G_j)\Psi_\ve
+iC_\ve(G_{j})^*\Psi\quad\text{in $L^1_\loc(\Lambda,\sF)$.}
\end{align}
(To see this we apply the quoted lemma with dispersion relation $1$;
notice that in fact $\vp(G_{j})\Psi_\ve\in L^1_\loc(\Geb,\sF)$ by \eqref{rbvp} with $\vk=1$.)
\end{rem}

\begin{rem}\label{remapprox2}
Let $j\in\{1,\ldots,\nu\}$, $\Psi\in\dom(v_{\Geb,j}^*)$, and let $\Psi_\ve$ be given by 
\eqref{defPsive} and \eqref{defNve}. Under the additional condition that
\begin{align}\label{romeo1}
\Geb\ni\V{x}\longmapsto\|G_{j,\V{x}}\|_{\HP}\|\Psi(\V{x})\|_{\sF} \
\text{is in $L^2(\Geb)$,}
\end{align}
we observed in \cite[Lem.~3.3]{Matte2017} (here applied with dispersion
relation $1$) that $\Psi_\ve\in\dom(v_{\Geb,j}^*)$, for all $\ve>0$, and
$\Psi_\ve\to\Psi$, $\ve\downarrow0$, with respect to the graph norm of $v_{\Geb,j}^*$.
\end{rem}

In what follows, the symbol $L_0^\infty$ stands for essentially bounded functions of compact support.

\begin{rem}\label{remapprox3}
Let $\Psi\in\dom(\mathfrak{h}_{\Geb,\mathrm{N}})\cap L_0^\infty(\Geb,\sF)$. Then
the dominated convergence theorem implies that
$V^\eh N_\ve^\mh\Psi\to V^\eh\Psi$ in $L^2(\Geb,\sF)$ and 
$N_\ve^\mh\Psi\to\Psi$ in $L^2(\Geb,\fdom(\Id\Gamma(\omega)))$, as $\ve\downarrow0$.
Since $\Psi$ satisfies \eqref{romeo1} for all
$j\in\{1,\ldots,\nu\}$, we may thus infer from Rem.~\ref{remapprox2} that
$\Psi_\ve\to\Psi$ with respect to the form norm of $\mathfrak{h}_{\Geb,\mathrm{N}}$. Of course,
$\Psi_\ve\in L_0^\infty(\Geb,\fdom(\Id\Gamma(1)))$, for every $\ve>0$.

In particular, if $\{\Phi\in\dom(\mathfrak{h}_{\Geb,\mathrm{N}})|\,\Phi\in L_0^\infty(\Geb,\sF)\}$
is a core for $\mathfrak{h}_{\Geb,\mathrm{N}}$, then the set
$\{\Phi\in\dom(\mathfrak{h}_{\Geb,\mathrm{N}})|\,\Psi\in L_0^\infty(\Geb,\fdom(\Id\Gamma(1)))\}$
is a core for $\mathfrak{h}_{\Geb,\mathrm{N}}$ as well.
\end{rem}

In our first approximation lemma we treat cutoffs in the range of $\Psi$.
Similar cutoffs have been used in \cite[Lem.~2]{LeinfelderSimader1981} 
and \cite[Step~1 on p.~125]{HundertmarkSimon} to study magnetic Schr\"{o}dinger operators.

\begin{lem}\label{lemrangecutoff}
Let $j\in\{1,\ldots,\nu\}$ and $\Psi\in\dom(v_{\Geb,j}^*)$.
Define the cutoff functions $\beta_n=\vr_n(Z_1(\Psi))$, $n\in\NN$, as in the statement of 
Lem.~\ref{lem-abs-val} so that $\beta_n\Psi\in L^\infty(\Geb,\sF)$. Then
$\beta_n\Psi\in\dom(v_{\Geb,j}^*)$, for all $n\in\NN$,
and $\beta_n\Psi\to\Psi$, $n\to\infty$, with respect to the graph norm of $v_{\Geb,j}^*$.
\end{lem}

\begin{proof}
It is clear that $\beta_n\Psi\to\Psi$, $n\to\infty$, in $L^2(\Geb,\sF)$.
Let $\ve>0$ and consider the vector $\Psi_\ve$ defined in \eqref{defPsive}.
Combining \eqref{dia00delta} and \eqref{partialPsive} we obtain
\begin{align}\nonumber
\partial_{x_j}Z_1(\Psi_\ve)&=\Re\SPn{\fS_{1,\Psi_\ve}}{\partial_{x_j}\Psi_\ve}_{\sF}
\\\nonumber
&=\Re\SPn{\fS_{1,\Psi_\ve}}{iN_\ve^\mh v_{\Geb,j}^*\Psi+iA_j\Psi_\ve+i\vp(G_j)\Psi_\ve
+iC_\ve(G_{j})^*\Psi}_\sF
\\\label{sissel1}
&=\Re\SPn{\fS_{1,\Psi_\ve}}{iN_\ve^\mh v_{\Geb,j}^*\Psi+iC_\ve(G_{j})^*\Psi}_\sF
\in L^1_\loc(\Geb).
\end{align}
In the third step we used that $\Re\SPn{\Psi_\ve}{iA_j\Psi_\ve}_{\sF}$
and $\Re\SPn{\Psi_\ve}{i\vp(G_{j})\Psi_\ve}_{\sF}$ vanish a.e. on $\Geb$ since $A_j$ is real and
$\vp(G_{j,\V{x}})$ symmetric on $\fdom(\Id\Gamma(1))$. Let also $n\in\NN$.
Applying the chain rule for distributional derivatives
(see, e.g., \cite[Thm.~6.16]{LiebLoss2001}) to compute the weak partial derivative of
\begin{align*}
\beta_{n,\ve}&:=\vr_n(Z_1(\Psi_\ve)),
\end{align*} 
and combining the result with the Leibniz rule of Thm.~\ref{LeibnizL1}, we further find
\begin{align}\label{sissel2}
\partial_{x_j}(\beta_{n,\ve}\Psi_\ve)&=\vr_n'(Z_1(\Psi_\ve))(\partial_{x_j}Z_1(\Psi_\ve))\Psi_\ve
+\beta_{n,\ve}\partial_{x_j}\Psi_\ve \ \ \text{in $L^1_\loc(\Geb,\sF)$.}
\end{align}
Here we took into account that, by the construction of $\vr_n(t)=\vr(n^{-1}\ln(t))$, 
\begin{align}\label{sissel3}
|\vr_n'(Z_1(\Psi_\ve))|\|\Psi_\ve\|_{\sF}&\le\|\vr'\|_\infty\frac{\|\Psi_\ve\|_{\sF}}{nZ_1(\Psi_\ve)}
\le\frac{\|\vr'\|_\infty}{n}.
\end{align}
Together with \eqref{sissel1} this shows that 
$|\partial_{x_j}\beta_{n,\ve}|\|\Psi_\ve\|_\sF\in L^1_\loc(\Geb)$,
whence the Leibniz rule of Thm.~\ref{LeibnizL1} was indeed applicable.

Next, we subtract $i\beta_{n,\ve}A_j\Psi_\ve+i\beta_{n,\ve}\vp(G_j)\Psi_\ve$ from both sides of 
\eqref{sissel2}. In view of \eqref{partialPsive} this results in
\begin{align}\nonumber
\partial_{x_j}(\beta_{n,\ve}\Psi_\ve)-iA_j\beta_{n,\ve}\Psi_\ve-i\vp(G_j)\beta_{n,\ve}\Psi_\ve
&=\vr_n'(Z_1(\Psi_\ve))(\partial_{x_j}Z_1(\Psi_\ve))\Psi_\ve
\\\label{sissel4}
&\quad+\beta_{n,\ve}(iN_\ve^\mh v_{\Geb,j}^*\Psi+iC_\ve(G_{j})^*\Psi).
\end{align}
In the next step we compute, a.e. on $\Geb$, the $\sF$-scalar product of both sides of \eqref{sissel4}
with $\eta\in\sD(\Geb,\fdom(\Id\Gamma(1)))$, integrate the result with respect to $\V{x}\in\Geb$,
and pass to the limit $\ve\downarrow0$ afterwards. In doing so we observe that, as $\ve\downarrow0$,
\begin{enumerate}
\item[(a)] $\beta_{n,\ve}\to\beta_n$ pointwise (recall $\beta_{n,\ve}\le1$);
\item[(b)] $\beta_{n,\ve}\Psi_\ve\to\beta_n\Psi$ in $L^2(\Geb,\sF)$;
\item[(c)] $N_\ve^{\mh}v_{\Geb,j}^*\Psi\to v_{\Geb,j}^*\Psi$ in $L^2(\Geb,\sF)$;
\item[(d)] $C_\ve(G_{j})^*\Psi\to0$ in $L^1_\loc(\Geb,\sF)$ by \eqref{convCve1};
\item[(e)] $\partial_{x_j}Z_1(\Psi_\ve)\to\Re\SPn{\fS_{1,\Psi}}{iv_{\Geb,j}^*\Psi}$ 
in $L^1_\loc(\Geb)$ by \eqref{sissel1}, (d), and (e);
\item[(f)] $\vr_n'(Z_1(\Psi_\ve))\Psi_\ve\to\vr_n'(Z_1(\Psi))\Psi$ pointwise with
the $\ve$-uniform bound \eqref{sissel3}.
\end{enumerate}
We thus arrive at
\begin{align*}
&\SPn{v_{\Geb,j}\eta}{i\beta_n\Psi}
\\
&=\int_\Geb\SPB{\eta(\V{x})}{\vr_n'(Z_1(\Psi(\V{x})))\Re\SPn{\fS_{1,\Psi}(\V{x})}{
i(v_{\Geb,j}^*\Psi)(\V{x})}_\sF\Psi(\V{x})+i(\beta_nv_{\Geb,j}^*\Psi)(\V{x})}_\sF\Id\V{x}.
\end{align*}
Next, we observe that the preceding integral is the scalar product of $\eta$ with a vector in
$L^2(\Geb,\sF)$ since, in analogy to \eqref{sissel3}, 
\begin{align}\label{sissel30}
|\vr_n'(Z_1(\Psi(\V{x})))|\|\Psi\|_\sF\le\frac{\|\vr'\|_\infty}{n},
\end{align} 
and since $\Re\SPn{\fS_{1,\Psi}}{iv_{\Geb,j}^*\Psi}_\sF\in L^2(\Geb)$ and, 
of course, $i\beta_nv_{\Geb,j}^*\Psi\in L^2(\Geb,\sF)$.
By the definition of the adjoint operator $v_{\Geb,j}^*$ this reveals that 
$\beta_n\Psi\in\dom(v_{\Geb,j}^*)$ with
\begin{align*}
v_{\Geb,j}^*(\beta_n\Psi)&=-i\vr_n'(Z_1(\Psi))\Re\SPn{\fS_{1,\Psi}}{iv_{\Geb,j}^*\Psi}_\sF\Psi
+\beta_nv_{\Geb,j}^*\Psi.
\end{align*}
Taking into account \eqref{sissel30}, $\Re\SPn{\fS_{1,\Psi}}{iv_{\Geb,j}^*\Psi}_\sF\in L^2(\Geb)$, and 
$\beta_nv_{\Geb,j}^*\Psi\to v_{\Geb,j}^*\Psi$ in $L^2(\Geb,\sF)$ , we further conclude that 
$v_{\Geb,j}^*(\beta_n\Psi)\to v_{\Geb,j}^*\Psi$, as $n\to\infty$.
\end{proof}

\begin{rem}\label{remrangecutoff}
Let $\Psi\in\dom(\mathfrak{h}_{\Geb,\mathrm{N}})$ and consider again the cutoffs $\beta_n$
appearing in Lem.~\ref{lemrangecutoff}. Then the dominated convergence theorem implies 
$V^\eh\beta_n\Psi\to V^\eh\Psi$ in $L^2(\Geb,\sF)$ and
$\beta_n\Psi\to\Psi$ in $L^2(\Geb,\fdom(\Id\Gamma(\omega)))$.
Together with Lem.~\ref{lemrangecutoff} this shows in particular that
$\{\Phi\in\dom(\mathfrak{h}_{\Geb,\mathrm{N}})|\,\Phi\in L^\infty(\Geb,\sF)\}$
is a core for $\mathfrak{h}_{\Geb,\mathrm{N}}$.
\end{rem}

We continue with a simple result on spatial cutoffs: 

\begin{lem}\label{cutofflem}
Pick cutoff functions $\vt_\ell\in C_0^\infty(\RR^\nu,\RR)$, $\ell\in\NN$, satisfying 
\begin{align*}
0\le\vt_\ell\le1,\quad \text{$\vt_{\ell+1}=1$ on 
$\supp(\vt_\ell)$},\quad\ell\in\NN;\qquad\Geb\subset\bigcup_{\ell\in\NN}\supp(\vt_\ell).
\end{align*} 
Let $j\in\{1,\ldots,\nu\}$ and define $\Theta_j:\RR^\nu\to[0,\infty)$ by
\begin{align}\label{defThetaJ}
\Theta_j&:=\Big(\sum_{\ell=1}^\infty|\partial_{x_j}\vt_\ell|^2\Big)^\eh.
\end{align} 
Finally, let $\Psi\in\dom(v_{\Geb,j}^*)$ satisfy $\Theta_j\Psi\in L^2(\Geb,\sF)$. Then 
$\vt_n\Psi\to\Psi$, $n\to\infty$, with respect to the graph norm of $v_{\Geb,j}^*$.
\end{lem}

\begin{proof}
Of course, $\vt_n\Psi\to\Psi$ in $L^2(\Geb,\sF)$. Furthermore,
it is straightforward to check that $\vt_n\dom(v_{\Geb,j}^*)\subset\dom(v_{\Geb,j}^*)$ with 
$v_{\Geb,j}^*(\vt_n\Phi)=\vt_nv_{\Geb,j}^*\Phi-i(\partial_{x_j}\vt_n)\Phi$, 
for all $\Phi\in\dom(v_{\Geb,j}^*)$. The condition $\Theta_j\Psi\in L^2(\Geb,\sF)$ 
and the dominated convergence theorem imply
\begin{align*}
\|(\partial_{x_j}\vt_n)\Psi\|\le\Big\|\Big(\sum_{\ell=n}^\infty
|\partial_{x_j}\vt_\ell|^2\Big)^\eh\Psi\Big\|\xrightarrow{\;\;n\to\infty\;\;}0.
\end{align*}
Since also $\vt_n v_{\Geb,j}^*\Psi\to v_{\Geb,j}^*\Psi$, these remarks show that 
$v_{\Geb,j}^*(\vt_n\Psi)\to v_{\Geb,j}^*\Psi$, as $n\to\infty$.
\end{proof}

\begin{lem}\label{lemwiesel1}
Assume that the cutoff functions in Lem.~\ref{cutofflem} are chosen such that
$\supp(\vt_\ell)\subset\Geb$, for all $\ell\in\NN$. Furthermore, assume that the functions
$\Theta_{j}$ defined in \eqref{defThetaJ} satisfy
\begin{align}\label{VdomTheta}
\sum_{j=1}^\nu\Theta_{j}^2\le C+ V,\quad\text{for some constant $C>0$.}
\end{align}
Then $\{\Phi\in\dom(\mathfrak{h}_{\Geb,\mathrm{N}})|\,
\Phi\in L_0^\infty(\Geb,\fdom(\Id\Gamma(1)))\}$ is a core for $\mathfrak{h}_{\Geb,\mathrm{N}}$.
\end{lem}

\begin{proof}
Let $\Psi\in\dom(\mathfrak{h}_{\Geb,\mathrm{N}})$. Then 
$V^\eh\vt_\ell\Psi\to V^\eh\Psi$ in $L^2(\Geb,\sF)$ and
$\vt_\ell\Psi\to\Psi$ in $L^2(\Geb,\fdom(\Id\Gamma(\omega)))$, as $\ell\to\infty$, by dominated 
convergence. Since \eqref{VdomTheta} entails $\Theta_{j}\Psi\in L^2(\Geb,\sF)$, for all 
$j\in\{1,\ldots,\nu\}$, Lem.~\ref{cutofflem} implies that $\vt_\ell\Psi\to\Psi$ with respect to the graph 
norm of every $v_{\Geb,j}^*$. Altogether this shows that 
$\vt_\ell\Psi\in\dom(\mathfrak{h}_{\Geb,\mathrm{N}})$, for all $\ell\in\NN$, and
$\vt_\ell\Psi\to\Psi$ with respect to the form norm on $\dom(\mathfrak{h}_{\Geb,\mathrm{N}})$.
By virtue of Rem.~\ref{remrangecutoff} we conclude that
$\{\Phi\in\dom(\mathfrak{h}_{\Geb,\mathrm{N}})|\,\Phi\in L_0^\infty(\Geb,\sF)\}$
is a core for $\mathfrak{h}_{\Geb,\mathrm{N}}$. Now the assertion follows 
directly from Rem.~\ref{remapprox3}.
\end{proof}

In the next lemma we consider the choice $\Geb=\RR^\nu$:

\begin{lem}\label{lemwiesel2}
The set $\{\Phi\in\dom(\mathfrak{h}_{\RR^\nu,\mathrm{N}})|\,
\Phi\in L_0^\infty(\RR^\nu,\fdom(\Id\Gamma(1)))\}$
is a core for $\mathfrak{h}_{\RR^\nu,\mathrm{N}}$.
\end{lem}

\begin{proof}
In the case $\Geb=\RR^\nu$, the functions $\vt_\ell$ appearing Lem.~\ref{cutofflem} can 
obviously be chosen such that $\Theta_{1},\ldots,\Theta_\nu$ are bounded. Then
\eqref{VdomTheta} is satisfied, whence the assertion follows from Lem.~\ref{lemwiesel1}.
\end{proof}

Next, we study approximations by elements of
\begin{align*}
\sC\otimes\sE&:=\mathrm{span}_{\CC}\{f\phi|\,f\in\sC,\,\phi\in\sE\},
\end{align*}
 with suitable subspaces $\sC\subset L^2(\Geb)$ and $\sE\subset\sF$.
 
\begin{lem}\label{lemvstar2}
Let $\Psi\in\dom(\mathfrak{h}_{\Geb,\mathrm{N}})$. Then the following holds:
\begin{enumerate}[leftmargin=0.8cm]
\item[{\rm(1)}] Assume in addition that
$\Psi(\V{x})\in\fdom(\Id\Gamma(1))$, for a.e. $\V{x}\in\Geb$, and
\begin{align}\label{extra1}
\Geb\ni\V{x}\longmapsto\|G_{j,\V{x}}\|_{\HP}
\|\Psi(\V{x})\|_{\fdom(\Id\Gamma(1))}\quad\text{is in $L^2(\Lambda)$ for all $j\in\{1,\ldots,\nu\}$.}
\end{align}
Then there exist
\begin{align}\label{wiesel1}
\Psi_n\in \{L^2(\Geb)\otimes\fdom(\Id\Gamma(1\vee\omega))\}
\cap\dom(\mathfrak{h}_{\Geb,\mathrm{N}}),\quad n\in\NN,
\end{align}
such that $\Psi_n\to\Psi$, $n\to\infty$, with respect to the form norm of
$\mathfrak{h}_{\Geb,\mathrm{N}}$.
\item[{\rm(2)}] Assume in addition that $\Psi\in L^\infty_0(\Geb,\fdom(\Id\Gamma(1)))$, 
Then there exist
\begin{align}\label{wiesel11}
\Psi_n\in \{L_0^\infty(\Geb)\otimes\fdom(\Id\Gamma(1\vee\omega))\}\cap
\dom(\mathfrak{h}_{\Geb,\mathrm{N}}),\quad n\in\NN,
\end{align}
such that $\Psi_n\to\Psi$, $n\to\infty$, with respect to the form norm of
$\mathfrak{h}_{\Geb,\mathrm{N}}$.
\end{enumerate}
\end{lem}

\begin{proof}
We will always assume that $\Psi$ satisfies the additional condition imposed on it in Part~(1) and
we shall fix  $j\in\{1,\ldots,\nu\}$ in the first four steps of this proof.

{\em Step~1.}
We define a symmetric operator $w_{\Geb,j}$ in $L^2(\Geb,\sF)$ by setting 
$\dom(w_{\Geb,j}):=\sD(\Geb,\sF)$ and
\begin{align*}
w_{\Geb,j}\Phi&:=-i\partial_{x_j}\Phi-A_j\Phi,\quad\Phi\in\dom(w_{\Geb,j}).
\end{align*}
According to \cite[Rem.~3.1(1)]{Matte2017} we then have $\Psi\in\dom(w_{\Geb,j}^*)$ and
\begin{align*}
v_{\Geb,j}^*\Psi&=w_{\Geb,j}^*\Psi-\vp(G_j)\Psi.
\end{align*}
With the help of \eqref{rbvp} and \eqref{extra1}, which together imply
$\vp(G_j)\Psi\in L^2(\Geb,\sF)$, this is indeed straightforward to verify.

Let $Q\in\LO(\sF)$ and write $(Q\Phi)(\V{x}):=Q\Phi(\V{x})$, a.e. $\V{x}\in\Geb$, for all
$\Phi\in L_\loc^1(\Geb,\sF)$. Then it is clear that $Qw_{\Geb,j}\Phi=w_{\Geb,j}Q\Phi$, 
for every $\Phi\in\dom(w_{\Geb,j})$, from which we infer that $Q\Psi\in\dom(w_{\Geb,j}^*)$ 
with $w_{\Geb,j}^*Q\Psi=Qw_{\Geb,j}^*\Psi$.

Suppose we further have $Q\Psi(\V{x})\in\fdom(\Id\Gamma(1))$, a.e. $\V{x}\in\Geb$, with
$\|Q\Psi\|_{\fdom(\Id\Gamma(1))}\le C_{\Psi}\|\Psi\|_{\fdom(\Id\Gamma(1))}$ a.e. on $\Geb$,
for some $C_{\Psi}>0$. Then it follows from \eqref{rbvp} and \eqref{extra1} that
$\vp(G_j)Q\Psi\in L^2(\Geb,\sF)$ and the definition of the adjoint operators
$v_{\Geb,j}^*$ and $w_{\Geb,j}^*$ entails $Q\Psi\in\dom(v_{\Geb,j}^*)$ with
\begin{align}\label{fiona1}
v_{\Geb,j}^*Q\Psi&=w_{\Geb,j}^*Q\Psi-\vp(G_j)Q\Psi=Qw_{\Geb,j}^*\Psi-\vp(G_j)Q\Psi.
\end{align}

{\em Step~2.} 
For every $r\in\NN$, we define $Q_r\in\LO(\sF)$ by setting
\begin{align}\label{defQr}
Q_r\psi=\big(\psi_0,{\chi_r^{\otimes_1}}\psi_1,{\chi_r^{\otimes_2}}\psi_2,\ldots,
{\chi_r^{\otimes_r}}\psi_r,0,0,\ldots\,),
\end{align}
for every $\psi=(\psi_n)_{n=0}^\infty\in\sF$, where $\chi_r^{\otimes_m}$ 
denotes the characteristic function of the set
$$
\big\{(k_1,\ldots,k_m)\in\cK^m\big|\,\omega(k_1)\le r,\ldots,\omega(k_m)\le r\big\},
\quad m,r\in\NN.
$$
Then $Q_r\to\id$, $r\to\infty$, strongly in $\sF$ as well as in $L^2(\Geb,\sF)$. 
By the remarks in Step~1 we know that
$Q_r\Psi\in\dom(v_{\Geb,j}^*)\cap\dom(w_{\Geb,j}^*)$, $\vp(G_j)Q_r\Psi\in L^2(\Geb,\sF)$, 
and \eqref{fiona1} is satisfied with $Q=Q_r$. Furthermore,
\begin{align}\label{fiona1b}
\|\vp(G_j)Q_r\Psi-\vp(G_j)\Psi\|_\sF\le2\|G_j\|_{\HP}\|(Q_r-1)\Psi\|_{\fdom(\Id\Gamma(1))}
\xrightarrow{\;\;r\to\infty\;\;}0,
\end{align}
pointwise a.e., since $\Psi(\V{x})\in\fdom(\Id\Gamma(1))$ for a.e. $\V{x}\in\Geb$.
Employing \eqref{extra1} and the dominated convergence theorem we deduce that
$\vp(G_j)Q_r\Psi\to\vp(G_j)\Psi$, $r\to\infty$, in $L^2(\Geb,\sF)$.
Putting all these remarks together we conclude that $v_{\Geb,j}^*Q_r\Psi\to v_{\Geb,j}^*\Psi$.
The dominated convergence theorem further implies that
$V^\eh Q_r\Psi\to V^\eh \Psi$ in $L^2(\Geb,\sF)$ and $Q_r\Psi\to\Psi$ in
$L^2(\Geb,\fdom(\Id\Gamma(\omega)))$.

{\em Step~3.}
We fix $r\in\NN$ in this and the next step. The definition of $Q_r$ ensures that
$Q_r\Psi\in L^2(\Geb,\fdom(\Id\Gamma(1\vee\omega)))$.
Let $\{e_\ell:\ell\in\NN\}$ be an orthonormal basis of $\fdom(\Id\Gamma(1\vee\omega))$ and put 
\begin{align}\label{defPn}
P_n\phi:=\sum_{\ell=1}^n\SPn{e_\ell}{\phi}_{\fdom(\Id\Gamma(1\vee\omega))}e_\ell,\quad
\phi\in\fdom(\Id\Gamma(1\vee\omega)),\,n\in\NN.
\end{align} 
Then $P_nQ_r\Psi\to Q_r\Psi$, $n\to\infty$, in $L^2(\Geb,\fdom(\Id\Gamma(1\vee\omega)))$ by 
dominated convergence. Since the canonical injections
$L^2(\Geb,\fdom(\Id\Gamma(1\vee\omega)))\subset L^2(\Geb,\fdom(\Id\Gamma(\omega)))
\subset L^2(\Geb,\sF)$
are continuous, we also have $P_nQ_r\Psi\to Q_r\Psi$, $n\to\infty$, in both $L^2(\Geb,\sF)$ and
$L^2(\Geb,\fdom(\Id\Gamma(\omega)))$. Likewise,
$V^\eh P_nQ_r\Psi=P_n V^\eh Q_r\Psi\to V^\eh Q_r\Psi$ in $L^2(\Geb,\sF)$.

It remains to show that $P_nQ_r\Psi\to Q_r\Psi$, $n\to\infty$, 
with respect to the graph norm of $v_{\Geb,j}^*$, which is done in the next step.

{\em Step~4.}
Since $Q_r$ maps $\sF$ into $\fdom(\Id\Gamma(1\vee\omega))$, we see that $P_nQ_r$
defines a finite rank operator on $\sF$. Furthermore, we
notice that $P_nQ_r\Psi(\V{x})\in\fdom(\Id\Gamma(1))$, for a.e. $\V{x}\in\Geb$, with
\begin{align}\nonumber
\|P_nQ_r\Psi\|_{\fdom(\Id\Gamma(1))}&\le\|P_nQ_r\Psi\|_{\fdom(\Id\Gamma(1\vee\omega))}
\le\|Q_r\Psi\|_{\fdom(\Id\Gamma(1\vee\omega))}
\\\label{fiona2}
&\le r^\eh\|Q_r\Psi\|_{\fdom(\Id\Gamma(1))}\le r^\eh\|\Psi\|_{\fdom(\Id\Gamma(1))},
\end{align}
a.e. on $\Geb$. In the penultimate step we used that 
$\chi_r^{\otimes_m}(k_1,\ldots,k_m)\not=0$ entails
\begin{align*}
1+1\vee\omega(k_1)+\ldots+1\vee\omega(k_m)\le 1+rm\le r(1+m).
\end{align*}
Applying the remarks in Step~1 we conclude that
$$
P_nQ_r\Psi\in\dom(v_{\Geb,j}^*)\cap\dom(w_{\Geb,j}^*),
\quad\vp(G_j)P_nQ_r\Psi\in L^2(\Geb,\sF),\quad n\in\NN,
$$ 
and \eqref{fiona1} is satisfied with $Q=P_nQ_r$. Since
$Q_rw_{\Geb,j}^*\Psi$ is in $L^2(\Geb,\fdom(\Id\Gamma(1\vee\omega)))$, we further have
 $P_nQ_rw_{\Geb,j}^*\Psi\to Q_rw_{\Geb,j}^*\Psi$, as $n\to\infty$, in
$L^2(\Geb,\fdom(\Id\Gamma(1\vee\omega)))$ and, hence, also in $L^2(\Geb,\sF)$. 
Similarly to \eqref{fiona1b} we find
\begin{align*}
\|\vp(G_j)P_nQ_r\Psi-\vp(G_j)Q_r\Psi\|_{\sF}\le2\|G_j\|_{\HP}
\|(P_n-1)Q_r\Psi\|_{\fdom(\Id\Gamma(1\vee\omega))}\xrightarrow{\;\;n\to\infty\;\;}0,
\end{align*}
a.e. on $\Geb$, because $Q_r\Psi(\V{x})\in\fdom(\Id\Gamma(1\vee\omega))$, a.e. $\V{x}$.
On account of \eqref{fiona2} we further have the uniform bounds
\begin{align*}
\|\vp(G_j)P_nQ_r\Psi-\vp(G_j)Q_r\Psi\|_{\sF}\le2(r^\eh+1)\|G_j\|_{\HP}
\|\Psi\|_{\fdom(\Id\Gamma(1))}\in L^2(\Geb).
\end{align*}
Thus, $\vp(G_j)P_nQ_r\Psi\to\vp(G_j)Q_r\Psi$, $n\to\infty$, in $L^2(\Geb,\sF)$ by
dominated convergence. Altogether this shows that 
$v_{\Geb,j}^*P_nQ_r\Psi\to v_{\Geb,j}^*Q_r\Psi$, as $n\to\infty$.

{\em Step~5.} We can now conclude as follows: Let $n\in\NN$. According to Step~2 we then find
some $r_n\in\NN$ such that $\|Q_{r_n}\Psi-\Psi\|_{\mathfrak{h}_{\Geb,\mathrm{N}}}<1/2n$.
After that Steps~3 and~4 permit to pick some $m_n\in\NN$ such that
$\|P_{m_n}Q_{r_n}\Psi-Q_{r_n}\Psi\|_{\mathfrak{h}_{\Geb,\mathrm{N}}}<1/2n$. This proves
Part~(1) with
\begin{align*}
\Psi_n&:=P_{m_n}Q_{r_n}\Psi=\sum_{\ell=1}^{m_n}
\SPn{e_\ell}{Q_{r_n}\Psi}_{\fdom(\Id\Gamma(1\vee\omega))}e_\ell,\quad n\in\NN.
\end{align*}
Here $\SPn{e_\ell}{Q_{r_n}\Psi}_{\fdom(\Id\Gamma(1\vee\omega))}\in L_0^\infty(\Geb)$, whenever
$\Psi\in L_0^\infty(\Geb,\fdom(\Id\Gamma(1)))$. Since the latter condition on $\Psi$ entails
\eqref{extra1}, this also proves Part~(2).
\end{proof}

Before we consider mollifications we note a simple observation that also is part of 
\cite[Rem.~3.1(2)]{Matte2017}:

\begin{rem}\label{remvj}
Let $j\in\{1,\ldots,\nu\}$, $\Psi\in\dom(v_{\Geb,j}^*)$, and assume that 
\begin{align}\label{sigrid1}
\Geb\ni\V{x}\longmapsto\|G_{j,\V{x}}\|_{\HP}
\|\Psi(\V{x})\|_{\fdom(\Id\Gamma(1))} \ \text{is in $L^1_\loc(\Geb)$.}
\end{align}
In view of \eqref{rbvp} this entails
$\vp(G_{j})\Psi\in L_\loc^1(\Geb,\sF)$ and it is clear that $A_j\Psi\in L^1_\loc(\Geb,\sF)$. 
By the definitions of the adjoint operator $v_{\Geb,j}^*$ and the weak partial derivatives, this implies 
that $\partial_{x_j}\Psi\in L^1_\loc(\Geb,\sF)$ exists and 
\begin{align}\label{forPsi}
v_{\Geb,j}^*\Psi=-i\partial_{x_j}\Psi-A_j\Psi-\vp(G_j)\Psi
\quad(\text{sum in $L^1_\loc(\Geb,\sF)$ on the RHS}).
\end{align}
\end{rem}

\begin{lem}\label{propvjesa}
Let $\Psi\in \{L_0^\infty(\Geb)\otimes\fdom(\Id\Gamma(1\vee\omega))\}
\cap\dom(\mathfrak{h}_{\Geb,\mathrm{N}})$.
Then there exist $\Psi_n\in\sD(\Geb,\fdom(\Id\Gamma(1\vee\omega)))$, $n\in\NN$, such that
$\Psi_n\to\Psi$, $n\to\infty$, with respect to the form norm of $\mathfrak{h}_{\Geb,\mathrm{N}}$. 
\end{lem}

\begin{proof}
Thanks to Rem.~\ref{remvj} we know that the weak partial derivatives of 
$\Psi$ with respect to every $x_j$ exist and are given by the $L^1_\loc(\Geb,\sF)$-sum
$\partial_{x_j}\Psi=iv_{\Geb,j}^*\Psi+iA_j\Psi+i\vp(G_j)\Psi$. 
Together with the present assumptions on $\Psi$, \eqref{rbvp}, \eqref{introA}, and \eqref{introG} 
the latter formula reveals that actually $\partial_{x_j}\Psi\in L^2(\Geb,\sF)$. Define $\Psi_n$ as in 
\eqref{klaus1}, for all integers $n\ge n_0$ and some $n_0\in\NN$ such that 
$\mathrm{dist}(\supp(\Psi),\Geb_{n_0}^c)\ge1/n_0$. 
Extending them by $0$ outside $\Geb_n$, we obtain functions 
$\Psi_n\in\sD(\Geb,\fdom(\Id\Gamma(1\vee\omega)))$, $n\ge n_0$,
such that $\Psi_n\to\Psi$ in $L^2(\Geb,\fdom(\Id\Gamma(1\vee\omega)))$
and $\partial_{x_j}\Psi_n\to\partial_{x_j}\Psi$, $j\in J$, in $L^2(\Geb,\sF)$, as $n\to\infty$. 
All $\Psi_n$ have their supports in a fixed compact subset of $\Geb$. 
Recall the notation $N_1:=1+\Id\Gamma(1)$.
Since $\Psi\in L_0^\infty(\Geb)\otimes\fdom(\Id\Gamma(1\vee\omega))$,
we further have $\|\Psi_n\|_\infty\le\|N_1^\eh\Psi_n\|_\infty\le\|N_1^\eh\Psi\|_\infty<\infty$.
(Here $\|\cdot\|_\infty$ stands for the essential supremum of $\|\cdot\|_\sF$-norms of 
Fock space-valued functions.)
It is also clear that $N_1^\eh\Psi_n\to N_1^\eh\Psi$ in $L^2(\Geb,\sF)$.
Therefore, we find a subsequence of $\{\Psi_n\}_{n\ge n_0}$, call it $\{\Psi_n'\}_{n\in\NN}$, 
such that $N_1^\eh\Psi_n'\to N_1^\eh\Psi$ a.e. on $\Geb$.
From these remarks and the dominated convergence theorem
we infer that $A_j\Psi_n'\to A_j\Psi$ in $L^2(\Geb,\sF)$. In the same way we see that
$V^\eh\Psi_n'\to V^\eh\Psi$ in $L^2(\Geb,\sF)$.
The above remarks, the dominated convergence theorem, and \eqref{rbvp} further imply that
$\vp(G_j)\Psi_n'\to\vp(G_j)\Psi$ in $L^2(\Geb,\sF)$. Moreover, it is clear that
$\Psi_n'\in\dom(v_{\Geb,j})\subset\dom(v_{\Geb,j}^*)$ with 
$v_{\Geb,j}^*\Psi_n'=v_{\Geb,j}\Psi_n'=-i\partial_{x_j}\Psi_n'-A_j\Psi_n'-\vp(G_j)\Psi_n'$, and
we conclude that $v_{\Geb,j}^*\Psi_n'\to v_{\Geb,j}^*\Psi$ in $L^2(\Geb,\sF)$, 
for all $j\in\{1,\ldots,\nu\}$.
\end{proof}

The next lemma will be used to derive a diamagnetic inequality for resolvents of
Dirichlet-Pauli-Fierz operators in Thm.~\ref{thmdiaT} below.

\begin{lem}\label{exdomNbdcdomD}
Let $\Psi\in\dom(\mathfrak{h}_{\Geb,\mathrm{N}})\cap L_0^\infty(\Geb,\sF)$.  
Then $\Psi\in\dom(\mathfrak{h}_{\Geb,\mathrm{D}})$.
\end{lem}

\begin{proof}
We have to show that $\Psi$ can be approximated with respect to the form norm of 
$\mathfrak{h}_{\Geb,\mathrm{N}}$ by elements of 
$\sD(\Geb,\fdom(\Id\Gamma(1\vee\omega)))$. But this follows upon 
combining Rem.~\ref{remapprox3}, Lem.~\ref{lemvstar2}(2), and Lem.~\ref{propvjesa}.
\end{proof}

An example for the applicability of the above approximation results in the case $\Geb=\RR^\nu$
is the following analogue of a well-known result on Schr\"{o}dinger forms \cite{SimonJOT1979}. 
The next theorem also generalizes \cite[Cor.~4.7]{Matte2017} by weakening the condition 
imposed on $\V{G}$ there. The theorem will be used in the proof of Prop.~\ref{propDFKhyp} below.

\begin{thm}\label{thmminmaxPF}
The maximal and minimal Pauli-Fierz forms on $\RR^\nu$ agree, i.e., 
\begin{align}\label{minmaxPF}
\mathfrak{h}_{\RR^\nu,\mathrm{D}}&=\mathfrak{h}_{\RR^\nu,\mathrm{N}}.
\end{align}
\end{thm}

\begin{proof}
Combine Lem.~\ref{lemwiesel2}, Lem.~\ref{lemvstar2}(2), and Lem.~\ref{propvjesa}.
\end{proof}

In view of the preceding theorem we abbreviate
\begin{align}\label{defhRRnu}
\mathfrak{h}_{\RR^\nu}:=\mathfrak{h}_{\RR^\nu,\mathrm{D}}=\mathfrak{h}_{\RR^\nu,\mathrm{N}},
\end{align}
and we shall refer to $\mathfrak{h}_{\RR^\nu}$ simply as the {\em Pauli-Fierz form on $\RR^\nu$.}

We conclude this section with a proposition providing a crucial technical ingredient needed to 
derive our Feynman-Kac formulas for $H_{\Geb}$: 
We shall verify the conditions (a) and (b) of Hyp.~\ref{hypabD}
when the forms $\mathfrak{h}_{\RR^\nu}$ and $\mathfrak{h}_{\Geb,\mathrm{D}}'$ are put in
place of $\mathfrak{q}_{\RR^\nu}$ and $\mathfrak{q}_{\Geb}$, respectively,
where we use the following notational conventions:

For any function $\Phi:\Geb\to\sF$, we denote by $\Phi'$ its extension to $\RR^\nu$ by $0$. 
For a set $\sM$ of functions from $\Geb$ to $\sF$, we put
$\sM':=\{\Phi'|\,\Phi\in\sM\}$. Restrictions of functions on $\RR^\nu$ to $\Geb$ are denoted
by a subscript $\Geb$. Finally, we define
\begin{align}\label{wiesel1700}
\mathfrak{h}_{\Geb,\mathrm{D}}'[\Psi]&:=\mathfrak{h}_{\Geb,\mathrm{D}}[\Psi_\Geb],
\quad\Psi\in\dom(\mathfrak{h}_{\Geb,\mathrm{D}}'):=\dom(\mathfrak{h}_{\Geb,\mathrm{D}})'.
\end{align}
In other words, $\mathfrak{h}_{\Geb,\mathrm{D}}'$ is $\mathfrak{h}_{\Geb,\mathrm{D}}$
considered as a form in $1_{\Geb}L^2(\RR^\nu,\sF)$ in the canonical way.

\begin{prop}\label{propDFKhyp}
Assume that $\V{A}\in L_\loc^2(\RR^\nu,\RR^\nu)$, $\V{G}\in L_\loc^2(\RR^\nu,\HP^\nu)$,
and $0\le V\in L_\loc^1(\RR^\nu,\RR)$.
Let $Y^{\Geb}_\infty$ be given by \eqref{defYsGinfty}, the functions $\vt_\ell$ being chosen as in
the paragraph preceding \eqref{defYsGinfty}. Set
\begin{align*}
\dom(\mathfrak{h}_{\RR^\nu}^{1,\infty})&:=\bigg\{\Psi\in\dom(\mathfrak{h}_{\RR^\nu})\bigg|
\:\int_{\RR^\nu}Y^{\Geb}_\infty(\V{x})\|\Psi(\V{x})\|^2\Id\V{x}<\infty\bigg\}.
\end{align*}
Then $\dom(\mathfrak{h}_{\RR^\nu}^{1,\infty})\subset\dom(\mathfrak{h}_{\Geb,\mathrm{D}}')$ 
and the closure of $\dom(\mathfrak{h}_{\RR^\nu}^{1,\infty})$ with respect to the form norm of 
$\mathfrak{h}_{\Geb,\mathrm{D}}'$ is $\dom(\mathfrak{h}_{\Geb,\mathrm{D}}')$.
Furthermore, $\mathfrak{h}_{\RR^\nu}[\Psi]=\mathfrak{h}_{\Geb,\mathrm{D}}'[\Psi]$, for every 
$\Psi\in\dom(\mathfrak{h}_{\RR^\nu}^{1,\infty})$.
\end{prop}

\begin{proof}
Let $\Psi\in\dom(\mathfrak{h}_{\RR^\nu}^{1,\infty})$. Clearly, $\Psi=0$ a.e. on $\Geb^c$.
Pick some $\Phi\in\sD(\Geb,\fdom(\Id\Gamma(1)))$. For every $j\in\{1,\ldots,\nu\}$, we then find 
\begin{align*}
\SPn{v_{\Geb,j}\Phi}{\Psi_\Geb}_{L^2(\Geb,\sF)}
&=\SPn{v_{\RR^\nu,j}\Phi'}{\Psi}_{L^2(\RR^\nu,\sF)}
\\
&=\SPn{\Phi'}{v_{\RR^\nu,j}^*\Psi}_{L^2(\RR^\nu,\sF)}
=\SPn{\Phi}{(v_{\RR^\nu,j}^*\Psi)_\Geb}_{L^2(\Geb,\sF)}.
\end{align*}
Thus, $\Psi_\Geb\in\dom(v_{\Geb,j}^*)$ with 
\begin{align}\label{wiesel2000}
v_{\Geb,j}^*\Psi_\Geb=(v_{\RR^\nu,j}^*\Psi)_\Geb.
\end{align}
This shows that $\Psi_\Geb\in\dom(\mathfrak{h}_{\Geb,\mathrm{N}}^Y)
\subset\dom(\mathfrak{h}_{\Geb,\mathrm{N}})$, where the form
$\mathfrak{h}_{\Geb,\mathrm{N}}^Y$ is defined by putting the potential
$(V+Y_\infty^\Geb)_\Geb$ in place of $V_\Geb$ in the construction of
$\mathfrak{h}_{\Geb,\mathrm{N}}$. But
$\{\Phi\in\dom(\mathfrak{h}_{\Geb,\mathrm{N}}^Y)|
\,\Phi\in L_0^\infty(\Geb,\fdom(\Id\Gamma(1)))\}$
is a core for $\mathfrak{h}_{\Geb,\mathrm{N}}^Y$ according to Lem.~\ref{lemwiesel1}.
Taking also Lem.~\ref{lemvstar2}(2) and Lem.~\ref{propvjesa} into account we see that
$\sD(\Geb,\fdom(\Id\Gamma(1\vee\omega)))$ is a core for $\mathfrak{h}_{\Geb,\mathrm{N}}^Y$
as well. Since $\mathfrak{h}_{\Geb,\mathrm{N}}\le\mathfrak{h}_{\Geb,\mathrm{N}}^Y$,
it is now clear that $\Psi_\Geb$ can be approximated by elements of
$\sD(\Geb,\fdom(\Id\Gamma(1\vee\omega)))\subset\dom(\mathfrak{h}_{\Geb,\mathrm{D}})$
with respect to the form norm of $\mathfrak{h}_{\Geb,\mathrm{N}}$, that is,
\begin{align}\label{wiesel2001}
\Psi_\Geb\in\dom(\mathfrak{h}_{\Geb,\mathrm{D}})\quad\text{and}\quad
\mathfrak{h}_{\Geb,\mathrm{D}}[\Psi_\Geb]&=\mathfrak{h}_{\Geb,\mathrm{N}}[\Psi_\Geb].
\end{align}
In particular, $\Psi\in\dom(\mathfrak{h}_{\Geb,\mathrm{D}}')$.
Since $\Psi=0$ a.e. on $\Geb^c$, we further know that $v_{\RR^\nu,j}^*\Psi=0$ a.e. on $\Geb^c$,
for every $j\in\{1,\ldots,\nu\}$; see \cite[Lem.~3.4]{Matte2017}. Employing \eqref{wiesel1700}, 
\eqref{wiesel2000}, and \eqref{wiesel2001} we conclude that
$\mathfrak{h}_{\Geb,\mathrm{D}}'[\Psi]=\mathfrak{h}_{\Geb,\mathrm{N}}[\Psi_\Geb]
=\mathfrak{h}_{\RR^\nu}[\Psi]$.

Since $Y_\infty^{\Geb}$ is locally bounded on $\Geb$, it is also
clear that $\sD(\Geb,\fdom(\Id\Gamma(1\vee\omega)))'
\subset\dom(\mathfrak{h}_{\RR^\nu}^{1,\infty})$. Furthermore,
by the definition of $\mathfrak{h}_{\Geb,\mathrm{D}}$ and \eqref{wiesel1700}, 
$\sD(\Geb,\fdom(\Id\Gamma(1\vee\omega)))'$ is a core for the form
$\mathfrak{h}_{\Geb,\mathrm{D}}'$. This reveals that the closure
of $\dom(\mathfrak{h}_{\RR^\nu}^{1,\infty})$ with respect to the form norm of 
$\mathfrak{h}_{\Geb,\mathrm{D}}'$ is $\dom(\mathfrak{h}_{\Geb,\mathrm{D}}')$.
\end{proof}

%%%%%%%%%%%%%%%%%%%%%%%%%%%%%%%%%%%%%%%%%%
%%%%%%%%%%%%%%%%%%%%%%%%%%%%%%%%%%%%%%%%%%
%%%%%%%%%%%%%%%%%%%%%%%%%%%%%%%%%%%%%%%%%%

\section{A diamagnetic inequality for resolvents}\label{secmoredia}

\noindent
The purpose of this section is to derive a diamagnetic inequality comparing resolvents
of Dirichlet-Pauli-Fierz operators and resolvents of Dirichlet-Schr\"{o}dinger operators;
see \eqref{diaT} in Thm.~\ref{thmdiaT} below. This inequality will be used to discuss strong 
resolvent convergence of certain sequences of Dirichlet-Pauli-Fierz operators in the succeeding 
Sect.~\ref{secstrresconv}. Even for $\Geb=\RR^\nu$ and $\V{A}=0$, Thm.~\ref{thmdiaT} relaxes 
assumptions imposed on $\V{G}$ in earlier derivations \cite{Hiroshima1996,Hiroshima1997,KMS2013} 
of the bound \eqref{diaT}. The proofs in this section follow the lines of the corresponding ones in 
\cite{HundertmarkSimon} but require additional arguments to deal with the quantized fields.

We start with a complement to Lem.~\ref{lem-abs-val}. Recall that the symbols $Z_\delta(\Psi)$ and 
$\fS_{\delta,\Psi}$ have been introduced in \eqref{defZdelta}.

\begin{lem}\label{dialoglem}
Let $\sK$ be a separable Hilbert space,
$j\in\{1,\ldots,\nu\}$, $p\in[1,\infty]$, $\delta>0$, and let $\Psi\in L_\loc^{p}(\Geb,\sK)$ have a 
weak partial derivative with respect to $x_j$ satisfying $\partial_{x_j}\Psi\in L_\loc^{p}(\Geb,\sK)$. 
Then $\fS_{\delta,\Psi}\in L^\infty(\Geb,\sK)$ has a weak partial derivative with respect to 
$x_j$ which blongs to $L_{\loc}^{p}(\Geb,\sK)$ and is given by
\begin{align}\label{dersgndelta}
\partial_{x_j}\fS_{\delta,\Psi}&=Z_\delta(\Psi)^{-1}\big(
\partial_{x_j}\Psi-\Re\SPn{{\fS_{\delta,\Psi}}}{\partial_{x_j}\Psi}_{\sK}\fS_{\delta,\Psi}\big).
\end{align}
Furthermore, let $\chi\in W^{1,2}(\Geb)$ satisfy $|\chi|\le cZ_1(\Psi)$, for some $c>0$. 
Then $\chi\fS_{\delta,\Psi}\in L^2(\Geb,\sK)$ has a weak partial derivative with respect to $x_j$
which is in $L^{2\wedge p}_\loc(\Geb,\sK)$ and given by
\begin{align}\label{dialog}
\partial_{x_j}(\chi\fS_{\delta,\Psi})&=(\partial_{x_j}\chi)\fS_{\delta,\Psi}
+\frac{\chi}{Z_\delta(\Psi)}\big(\partial_{x_j}\Psi-\Re\SPn{{\fS_{\delta,\Psi}}}{\partial_{x_j}
\Psi}_{\sK}\fS_{\delta,\Psi}\big).
\end{align}
\end{lem}

\begin{proof}
Employing \eqref{dia00delta} and the usual chain rule for weak partial derivatives we compute
$$
\partial_{x_j}Z_\delta(\Psi)^{-1}=-Z_\delta(\Psi)^{-2}
\Re\SPn{{\fS_{\delta,\Psi}}}{\partial_{x_j}\Psi}_{\sK}\in L^p_\loc(\Geb),
$$
which together with Thm.~\ref{LeibnizL1} yields \eqref{dersgndelta};
notice that the product $(\partial_{x_j}Z_\delta(\Psi)^{-1})\Psi$ is indeed in $L^1_\loc(\Geb,\sK)$
so that Thm.~\ref{LeibnizL1} is applicable. We read off from \eqref{dersgndelta}
that $\partial_{x_j}\fS_{\delta,\Psi}\in L_\loc^p(\Geb,\sK)$.
Finally, \eqref{dialog} follows from \eqref{dersgndelta} and Thm.~\ref{LeibnizL1}; here we use that
the product $\chi\partial_{x_j}\fS_{\delta,\Psi}$ is in $L^1_\loc(\Geb,\sK)$ thanks to
the postulated bound $|\chi|\le cZ_1(\Psi)$.
\end{proof}

\begin{prop}\label{propdia1}
Assume \eqref{introA} and \eqref{introG}.
Let $\delta>0$, $j\in\{1,\ldots,\nu\}$, $\Psi\in\dom(v_{\Geb,j}^*)$, and let 
$\chi\in W^{1,2}(\Geb)$ be nonnegative and satisfy $\chi\le cZ_1(\Psi)$, for some $c>0$.
Then $\chi\fS_{\delta,\Psi}\in\dom(v_{\Geb,j}^*)$ and, a.e. on $\Geb$,
\begin{align}\label{diacharly}
\Re\SPn{v_{\Geb,j}^*(\chi\fS_{\delta,\Psi})}{v_{\Geb,j}^*\Psi}_{\sF}
&\ge\|\fS_{\delta,\Psi}\|_{\sF}(\partial_{x_j}\chi)\partial_{x_j}\|\Psi\|_{\sF}.
\end{align}
\end{prop}

\begin{proof}
We pick some $\ve>0$ and start by considering 
$\Psi_\ve=N_\ve^\mh\Psi$ with $N_\ve$ given by \eqref{defNve}.
According to Rem.~\ref{remapprox1}, $\Psi_\ve$ has a weak partial derivative with respect to $x_j$
which is given by \eqref{partialPsive}.
Plugging $\Psi_\ve$ and $\chi$ into \eqref{dialog}, subtracting 
$i\chi A_j\fS_{\delta,\Psi_\ve}+i\chi\vp(G_j)\fS_{\delta,\Psi_\ve}$ on both sides, and using 
$\Re\SPn{\fS_{\delta,\Psi_\ve}}{i(A_j+\vp(G_j))\Psi_\ve}_{\sF}=0$ and \eqref{partialPsive}, we find
\begin{align}\nonumber
\partial_{x_j}&(\chi\fS_{\delta,\Psi_\ve})-iA_j(\chi\fS_{\delta,\Psi_\ve})
-i\vp(G_j)(\chi\fS_{\delta,\Psi_\ve})
\\\nonumber
&=(\partial_{x_j}\chi)\fS_{\delta,\Psi_\ve}+\frac{\chi}{Z_\delta(\Psi_\ve)}
\Big(iN_\ve^\mh v_{\Geb,j}^*\Psi+iC_\ve(G_j)^*\Psi
\\\label{charly1}
&\quad\qquad-\Re\SPn{\fS_{\delta,\Psi_\ve}}{iN_\ve^\mh v_{\Geb,j}^*\Psi
+iC_\ve(G_j)^*\Psi}_{\sF}\fS_{\delta,\Psi_\ve}\Big)\quad\text{in $L^1_\loc(\Geb,\sF)$.}
\end{align}
Next, we compute the $\sF$-scalar product of $\eta\in\sD(\Geb,\fdom(\Id\Gamma(1)))$ with the
vectors on both sides of \eqref{charly1} and integrate the result over $\Geb$. After that
we pass to the limit $\ve\downarrow0$ taking into account that
\begin{enumerate}
\item[(a)] $\delta\le Z_\delta(\Psi_\ve)\to Z_\delta(\Psi)$ pointwise;
\item[(b)] $N_\ve^\mh v_{\Geb,j}^*\Psi\to v_{\Geb,j}^*\Psi$ pointwise and in $L^2(\Geb,\sF)$;
\item[(c)] by (a), (b), and $\chi\in L^2(\Geb)$,
\begin{align*}
\frac{\chi}{Z_\delta(\Psi_\ve)}N_\ve^\mh v_{\Geb,j}^*\Psi\longrightarrow
\frac{\chi}{Z_\delta(\Psi)}v_{\Geb,j}^*\Psi\quad\text{pointwise and in $L^1(\Geb,\sF)$,}
\end{align*}
with integrable majorant $|\chi|\|v_{\Geb,j}^*\Psi\|_\sF/\delta$;
\item[(d)] in view of $\chi\in L^2(\Geb)$, $\|G_j\|_{\HP}\in L^2_\loc(\Geb)$, and \eqref{convCve2},
\begin{align*}
\frac{|\chi|}{Z_\delta(\Psi_\ve)}\|C_\ve(G_j)^*\Psi\|_\sF
&\le 2\ve^\eh \|G_j\|_{\HP}|\chi|\frac{\|\Psi_\ve\|_\sF}{Z_\delta(\Psi_\ve)}
\longrightarrow0,
\end{align*}
where the convergence is understood in $L_\loc^1(\Geb)$;
\item[(e)] $\fS_{\delta,\Psi_\ve}\to\fS_{\delta,\Psi}$ pointwise with
$\|\fS_{\delta,\Psi_\ve}\|_\sF\le1$.
\end{enumerate} 
In this way we arrive at the identity
\begin{align*}
i\SPn{v_{\Geb,j}}{\chi\fS_{\delta,\Psi}}
&=\int_\Geb\SPb{\eta(\V{x})}{(\partial_{x_j}\chi)(\V{x})\fS_{\delta,\Psi}(\V{x})}\Id\V{x}
\\
&\quad+\int_\Geb\SPB{\eta(\V{x})}{
\frac{\chi(\V{x})}{Z_\delta(\Psi(\V{x}))}\Big(iv_{\Geb,j}^*\Psi
-\Re\SPn{\fS_{\delta,\Psi}}{iv_{\Geb,j}^*\Psi}_{\sF}\fS_{\delta,\Psi}\Big)(\V{x})}\Id\V{x}.
\end{align*}
Since we are assuming that $\partial_{x_j}\chi\in L^2(\Geb)$ and $|\chi|\le cZ_1(\Psi)$,
the last two integrals can be read as scalar products of $\eta$ with vectors in $L^2(\Geb,\sF)$. 
Thus, $\chi\fS_{\delta,\Psi}\in\dom(v_{\Geb,j}^*)$ with
\begin{align}\label{hassan2}
iv_{\Geb,j}^*(\chi\fS_{\delta,\Psi})&=(\partial_{x_j}\chi)\fS_{\delta,\Psi}+\frac{\chi}{Z_\delta(\Psi)}
\Big(iv_{\Geb,j}^*\Psi-\Re\SPn{\fS_{\delta,\Psi}}{iv_{\Geb,j}^*\Psi}_{\sF}\fS_{\delta,\Psi}\Big).
\end{align}

From here on we can copy the proof of \cite[Lem.~3.1]{HundertmarkSimon}: Computing the 
$\sF$-scalar product with $iv_{\Geb,j}^*\Psi$ on both sides of \eqref{hassan2}
and taking real parts we arrive at
\begin{align*}
\Re&\SPn{v_{\Geb,j}^*(\chi\fS_{\delta,\Psi})}{v_{\Geb,j}^*\Psi}_{\sF}
\\
&=(\partial_{x_j}\chi)\Re\SPn{\fS_{\delta,\Psi}}{iv_{\Geb,j}^*\Psi}_{\sF}+\frac{\chi}{Z_\delta(\Psi)}
\Big(\|v_{\Geb,j}^*\Psi\|_{\sF}^2-\big(\Re\SPn{\fS_{\delta,\Psi}}{iv_{\Geb,j}^*\Psi}_{\sF}\big)^2\Big)
\\
&\ge(\partial_{x_j}\chi)\Re\SPn{\fS_{\delta,\Psi}}{iv_{\Geb,j}^*\Psi}_{\sF}
=\|\fS_{\delta,\Psi}\|_{\sF}(\partial_{x_j}\chi)\partial_{x_j}\|\Psi\|_{\sF}.
\end{align*}
Here we also used $\chi\ge0$ in the penultimate step and \eqref{dia0} in the last one.
\end{proof}

Now we are in a position to prove the promised diamagnetic inequality for resolvents.
Recall that the Dirchlet-Pauli-Fierz operator $H_{\Geb}$ has been defined in 
Subsect.~\ref{ssecDPFform}. By $S_\Geb$ we denote the Dirichlet-Schr\"{o}dinger
operator with potential $V$ on $\Geb$, i.e., the selfadjoint operator representing the
nonnegative closed form 
\begin{align*}
\mathfrak{s}_{\Geb,\mathrm{D}}[f]&:=\frac{1}{2}\|\nabla f\|^2
+\int_{\Geb}V(\V{x})|f(\V{x})|^2\Id\V{x},\quad
f\in\dom(\mathfrak{s}_{\Geb,\mathrm{D}}):=\mr{W}^{1,2}(\Geb)\cap\fdom(V).
\end{align*}

\begin{thm}\label{thmdiaT}
Assume \eqref{Vlocint}, \eqref{introA}, and \eqref{introG}. Let
$\Phi\in L^2(\Lambda,\sF)$ and $E>0$. Then, a.e. on $\Geb$, 
\begin{align}\label{diaT}
\|(H_\Geb+E)^{-1}\Phi\|_{\sF}&\le(S_{\Geb}+E)^{-1}\|\Phi\|_{\sF}.
\end{align}
\end{thm}

\begin{proof}
We can adapt the proof of \cite[Thm.~3.3]{HundertmarkSimon}.
Put $\Psi:=(H_\Geb+E)^{-1}\Phi\in\dom(H_\Geb)\subset\dom(\mathfrak{h}_{\Geb,\mathrm{D}})$. 
Then \cite[Cor.~4.1]{Matte2017} implies $\|\Psi\|_{\sF}\in\dom(\mathfrak{s}_{\Geb,\mathrm{D}})$.
Pick some $\delta>0$ and let $\chi\in\dom(\mathfrak{s}_{\Geb,\mathrm{D}})\subset W^{1,2}(\Geb)$ 
be nonnegative, compactly supported, and bounded. Employing Prop.~\ref{propdia1} we then infer
that $\chi\fS_{\delta,\Psi}\in\dom(\mathfrak{h}_{\Geb,\mathrm{N}})$. 
Since $\chi$ is compactly supported and bounded, Lem.~\ref{exdomNbdcdomD} now implies that 
actually $\chi\fS_{\delta,\Psi}\in\dom(\mathfrak{h}_{\Geb,\mathrm{D}})$.
Integrating \eqref{diacharly}, summing the result over $j\in\{1,\ldots,\nu\}$, and observing
\begin{align*}
\SPn{\fS_{\delta,\Psi}}{\Psi}_{\sF}&=\|\fS_{\delta,\Psi}\|_{\sF}\|\Psi\|_{\sF},\quad
\int_\Geb\SPn{\Id\Gamma(\omega)^\eh\chi(\V{x})\fS_{\delta,\Psi}(\V{x})}{
\Id\Gamma(\omega)^\eh\Psi(\V{x})}_\sF\Id\V{x}\ge0,
\end{align*}  
we further find
\begin{align*}
\frac{1}{2}\SPb{\|\fS_{\delta,\Psi}\|_{\sF}\nabla\chi}{\nabla\|\Psi\|_{\sF}}_{L^2(\Geb)}
+\int_\Geb(V(\V{x})+E)\chi(\V{x})\|\fS_{\delta,\Psi}(\V{x})\|_{\sF}\|\Psi(\V{x})\|_{\sF}\Id\V{x}&
\\
\le\big|(\mathfrak{h}_{\Geb,\mathrm{D}}+E)[\chi\fS_{\delta,\Psi},\Psi]\big|
=|\SPn{\chi\fS_{\delta,\Psi}}{\Phi}|\le\SPn{\chi}{\|\Phi\|_{\sF}}_{L^2(\Geb)}&.
\end{align*}
Here we also used $\chi\ge0$ and $\|\fS_{\delta,\Psi}\|_{\sF}\le1$ in the last step.
(Furthermore, symbols like $\mathfrak{q}[\phi,\psi]$ denote values of the sesquilinear form
associated with a quadratic form $\mathfrak{q}$.)
By dominated convergence, we can pass to the limit $\delta\downarrow0$ on the left hand
side of the previous estimation. Since $\nabla\|\Psi\|_{\sF}=0$ a.e. on $\{\Psi=0\}$, we may
drop the term $\|\fS_{\Psi}\|_\sF$ found in this way whenever it is multiplied with $\|\Psi\|_{\sF}$ or
$\nabla\|\Psi\|_{\sF}$. This yields
\begin{align}\label{josef1}
(\mathfrak{s}_{\Geb,\mathrm{D}}+E)[\chi,\|\Psi\|_{\sF}]&\le\SPn{\chi}{\|\Phi\|_{\sF}}_{L^2(\Geb)}.
\end{align}
The bound \eqref{josef1} is actually available for all nonnegative
$\chi\in\dom(\mathfrak{s}_{\Geb,\mathrm{D}})$ since any such $\chi$ can be approximated 
in the form norm of $\mathfrak{s}_{\Geb,\mathrm{D}}$ by bounded and compactly supported
nonnegative elements of $W^{1,2}(\Geb)$ (using \cite[Cor.~6.18]{LiebLoss2001}).
In particular, we may choose $\chi:=(S_{\Geb}+E)^{-1}\eta$, for some 
$\eta\in L^2(\Geb)$ with $\eta\ge0$, because 
$\dom(S_{\Geb})\subset\dom(\mathfrak{s}_{\Geb,\mathrm{D}})$ and the resolvent
$(S_{\Geb,\mathrm{D}}+E)^{-1}$ is positivity preserving. This yields \eqref{diaT} integrated
with respect to the density $\eta$.
\end{proof}

%%%%%%%%%%%%%%%%%%%%%%%%%%%%%%%%%%%%%%%%%%%%%%%%%
%%%%%%%%%%%%%%%%%%%%%%%%%%%%%%%%%%%%%%%%%%%%%%%%%
%%%%%%%%%%%%%%%%%%%%%%%%%%%%%%%%%%%%%%%%%%%%%%%%%

\section{Strong resolvent convergence}\label{secstrresconv}

\noindent
In the presence of singular electromagnetic fields, a Feynman-Kac formula is typically 
obtained in a chain of extension steps establishing the formula for ever more singular (vector)
potentials. To ensure convergence of the functional analytic side of the Feynman-Kac formula, at least
along suitable subsequences,
when singular (vector) potentials are approximated by more regular ones, it is sufficient to
prove strong resolvent convergence of the corresponding selfadjoint operators. For our model this is
done in the present section. Since the approximation of electrostatic potentials is quite standard, we
shall concentrate on the simultaneous approximation of classical and quantized vector potentials here.

Results for Schr\"{o}dinger operators similar to Thm.~\ref{thmstrresconv} below appear in
\cite{Kato1978} for $\Geb=\RR^\nu$ and in \cite{LiskevichManavi1997} for general open $\Geb$.
Both the limiting vector potential and the ones approximating it are merely supposed
to be in $L_\loc^2$ in \cite{Kato1978}. In \cite{LiskevichManavi1997} results for even
more general vector potentials can be found. We shall restrict our
attention to the situation we actually encounter later on as this admits a comparatively short proof.

In the next theorem and henceforth $C_b^\ell$ stands for bounded, $\ell$-times continuously 
differentiable maps with bounded derivatives of order $\le\ell$. Recalling \eqref{defhRRnu} we
further abbreviate $\mathfrak{h}:=\mathfrak{h}_{\RR^\nu}$ and refer to the selfadjoint operator
$H:=H_{\RR^\nu}$ representing this form simply as the {\em Pauli-Fierz operator on $\RR^\nu$}.

\begin{thm}\label{thmstrresconv}
Let $\V{A}\in L^2_\loc(\RR^\nu,\RR^\nu)$, $\V{G}\in L_\loc^2(\RR^\nu,\HP^\nu)$ and
$\V{A}^n\in C^1_b(\RR^\nu,\RR^\nu)$, $\V{G}^n\in C^1_b(\RR^\nu,\HP^\nu)$, 
$n\in\NN$, satisfy
\begin{align}\label{tosio0A}
\int_K|\V{A}^n(\V{x})-\V{A}(\V{x})|^2\Id\V{x}&\xrightarrow{\;\;n\to\infty\;\;}0,
\\\label{tosio0G}
\int_K\|\V{G}^n_{\V{x}}-\V{G}_{\V{x}}\|_{\HP}^2\Id\V{x}
&\xrightarrow{\;\;n\to\infty\;\;}0,
\end{align}
for all compact $K\subset\RR^\nu$. Assume that $V\ge0$ is measurable and bounded. Let $H$ be the 
Pauli-Fierz operator on $\RR^\nu$ defined by means of $\V{A}$, $\V{G}$, and $V$. 
For every $n\in\NN$, let $H^n$ be the Pauli-Fierz operator on $\RR^\nu$ defined by means 
of $\V{A}^n$, $\V{G}^n$, and $V$. Then
\begin{align*}
H^n\xrightarrow{\;\;n\to\infty\;\;}H\quad\text{in the strong resolvent sense.}
\end{align*}
\end{thm}

\begin{proof}
Recall that, for each $z\in\CC\setminus\RR$, strong convergence of $(H^n-z)^{-1}$ to $(H-z)^{-1}$
is implied by weak convergence of $(H^n-z)^{-1}$ to $(H-z)^{-1}$, because
\begin{align*}
&\|(H^n-z)^{-1}\Psi'\|^2-\|(H-z)^{-1}\Psi'\|^2
\\
&=\frac{1}{\Im[z]}\Im\SPn{\Psi'}{((H^n-z)^{-1}-(H-z)^{-1})\Psi'},\quad\Psi'\in L^2(\RR^\nu,\sF),
\end{align*}
by the first resolvent equation, and because weak convergence of a sequence in a Hilbert space
together with convergence of the norms of its elements implies norm convergence. 
In what follows we pick some $z\in\CC\setminus\RR$ with $\Re[z]\le-1$. 
Since $H_n\ge0$, the resolvents $(H^n-z)^{-1}$
 are uniformly bounded in $n\in\NN$, whence it suffices to show that
 \begin{align*}
 \SPn{\Xi}{((H^n-z)^{-1}-(H-z)^{-1})\Phi}\xrightarrow{\;\;n\to\infty\;\;}0,
 \end{align*}
 for all $\Xi$ and $\Phi$ in some dense subsets of $L^2(\RR^\nu,\sF)$. We pick
 $$
 \Xi:=(H-\bar{z})(H-\Re[z])^{-1}\wt{\Xi},\quad\text{for some}\quad
 \wt{\Xi}\in L^2(\RR^\nu,\sF)\cap L^\infty(\RR^\nu,\sF),
 $$
 noticing that $(H-\bar{z})(H-\Re[z])^{-1}$ is a bounded isomorphism on $L^2(\RR^\nu,\sF)$ which in 
 particular maps a dense subset onto another dense subset. In view of the diamagnetic inequality
\eqref{diaT}  this choice of $\Xi$ implies that
 \begin{align}\label{Xibd}
\Upsilon&:= (H-\bar{z})^{-1}\Xi=(H-\Re[z])^{-1}\wt{\Xi}\quad\text{is bounded.}
 \end{align}
 Furthermore, we know from \cite[Thm.~5.5]{Matte2017} that every $H^n$ with $n\in\NN$ is 
 essentially selfadjoint on $\sD(\RR^\nu,\dom(\Id\Gamma(1\vee\omega)))$. 
 (Here we use that the Schr\"{o}dinger operator $(1/2)(-i\nabla-\V{A}^n)^2+V$ is essentially
 selfadjoint on $C_0^\infty(\RR^\nu)$, exploiting that we work on the whole
 Euclidean space $\RR^\nu$ and not on a proper open subset of it.) In particular, 
 $(H^n-z)\sD(\RR^\nu,\dom(\Id\Gamma(1\vee\omega)))$ is a dense subspace of
 $L^2(\RR^\nu,\sF)$, and we choose $\Phi:=(H^n-z)\Psi$ for some
 $\Psi\in\sD(\RR^\nu,\dom(\Id\Gamma(1\vee\omega)))$. Then
\begin{align*}
\SPn{\Xi}{((H^n-z)^{-1}-(H-z)^{-1})\Phi}&=\SPn{\Xi}{\Psi}-\SPn{\Upsilon}{(H^n-z)\Psi}
 =\mathfrak{h}[\Upsilon,\Psi]-\SPn{\Upsilon}{H^n\Psi}.
 \end{align*}
 For all $j\in\{1,\ldots,\nu\}$ and $n\in\NN$, we now abbreviate
  \begin{align*}
 \tilde{A}_j^n&:={A}_j^n-{A}_j,\qquad\tilde{{G}}_j^n:={G}_j^n-{G}_j,
 \\
v_{j}^n&:=(-i\partial_{x_j}-A_j^n-\vp(G_j^n))\restr_{\sD(\RR^\nu,\fdom(\Id\Gamma(1)))},
\qquad v_j:=v_{\RR^\nu,j}.
\end{align*}
Furthermore, we pick $\Upsilon_m\in\sD(\RR^\nu,\fdom(\Id\Gamma(1\vee\omega)))$, $m\in\NN$,
such that $\Upsilon_m\to\Upsilon$, $m\to\infty$,
with respect to the form norm of $\mathfrak{h}$, which is possible
because $\Upsilon\in\dom(H)\subset\dom(\mathfrak{h})$. Then we obtain
\begin{align*}
&\mathfrak{h}[\Upsilon,\Psi]-\SPn{\Upsilon}{H^n\Psi}
\\
&=\lim_{m\to\infty}\mathfrak{h}[\Upsilon_m,\Psi]
-\lim_{m\to\infty}\SPn{\Upsilon_m}{H^n\Psi}
\\
&=\lim_{m\to\infty}\frac{1}{2}\sum_{j=1}^\nu\Big\{\SPn{v_{j}\Upsilon_m}{v_{j}\Psi}
-\SPn{v_{j}^n\Upsilon_m}{v_{j}^n\Psi}\Big\}
\\
&=\lim_{m\to\infty}\frac{1}{2}\sum_{j=1}^\nu\Big\{\SPb{v_{j}\Upsilon_m}{
(\tilde{A}_j^n+\vp(\tilde{G}_j^n))\Psi}
+\SPb{(\tilde{A}_j^n+\vp(\tilde{G}_j^n))\Upsilon_m}{v_{j}^n\Psi}\Big\}.
\end{align*}
Next, we take into account that convergence of $\Upsilon_m$ with respect to the form norm of
$\mathfrak{h}$ entails the convergences $v_{j}\Upsilon_m\to v_{j}^*\Upsilon$, 
for all $j\in\{1,\ldots,\nu\}$. On account of \eqref{rbvp} and \eqref{rbvpvp} we also know that 
$\tilde{A}_j^nv_{j}^n\Psi$ and $\vp(\tilde{G}_j^n)v_{j}^n\Psi$ belong to $L^2(\RR^\nu,\sF)$, 
for all $n\in\NN$. We thus arrive at
 \begin{align}\nonumber
\mathfrak{h}[\Upsilon,\Psi]-\SPn{\Upsilon}{H^n\Psi}
&=\frac{1}{2}\sum_{j=1}^\nu\SPn{v_{j}^*\Upsilon}{(\tilde{A}_j^n+\vp(\tilde{G}_j^n))\Psi}
\\\nonumber
&\quad+\frac{1}{2}\sum_{j=1}^\nu\int_{\RR^\nu}\SPn{\Upsilon(\V{x})}{
\tilde{A}_j^n(\V{x})(v_{j}^n\Psi)(\V{x})}_\sF\Id\V{x}
\\\label{ulli1}
&\quad+\frac{1}{2}\sum_{j=1}^\nu\int_{\RR^\nu}\SPn{\Upsilon(\V{x})}{\vp(\tilde{G}_{j,\V{x}}^n)
(v_{j}^n\Psi)(\V{x})}_\sF\Id\V{x},
 \end{align}
 for every $n\in\NN$. Here $(\tilde{A}_j^n+\vp(\tilde{G}_j^n))\Psi\to0$ in $L^2(\RR^\nu,\sF)$ 
 because of \eqref{rbvp} and \eqref{tosio0G}, since $\|\Psi\|_\sF$ and 
 $\|\Psi\|_{\fdom(\Id\Gamma(1))}$ 
 are compactly supported and bounded. Hence, the first term on the right hand side of \eqref{ulli1}
 goes to zero as $n\to\infty$. Next, we observe  (using \eqref{rbvp}, \eqref{tosio0A}, and 
 \eqref{tosio0G}) that the vectors $v_{j}^n\Psi$ are supported in $\supp(\Psi)$ and uniformly bounded in 
 $L^2(\RR^\nu,\sF)$. Together with \eqref{tosio0A} and \eqref{Xibd} this shows that the term in the 
 second line of \eqref{ulli1} converges to zero as well. Furthermore, setting
 \begin{align*}
 D_{j,\V{x}}^n\psi:=N_1^\eh\vp(G_{j,\V{x}}^n)N_1^\mh\psi-\vp(G_{j,\V{x}}^n)\psi,
 \quad\psi\in\fdom(\Id\Gamma(1)),
 \end{align*}
 where $N_1=1+\Id\Gamma(1)$ as in \eqref{defNve}, we obtain
 \begin{align*}
&\|\vp(\tilde{G}_{j,\V{x}}^n)(v_{j}^n\Psi)(\V{x})\|_\sF
\\
&\le2\|\tilde{G}_{j,\V{x}}^n\|_{\HP}\big(\|\partial_{x_j}\Psi(\V{x})\|_{\fdom(\Id\Gamma(1))}
+|{A}^n_{j}(\V{x})|\|\Psi(\V{x})\|_{\fdom(\Id\Gamma(1))}\big)
\\
&\quad+2\|\tilde{G}_{j,\V{x}}^n\|_{\HP}
\big(\|\vp(G_{j,\V{x}}^n)N_1^\eh\Psi(\V{x})\|_{\sF}+\|D_{j,\V{x}}^nN_1^\eh\Psi(\V{x})\|_{\sF}\big),
 \end{align*}
 for all $\V{x}\in\RR^\nu$ and $n\in\NN$.
 According to \cite[Lem.~2.9(2)]{Matte2017} (applied with dispersion relation $1$),  the operator 
 $D_{j,\V{x}}^n$ is indeed well-defined on its dense domain $\fdom(\Id\Gamma(1))$, and it is
 bounded with $\|D_{j,\V{x}}^n\|\le2\|G_{j,\V{x}}^n\|_{\HP}$. Notice that $D_{j,\V{x}}^n$
 can be applied to $N_1^\eh\Psi(\V{x})$ since $\Psi(\V{x})\in\dom(\Id\Gamma(1))$. Also taking into
 account that
 $$
 \|\vp(G_{j,\V{x}}^n)N_1^\eh\Psi(\V{x})\|_{\sF}
 \le2\|G_{j,\V{x}}^n\|_{\HP}\|\Psi(\V{x})\|_{\dom(\Id\Gamma(1))},
 $$ 
 where $\|\Psi\|_{\dom(\Id\Gamma(1))}$ is bounded,
 we find a $\Psi$-dependent constant $C_\Psi>0$ such that
 \begin{align*}
 \|\vp(\tilde{G}_{j,\V{x}}^n)(v_{j}^n\Psi)(\V{x})\|_\sF
 &\le C_{\Psi}1_{\supp(\Psi)}(\V{x})\|\tilde{G}_{j,\V{x}}^n\|_{\HP}
 \big(1+|{A}^n_{j}(\V{x})|+\|G_{j,\V{x}}^n\|_{\HP}\big),
 \end{align*}
  for all $\V{x}\in\RR^\nu$ and $n\in\NN$.
 In conjunction with \eqref{tosio0A}, \eqref{tosio0G}, and \eqref{Xibd} this finally proves 
convergence to zero of the term in the third line of \eqref{ulli1}.
\end{proof}

%%%%%%%%%%%%%%%%%%%%%%%%%%%%%%%%%%%%%%%%%%%%%%%%%
%%%%%%%%%%%%%%%%%%%%%%%%%%%%%%%%%%%%%%%%%%%%%%%%%
%%%%%%%%%%%%%%%%%%%%%%%%%%%%%%%%%%%%%%%%%%%%%%%%%

\section{Stochastic analysis for regular coefficients}\label{secstochana}

\noindent
Our objective in this section is to find Feynman-Kac formulas for the Pauli-Fierz operator on
$\RR^\nu$ with regular coefficients, more
precisely, coefficients satisfying the hypotheses collected in Subsect.~\ref{ssechypreg}. The main
tools will be a stochastic differential equation (\eqref{SDEW} below)
associated with the Pauli-Fierz model investigated in
\cite{GMM2017} and various results of the latter paper on the random functions 
$W_t(\V{x})$ and $W_t(\V{x},\V{y})$ for $\V{A}=0$. 
Before we can apply the findings of \cite{GMM2017} and extend them to non-zero $\V{A}$, we have,
however, to compare the formulas given in the introduction for $S_t(\V{x})$, $K_t(\V{x})$, 
$S_t(\V{x},\V{y})$, and
$K_t(\V{x},\V{y})$ with more familiar expressions for Stratonovich type stochastic integrals.
This is done in a discussion of the Feynman-Kac integrands in Subsect.~\ref{ssecFKintreg}, after
a more detailed explanation of the involved Brownian bridge processes and time reversed processes 
in Subsect.~\ref{ssecBMBB}. Finally, we verify in Subsect.~\ref{ssecTTxyreg} that the probabilistic sides
of the Feynman-Kac formula define a strongly continuous semigroup of bounded selfadjoint operators,
whose generator is identified as the Pauli-Fierz operator on $\RR^\nu$ in Subsect.~\ref{ssecFKRRnureg}.

%%%%%%%%%%%%%%%%%%%%%%%%%%%%%%%%%%%%%%%%%%%%%%%%%

\subsection{Assumptions on the coefficients used throughout Sect.~\ref{secstochana}}\label{ssechypreg}

\noindent
In the entire Sect.~\ref{secstochana} we assume 
\begin{align}\label{AVreg}
\V{A}\in C_b^1(\RR^\nu,\RR^\nu),\quad V\in C_b(\RR^\nu,\RR),\quad V\ge0.
\end{align}
Here $C_b:=C_b^0$, and the notation $C_b^\ell$ has been explained prior to Thm.~\ref{thmstrresconv}.
Throughout this section we further assume $\V{G}$ to fulfill the following two hypotheses:

\begin{hyp}\label{hypGreg}
$\V{G}\in C^2(\RR^\nu,\HP^\nu)$,
the components of $\V{G}_{\V{x}}$ and
$\partial_{x_1}\V{G}_{\V{x}},\ldots,\partial_{x_\nu}\V{G}_{\V{x}}$
are elements of $\fdom(\omega^{-1}+\omega^2)$, for every $\V{x}\in\RR^\nu$, 
and the following map is continuous and bounded,
$$
\RR^\nu\ni\V{x}\longmapsto
(\V{G}_{\V{x}},\partial_{x_1}\V{G}_{\V{x}},\ldots,\partial_{x_\nu}\V{G}_{\V{x}})
\in\fdom(\omega^{-1}+\omega^2)^{\nu(\nu+1)}.
$$
\end{hyp}

\begin{hyp}\label{hypGC}
There exists a completely real subspace $\HP_{\RR}\subset\HP$ such that
\begin{align*}
\V{G}_{\V{x}}\in\HP_{\RR}^\nu,\quad e^{-t\omega}\HP_{\RR}\subset\HP_{\RR},
\end{align*}
for all $\V{x}\in\RR^\nu$ and $t>0$.
\end{hyp}

These two hypotheses have been imposed on $\V{G}$ in \cite{GMM2017}. The second one,
Hyp.~\ref{hypGC}, leads to some crucial cancellations in the analysis of Feynman-Kac
integrands and their associated stochastic differential equations in \cite{GMM2017}; it will not be used
in a directly visible way in the present article.

%%%%%%%%%%%%%%%%%%%%%%%%%%%%%%%%%%%%%%%%%%%%%%%%

\subsection{Notation for Brownian bridges and time reversed processes}\label{ssecBMBB}

\noindent
Recall that we fixed the filtered probability space $(\Omega,\fF,(\fF_t)_{t\ge0},\PP)$ satisfying the usual assumptions
and the $(\fF_t)_{t\ge0}$-Brownian motion $\V{B}$ starting in $0$ in the introduction.

Let $t>0$ in what follows. If $\V{x}\in\RR^\nu$ and $\V{q}:\Omega\to\RR^\nu$ is
$\fF_0$-measurable, then we let $\V{b}^{t;\V{q},\V{x}}$ denote a choice of the up to 
indistinguishability unique continuous semimartingale with respect to $(\fF_s)_{s\in[0,t]}$
which $\PP$-a.s. solves the stochastic differential equation for a Brownian bridge in time $t$
starting at $\V{q}$ and ending at $\V{x}$, i.e.,
\begin{align}\label{SDEbridgeqx}
\V{b}_s&=\V{q}+\V{B}_s+\int_0^s\frac{\V{x}-\V{b}_r}{t-r}\Id r,\quad s\in[0,t),
\quad\V{b}_t=\V{x}.
\end{align}

Next, we explain some notation for time reversals of Brownian motions and bridges;
see \cite{HaussmannPardoux1986,PardouxLNM1204} and \cite[App.~4]{GMM2017} for more details.

We denote by $(\check{\fF}_s)_{s\ge0}$ the standard extension of the filtration 
$(\fH_s)_{s\ge0}$ where, for all $s\in[0,t]$, $\fH_s$ denotes the $\sigma$-algebra generated by 
$\V{B}_{t-s}$ and all increments $\V{B}_{t}-\V{B}_{t-r}$ with $r\in[0,s]$, and where
$\fH_{s}=\fH_t$ for all $s\ge t$. Let $\V{x}\in\RR^\nu$. 
Then the reversed process $\V{B}^{t;\V{x}}$ defined in \eqref{BMrev} is a semimartingale with respect to
$(\check{\fF}_s)_{s\in[0,t]}$. Furthermore, there exists a $(\check{\fF}_s)_{s\in[0,t]}$-Brownian 
motion $\check{\V{B}}$ such that  $\V{B}^{t;\V{x}}$ is $\PP$-a.s. a solution to
\begin{align}\label{SDEbridgebar}
\V{b}_s&=\check{\V{q}}+\check{\V{B}}_s+\int_0^s\frac{\V{x}-\V{b}_r}{t-r}\Id r,
\quad s\in[0,t),\quad\V{b}_t=\V{x},
\end{align}
provided that we choose the $\check{\fF}_0$-measurable initial condition 
$\check{\V{q}}=\V{B}_t^{\V{x}}$.
\begin{align}\label{defbarbridge}
\text{We denote by $\smash{\check{\V{b}}}^{t;\V{y},\V{x}}$ the solution of \eqref{SDEbridgebar} 
for the choice $\check{\V{q}}=\V{y}$.}
\end{align}

We further denote by $(\hat{\fF}_s)_{s\ge0}$ the standard extension 
of the filtration $(\fJ_s)_{s\ge0}$ where, for all $s\in[0,t]$, $\fJ_s$ denotes the $\sigma$-algebra 
generated by $\V{b}_{t-s}^{t;\V{y},\V{x}}$ and all increments $\V{B}_{t}-\V{B}_{t-r}$ with 
$r\in[0,s]$, and where $\fJ_{s}=\fJ_t$ for all $s\ge t$. Then the reversed process
$\smash{\hat{\V{b}}}^{t;\V{x},\V{y}}$ defined in \eqref{BBrev}
 is a semimartingale with respect to $(\hat{\fF}_s)_{s\in[0,t]}$, and there exists a 
 $(\hat{\fF}_s)_{s\in[0,t]}$-Brownian motion $\hat{\V{B}}$ such that 
 $\smash{\hat{\V{b}}}^{t;\V{x},\V{y}}$ is $\PP$-a.s. a solution to
\begin{align*}%\label{SDEbridgehat}
\V{b}_s&=\V{x}+\hat{\V{B}}_s+\int_0^s\frac{\V{y}-\V{b}_r}{t-r}\Id r,
\quad s\in[0,t),\quad\V{b}_t=\V{y}.
\end{align*}

%%%%%%%%%%%%%%%%%%%%%%%%%%%%%%%%%%%%%%%%%%%%%%%%

\subsection{The Feynman-Kac integrands for regular coefficients}\label{ssecFKintreg}

\noindent
To benefit from the results of \cite{GMM2017}, we first have to verify that the formulas 
\eqref{defSxintro}, \eqref{defKxintro}, \eqref{defSxyintro}, and \eqref{defKxyintro} for the Stratonovich 
type integrals in our Feynman-Kac integrands generalize the ones used in the latter article:

\begin{lem}\label{lemdiv}
Let $t>0$ and $\V{x},\V{y}\in\RR^\nu$. Then the following identities hold $\PP$-a.s.,
\begin{align}\label{Sdiv}
S_t(\V{x})&=\int_0^tV(\V{B}_s^{\V{x}})\Id s-i\int_0^t\V{A}(\V{B}_s^{\V{x}})\Id\V{B}_s
-\frac{i}{2}\int_0^t\Div\V{A}(\V{B}_s^{\V{x}})\Id s,
\\\label{Kdiv}
K_t(\V{x})&=\int_0^tj_s\V{G}_{\V{B}_s^{\V{x}}}\Id\V{B}_s
+\frac{1}{2}\int_0^tj_s\Div\V{G}_{\V{B}_s^{\V{x}}}\Id s,
\end{align}
as well as
\begin{align}
S_t(\V{x},\V{y})&=\int_0^tV(\V{b}_s^{t;\V{y},\V{x}})\Id s
\label{Sxydiv}
-i\int_0^t\V{A}(\V{b}_s^{t;\V{y},\V{x}})\Id\V{b}_s^{t;\V{y},\V{x}}
-\frac{i}{2}\int_0^t\Div\V{A}(\V{b}_s^{t;\V{y},\V{x}})\Id s,
\\\label{Kxydiv}
K_t(\V{x},\V{y})&=\int_0^tj_s\V{G}_{\V{b}_s^{t;\V{y},\V{x}}}\Id\V{b}_s^{t;\V{y},\V{x}}
+\frac{1}{2}\int_0^tj_s\Div\V{G}_{\V{b}_s^{t;\V{y},\V{x}}}\Id s.
\end{align}
\end{lem}

\begin{proof}
Under the present conditions on $\V{G}$, well-known results on Hilbert space-valued stochastic
integrals reveal that
\begin{align*}
\int_0^tj_s\V{G}_{\V{B}_s^{\V{x}}}\Id\V{B}_s^{\V{x}}&=
\underset{n\to\infty}{\mathrm{lim\,prob}}
\sum_{\ell=1}^nj_{(\ell-1)t/n}\V{G}(\V{B}^{\V{x}}_{(\ell-1)t/n})(\V{B}_{\ell t/n}-\V{B}_{(\ell-1)t/n}),
\\
\int_{0}^tj_{t-s}\V{G}_{\V{B}_{t-s}^{\V{x}}}\Id\V{B}_s^{t;\V{x}}&=
-\underset{n\to\infty}{\mathrm{lim\,prob}}
\sum_{\ell=1}^nj_{\ell t/n}\V{G}(\V{B}^{\V{x}}_{\ell t/n})(\V{B}_{\ell t/n}-\V{B}_{(\ell-1)t/n}).
\end{align*}
Moreover, we verified in \cite[Lem.~3.2]{GMM2017} that the term on the
right hand side of \eqref{Kdiv} equals
\begin{align*}
\underset{n\to\infty}{\mathrm{lim\,prob}}
\sum_{\ell=1}^n\frac{1}{2}\big(j_{\ell t/n}\V{G}_{\V{B}^{\V{x}}_{\ell t/n}}
+j_{(\ell-1)t/n}\V{G}_{\V{B}^{\V{x}}_{(\ell-1)t/n}}\big)(\V{B}_{\ell t/n}-\V{B}_{(\ell-1)t/n}).
\end{align*}
Altogether this proves \eqref{Kdiv}. An analogous argument, again employing 
\cite[Lem.~3.2]{GMM2017}, applies when $\V{b}^{t;\V{y},\V{x}}$ and 
$\smash{\hat{\V{b}}}^{t;\V{x},\V{y}}$
are put in place of $\V{B}^{\V{x}}$ and $\V{B}^{t;\V{x}}$, respectively.
The relations \eqref{Sdiv} and \eqref{Sxydiv} can be proved in the same fashion, using the more
well-known \eqref{stratos} below.
\end{proof}

In what follows we shall employ the following notation:
\begin{enumerate}[leftmargin=0.8cm]
\item[$\triangleright$]
$\check{W}_t(\V{x},\V{y})$ is the random operator obtained upon replacing
$\V{b}^{t;\V{y},\V{x}}$ by $\smash{\check{\V{b}}}^{t;\V{y},\V{x}}$ in 
\eqref{Sxydiv} and \eqref{Kxydiv} and plugging the result into \eqref{introdefWtxy}.
Recall that $\smash{\check{\V{b}}}^{t;\V{x},\V{y}}$ has been defined in \eqref{defbarbridge}.
\item[$\triangleright$]
$\hat{W}_t(\V{y},\V{x})$ is the random operator obtained upon replacing
$\V{b}^{t;\V{y},\V{x}}$ by $\smash{\hat{\V{b}}}^{t;\V{x},\V{y}}$ in \eqref{Sxydiv} and 
\eqref{Kxydiv} and plugging the result into \eqref{introdefWtxy}; 
$\smash{\hat{\V{b}}}^{t;\V{x},\V{y}}$ is defined in \eqref{BBrev}.
\end{enumerate}

\begin{thm}\label{thmWadjoint}
Let $t>0$ and $\V{x},\V{y}\in\RR^\nu$.  Then the following identities hold $\PP$-a.s.,
\begin{align}\label{Wadjoint}
W_{t}(\V{x})^*&=\check{W}_{t}(\V{x},\V{B}_t^{\V{x}}),\qquad
W_{t}(\V{x},\V{y})^*=\hat{W}_t(\V{y},\V{x}).
\end{align}
Furthermore, the random field
$(\check{W}_t(\V{x},\V{z}))_{\V{z}\in\RR^\nu}$ can be modified such that the following map is 
continuous, for every $\vo\in\Omega$,
\begin{align}\label{miroslav}
\RR^\nu\times\sF\ni(\V{z},\psi)\longmapsto(\check{W}_t(\V{x},\V{z}))(\vo)\psi\in\sF.
\end{align}
\end{thm}

\begin{proof}
For $\V{A}=0$, all assertions follow from \cite[Thm.~9.2 and Lem.~10.2]{GMM2017}. 
Assume without loss of generality that $V=0$. Then
\begin{align}\label{stratos}
S_{t}(\V{x})&=-\underset{n\to\infty}{\mathrm{lim\,prob}}
\sum_{\ell=1}^n\frac{i}{2}\big(\V{A}(\V{B}^{\V{x}}_{\ell t/n})
+\V{A}(\V{B}^{\V{x}}_{(\ell-1)t/n})\big)(\V{B}_{\ell t/n}-\V{B}_{(\ell-1)t/n}).
\end{align}
Under the replacements $\ell\to n-\ell+1$ we obviously obtain the 
complex conjugates of the approximating sums. Therefore,
\begin{align}\label{idris1}
\ol{S_t(\V{x})}&=-i\int_0^t\V{A}(\V{B}_s^{t;\V{x}})\Id\V{B}_s^{t;\V{x}}
-\frac{i}{2}\int_0^t\Div\V{A}(\V{B}_s^{t;\V{x}})\Id s,
\end{align}
where the stochastic integral on the right hand side is constructed with respect to the filtration 
$(\check{\fF}_s)_{s\in[0,t]}$.
Let $\smash{\check{S}}_t(\V{x},\V{y})$ denote the random variable obtained upon putting 
$\smash{\check{\V{b}}}^{t;\V{y},\V{x}}$ in place of $\V{b}^{t;\V{y},\V{x}}$ 
on the right hand side of \eqref{Sxydiv}.
Since $\V{B}^{t;\V{x}}$ solves \eqref{SDEbridgebar} with the $\check{\fF}_0$-measurable initial 
condition $\check{\V{q}}:=\V{B}_t^{\V{x}}$ and since $\V{A}\in C^1_b(\RR^\nu,\RR^\nu)$,
the random variable on the right hand side of \eqref{idris1} is $\PP$-a.s. equal to
$\check{S}_{t}(\V{x},\check{\V{q}})$ (where the integrals are {\em first} computed along
$\smash{\check{\V{b}}}^{\tau;\V{y},\V{x}}$, for each $\V{y}\in\RR^\nu$, and 
$\V{y}=\check{\V{q}}$ is substituted {\em afterwards}). These remarks extend the first
identity in \eqref{Wadjoint} to non-vanishing $\V{A}$.
The second identity in \eqref{Wadjoint} can be proved, slightly more directly, in the same fashion.
Finally, the last assertion extends
to non-vanishing $\V{A}\in C_b^1(\RR^\nu,\RR^\nu)$ by standard properties of the
stochastic integrals defining $\smash{\check{S}}_t(\V{x},\V{z})$.
\end{proof}

Next, we discuss a flow equation. 
For every $r\ge0$, we set
\begin{align*}
{}^r\!\V{B}_t&:=\V{B}_{t+r}-\V{B}_r,\quad t\ge0;\qquad {}^r\!\V{B}^{\V{x}}:=\V{x}+\!{}^r\!\V{B},
\quad \V{x}\in\RR^\nu.
\end{align*}
so that ${}^r\!\V{B}$ is a $(\mathfrak{F}_{r+t})_{t\ge0}$-Brownian motion on the 
time-shifted filtered probability space $(\Omega,\mathfrak{F},(\mathfrak{F}_{r+t})_{t\ge0},\PP)$.
Denoting by $(W_{r,r+t}(\V{x}))_{t\ge0}$ the process obtained upon putting ${}^r\!\V{B}$ in place
of $\V{B}$ in \eqref{Sdiv} and \eqref{Kdiv} and plugging the result into \eqref{defWWintroscal}, 
we have the following result:

\begin{thm}\label{thmWflow}
By choosing a suitable version of the process $(W_{r,r+t}(\V{x}))_{t\ge0}$, for each $r\ge0$ and
each $\V{x}\in\RR^\nu$, we can achieve the following:
\begin{enumerate}[leftmargin=0.8cm]
\item[{\rm(1)}] For all $r\ge0$ and $\vo\in\Omega$, the following map is continuous,
\begin{align*}
[r,\infty)\times\RR^\nu\times\sF\ni(t,\V{x},\psi)\longmapsto (W_{r,t}(\V{x}))(\vo)\psi\in\sF.
\end{align*}
\item[{\rm(2)}] Fix $r\ge0$ and $\V{x}\in\RR^\nu$. Then
$W_{r,r}(\V{x})=\id$ and the following flow equations hold $\PP$-a.s.,
\begin{align}\label{flowW}
W_{r,t}(\V{x})&=W_{s,t}(^r\!\V{B}_{s-r}^{\V{x}})W_{r,s}(\V{x}),\quad t\ge s\ge r.
\end{align}
\item[{\rm(3)}] For all $t\ge r\ge 0$ and $\V{x}\in\RR^\nu$, the random variable
$W_{r,t}(\V{x})$ is $\fF_r$-independent.
\end{enumerate}
\end{thm}

\begin{proof}
For $\V{A}=0$, all statements are contained in \cite[Thm.~9.2]{GMM2017}. By standard results
on stochastic integrals they extend to non-vanishing $\V{A}$ in $C_b^1(\RR^\nu,\RR^\nu)$.
\end{proof}

%%%%%%%%%%%%%%%%%%%%%%%%%%%%%%%%%%%%%%%%%%%%%%%%

\subsection{The semigroup and its integral kernel for regular coefficients}\label{ssecTTxyreg}

\noindent
For all $\Psi\in L^2(\RR^\nu,\sF)$, we abbreviate
\begin{align}\label{defTtx}
(T_t\Psi)(\V{x})&:=\EE\big[W_{t}(\V{x})^*\Psi(\V{B}_t^{\V{x}})\big],\quad 
t\ge0,\,\V{x}\in\RR^\nu.
\end{align}
In view of \eqref{WWWbd} this defines a bounded operator $T_t$ on $L^2(\RR^\nu,\sF)$ satisfying
\begin{align}\label{Tbd}
\|T_t\|&\le 1,\quad t\ge0.
\end{align}
Recalling our notation \eqref{Gausskern} for the Euclidean heat kernel we further write
\begin{align}\label{defTxy}
T_t(\V{x},\V{y})&:=p_t(\V{x},\V{y})\EE[W_t(\V{x},\V{y})],\quad t>0,\,\V{x},\V{y}\in\RR^\nu;
\end{align}
recall Rem.~\ref{remkernig} concerning the existence
of the $\LO(\sF)$-valued integral in \eqref{defTxy}. 

\begin{prop}\label{propTsa}
Let $t>0$. Then
\begin{align}\label{desint}
(T_t\Psi)(\V{x})&=\int_{\RR^\nu}T_t(\V{x},\V{y})\Psi(\V{y})\Id\V{y},\quad\V{x}\in\RR^\nu,
\end{align}
for all $\Psi\in L^2(\RR^\nu,\sF)$, and
\begin{align}\label{Txysym}
T_t(\V{x},\V{y})^*&=T_t(\V{y},\V{x}),\quad \V{x},\V{y}\in\RR^\nu.
\end{align}
In particular, $T_t$ is a bounded selfadjoint operator on $L^2(\RR^\nu,\sF)$.
\end{prop}

\begin{proof}
Let $t>0$ and $\V{x}\in\RR^\nu$. Combining \eqref{Wadjoint} and \eqref{defTtx} we find
\begin{align}\label{vivian1}
(T_t\Psi)(\V{x})&=\EE\big[\check{W}_{t}(\V{x},\V{B}_\tau^{\V{x}})\Psi(\V{B}_\tau^{\V{x}})\big]
=\EE\big[\EE^{\check{\fF}_0}[\check{W}_{t}(\V{x},\V{B}_\tau^{\V{x}})\Psi(\V{B}_\tau^{\V{x}})]\big],
\end{align}
where we also used the tower property of conditional expectations in the second equality.
By definition of the reversed filtration $(\check{\fF}_s)_{s\ge0}$, the random
functions $\V{B}_t^{\V{x}}$ and, hence, $\Psi(\V{B}_t^{\V{x}})$ are $\check{\fF}_0$-measurable.
Furthermore, $\check{W}_{t}(\V{x},\V{y})$ is $\check{\fF}_0$-independent, as this is the case for the
increments of solutions to \eqref{SDEbridgebar} with a constant initial condition $\check{\V{q}}=\V{y}$.
In view of the continuity result stated in Thm.~\ref{thmWadjoint} we may thus apply the computation 
rule for conditional expectations of Example~\ref{exusefulrule} to the rightmost member in
\eqref{vivian1}. This entails the first equality in
\begin{align*}
(T_t\Psi)(\V{x})&=\EE\Big[\EE[\check{W}_{t}(\V{x},\V{y})]\big|_{\V{y}=\V{B}_t^{\V{x}}}
\Psi(\V{B}_t^{\V{x}})\Big]
=\int_{\RR^\nu}p_t(\V{x},\V{y})\EE[\check{W}_{t}(\V{x},\V{y})]\Psi(\V{y})\Id\V{y}.
\end{align*}
In the second one we just used that the law of $\V{B}_t^{\V{x}}$ has density $p_t(\V{x},\cdot)$. Since 
$\check{W}_{t}(\V{x},\V{y})$ has the same law as ${W}_{t}(\V{x},\V{y})$, we arrive at \eqref{desint}.

The identity \eqref{Txysym} follows from the second relation in \eqref{Wadjoint} since 
$\hat{W}_t(\V{y},\V{x})$ and ${W}_t(\V{y},\V{x})$ have the same law.
\end{proof}

In the next proposition we again use the notation introduced in front of  Thm.~\ref{thmWflow}:

\begin{prop}\label{propTSG}
Let $\V{x}\in\RR^\nu$ and $\Psi\in L^2(\RR^\nu,\sF)$. Then
the following Markov property holds, for all $t\ge s\ge r\ge0$,
\begin{align}\label{MarkovT}
\EE^{\fF_{s}}[W_{r,t}(\V{x})^*\Psi(^r\!\V{B}_{t-r}^{\V{x}})]&=W_{r,s}(\V{x})^*
(T_{t-s}\Psi)(^r\!\V{B}_{s-r}^{\V{x}}),\quad\text{$\PP$-a.s.}
\end{align}
In particular, for all $s,t\ge0$,
\begin{align}\label{TSG}
(T_{s+t}\Psi)(\V{x})=(T_s(T_t\Psi))(\V{x}).
\end{align}
\end{prop}

\begin{proof}
Since taking the adjoint is continuous on $\LO(\sF)$, the map
$W_{u,v}(\V{y})^*:\Omega\to\LO(\sF)$ is again measurable and
separably valued, for all $v\ge u\ge 0$ and $\V{y}\in\RR^\nu$. Furthermore,
$W_{r,s}(\V{x})^*$ is $\fF_s$-measurable and $W_{s,t}(\V{y})^*$ is $\fF_s$-independent, for all 
$\V{y}\in\RR^\nu$ by Thm.~\ref{thmWflow}(3). The Markov property \eqref{MarkovT} thus follows 
from Parts~(1) and~(2) of Thm.~\ref{thmWflow} in conjunction with Example~\ref{exusefulrule}.
Taking the expectation of \eqref{MarkovT} with $r=0$ we further obtain \eqref{TSG}.
\end{proof}

%%%%%%%%%%%%%%%%%%%%%%%%%%%%%%%%%%%%%%%%%%%%%%%%%

\subsection{Feynman-Kac formulas on $\RR^\nu$ for regular coefficients}\label{ssecFKRRnureg}

\noindent
In this subsection we shall often use the shorthand
\begin{align*}
\theta:=1+\Id\Gamma(\omega),
\end{align*}
and abbreviate
\begin{align*}
\wh{H}(\V{x})\psi&:=\frac{1}{2}\vp(\V{G}_{{\V{x}}})^2\psi
-\frac{i}{2}\vp(\Div\V{G}_{{\V{x}}})\psi+\Id\Gamma(\omega)\psi,\quad
\psi\in\dom(\Id\Gamma(\omega)),\,\V{x}\in\RR^\nu.
\end{align*}

\begin{lem}
Let $\V{x}\in\RR^\nu$, $f\in C_b^2(\RR^\nu,\RR)$, and $\psi\in\dom(\Id\Gamma(\omega))$. Then
\begin{align}\label{defsmart}
\sM_{\bullet}(\V{x})&:=\int_0^\bullet\big(i(f\V{A})(\V{B}_s^{\V{x}})+f(\V{B}_s^{\V{x}})
i\vp(\V{G}_{\V{B}_s^{\V{x}}})+(\nabla f)(\V{B}_s^{\V{x}})\big)W_{s}(\V{x})\psi\Id\V{B}_s,
\end{align}
defines a continuous $\sF$-valued $L^2$-martingale $\sM(\V{x})$ on $[0,\infty)$ and, $\PP$-a.s.,
\begin{align}\nonumber
&f(\V{B}_t^{\V{x}})W_{t}(\V{x})\psi-{f(\V{x})\psi}
\\\nonumber
&= \int_0^t\bigg(\Big(\frac{1}{2}(\nabla+i\V{A})^2f-Vf\Big)(\V{B}_s^{\V{x}})
-f(\V{B}_s^{\V{x}})\wh{H}(\V{B}_s^{\V{x}})\bigg)W_{s}(\V{x})\psi\Id s
\\\label{richard1}
&\quad+\int_0^t(i\nabla f-f\V{A})(\V{B}_s^{\V{x}})\cdot\vp(\V{G}_{\V{B}_s^{\V{x}}})
W_{s}(\V{x})\psi\Id s+\sM_{t}(\V{x}),\quad t\ge0.
\end{align}
\end{lem}

\begin{proof}
According to \cite[Lem.~7.6]{GMM2017} there exists a monotone 
increasing function $c:[0,\infty)\to(0,\infty)$ such that
\begin{align}\label{rolf}
\sup_{\V{z}\in\RR^\nu}\EE\Big[\sup_{s\in[0,t]}\|
\theta W_{s}(\V{z})\psi\|^{2}\Big]&\le c(t)\|\theta\psi\|^{2},\quad t\ge0.
\end{align}
In view of \eqref{WWWbd}, \eqref{rbvp}, and \eqref{rolf} 
the integrand of the stochastic integral defining $\sM(\V{x})$, call it $(Y_s)_{s\ge0}$, 
is a continuous adapted $\sF$-valued stochastic process satisfying
\begin{align*}
\EE\big[\|Y_s\|^2\big]&\le3\big\{\|\V{A}\|_\infty^2\|f\|_\infty^2\|\psi\|^2 +\|\nabla f\|_\infty^2
\|\psi\|^2+\tilde{c}(s)\|\theta\psi\|^2\big\},
\end{align*}
for all $s\ge0$, where $\tilde{c}:[0,\infty)\to(0,\infty)$ is another monotone 
increasing function. Consequently, $\sM(\V{x})$ is a continuous $\sF$-valued $L^2$-martingale.

Put $W_t^0(\V{x}):=\Gamma(j_t)^*e^{i\vp(K_t(\V{x}))}\Gamma(j_0)$;
compare this with \eqref{defWWintroscal}.
Thanks to \cite[Thm.~5.3]{GMM2017} we know that
$(W_{t}^0(\V{x})\psi)_{t\ge0}$ is a $\sF$-valued semimartingale whose paths $\PP$-a.s. are continuous
$\dom(\Id\Gamma(\omega))$-valued functions and, $\PP$-a.s.,
\begin{align}\label{SDEW}
W_{t}^0(\V{x})\psi&=\psi-\int_0^t\wh{H}(\V{B}_s^{\V{x}})W_{s}^0(\V{x})\psi\Id s
+\int_0^ti\vp(\V{G}_{\V{B}_s^{\V{x}}})W_{s}^0(\V{x})\psi\Id\V{B}_s,\quad t\ge0.
\end{align}
 Thus, \eqref{richard1} follows from \eqref{SDEW} and It\^{o}'s formula.  
 %If we scalar multiply \eqref{richard1} with some $\phi\in\sF$, then it actually suffices to apply the standard textbook version of It\^{o}'s formula for $C^2$-functions composed with continuous $\RR^d$-valued semimartingales. In fact, the latter applies to $F(Z)$, where $F(s,\sigma,\V{y}):=e^{-s}\sigma f(\V{y})$, for all $s,\sigma\in\CC$ and $\V{y}\in\RR^\nu$, and where \begin{align*}Z_t:=\big(S_t(\V{x}),\SPn{\phi}{W_{t}(\V{x})\psi},\V{B}_t^{\V{x}}\big),\quad t\ge0,\end{align*}defines a $\CC^2\times\RR^{\nu+1}$-valued continuous semimartingale $Z$.
\end{proof}

\begin{lem}\label{lemanton}
There exists $c>0$ such that, for all $\psi\in\sF$,
\begin{equation}\label{antonio0}
\sup_{\V{x}\in\RR^\nu}\EE\Big[\sup_{s\in[0,t]}\|\theta^\mh
(W_{s}(\V{x})-\id)\psi\|^2\Big]\le ct\|\psi\|^2,\quad t\ge0.
\end{equation}
\end{lem}

\begin{proof}
Abbreviate  $\psi_t:=(W_{t}(\V{x})-\id)\psi$, so that $\psi_0=0$. We may assume that 
$\psi\in\dom(\Id\Gamma(\omega))$. (Otherwise approximate $\psi$ by the vectors 
$(1+\Id\Gamma(\omega)/n)^{-1}\psi$, $n\in\NN$, and take \eqref{WWWbd} into account.)
We may also assume $\|\psi\|=1$.
In virtue of \eqref{richard1} with $f=1$ and It\^{o}'s formula, we $\PP$-a.s. obtain, for all $t\ge0$,
\begin{align}\nonumber
\|\theta^\mh\psi_t\|^2&=-\int_0^t2\Re\SPb{\psi_s}{\theta^{-1}
\big(V-\tfrac{1}{2}\V{A}^2+\wh{H}\big)(\V{B}_s^{\V{x}})W_{s}(\V{x})\psi}\Id s
\\\nonumber
&\quad
-\int_0^t2\Re\SPb{\psi_s}{\theta^{-1}\V{A}(\V{B}_s^{\V{x}})\cdot\vp(\V{G}_{\V{B}_s^{\V{x}}})
W_s(\V{x})\psi}\Id s
\\\nonumber
&\quad
+\int_0^t\big\|\theta^\mh\big(\V{A}(\V{B}_s^{\V{x}})+\vp(\V{G}_{\V{B}_s^{\V{x}}})\big)
W_{s}(\V{x})\psi\big\|^2\Id s
\\\label{richard3}
&\quad+\int_0^t2\Re\SPb{\theta^\mh\psi_s}{i\theta^\mh
\big(\V{A}(\V{B}_s^{\V{x}})+\vp(\V{G}_{\V{B}_s^{\V{x}}})\big)W_{s}(\V{x})\psi}\Id\V{B}_s.
\end{align}
On account of \eqref{rbvp}, \eqref{rbvpvp}, \eqref{AVreg}, and Hyp.~\ref{hypGreg}, the operators
$$
\theta^{-1}\big(V-\tfrac{1}{2}\V{A}^2+\wh{H}\big)({\V{y}}),\quad
\theta^{-1}\V{A}(\V{y})\cdot\vp(\V{G}_{{\V{y}}}),\quad
\theta^\mh\big(\V{A}({\V{y}})+\vp(\V{G}_{{\V{y}}})\big),
$$
appearing here
are well-defined on $\dom(\Id\Gamma(\omega))$ and bounded uniformly in $\V{y}\in\RR^\nu$. 
Furthermore, we have the pointwise bound $\|\psi_t\|\le2$, $t\ge0$. From these
remarks we infer in particular that the stochastic integral in the last line of \eqref{richard3}, call it 
$\cM$, is a martingale to which Davis' inequality applies, i.e., 
\begin{align*}
\EE[\sup_{s\le t}|\cM_s|]\le c_0\EE[\langle\cM\rangle_t^\eh],\quad t\ge0,
\end{align*}
for some universal constant $c_0>0$.
According to the above remarks the quadratic variation of $\cM$ satisfies, however,
\begin{align*}
\langle\cM\rangle_t&=\int_0^t\big(2\Re\SPb{\theta^\mh\psi_s}{i\theta^\mh
\big(\V{A}(\V{B}_s^{\V{x}})+\vp(\V{G}_{\V{B}_s^{\V{x}}})\big)W_{s}(\V{x})\psi}\big)^2\Id s
\\
&\le c_1\int_0^t\|\theta^\mh\psi_s\|^2\Id s,\quad t\ge0,
\end{align*}
$\PP$-a.s., for some constant $c_1>0$, whence
\begin{align*}
\EE\big[\langle\cM\rangle_t^\eh\big]
&\le\EE\big[c_1^\eh t^\eh \sup_{s\le t}\|\theta^\mh\psi_s\|\big]
\le\frac{1}{2c_0}\EE\big[\sup_{s\le t}\|\theta^\mh\psi_s\|^2\big]+\frac{c_0c_1t}{2}.
\end{align*}
Since \eqref{richard3} and the above remarks entail
\begin{align}\label{anton1}
\EE\big[\sup_{s\le t}\|\theta^\mh\psi_s\|^2\big]&\le c_2t+\EE\big[\sup_{s\le t}|\cM_s|\big],
\quad t\ge0,
\end{align}
with another constant $c_2>0$,
we thus arrive at an inequality that we can solve for the left hand side of \eqref{anton1}
(which is finite, as we know {\em a priori}).
\end{proof}

\begin{prop}\label{prop-SG}
$(T_t)_{t\ge0}$ is a strongly continuous semigroup 
of bounded selfadjoint operators on $L^2(\RR^\nu,\sF)$.
\end{prop}

\begin{proof}
Boundedness and selfadjointness have already been observed in \eqref{Tbd} and
Prop.~\ref{propTsa}. In view of \eqref{Tbd} it only remains to show that $T_t\Psi\to\Psi$, 
as $t\downarrow0$, for all $\Psi\in\sF$ with $\|\theta^\eh\Psi\|_{\sF}\in L^2(\RR^\nu)$.
(Vectors $\Psi$ of the latter kind are dense in $\sF$.) 
For every such $\Psi$, the convergence $T_t\Psi\to\Psi$ follows, however, from an 
estimation which is virtually identical to the one in the proof of \cite[Lem.~10.11]{GMM2017}.
Let us nevertheless repeat it here to demonstrate where and how Lem.~\ref{lemanton} is used:
\begin{align*}
\|(T_t-\id)\Psi\|^2
&=\int_{\RR^\nu}\sup_{{\phi\in\sF:\|\phi\|=1}}\big|\SPb{\phi}{
\EE\big[(W_{t}(\V{x})^*-\id)\Psi(\V{B}_t^{\V{x}})\big]}\big|^2\Id\V{x}
\\
&=\int_{\RR^\nu}\sup_{{\phi\in\sF:\|\phi\|=1}}\big|\EE\big[\SPb{\theta^\mh
(W_{t}(\V{x})-\id)\phi}{\theta^\eh\Psi(\V{B}_t^{\V{x}})}\big]\big|^2\Id\V{x}
\\
&\le\sup_{{\V{y}\in\RR^\nu}}\sup_{{\phi\in\sF:\|\phi\|=1}}
\EE\big[\|\theta^\mh(W_{s}(\V{y})-\id)\phi\|^2\big]\int_{\RR^\nu}\EE\big[\|\theta^\eh
\Psi(\V{B}_t^{\V{x}})\|^2\big]\Id\V{x},
\end{align*}
where the double supremum of the first expectation in the last line is $\le ct$ 
by Lem.~\ref{lemanton} and the $\Id\V{x}$-integral in the same line is $\le\|\theta^\eh\Psi\|^2$.
\end{proof}

\begin{prop}\label{propFKreg}
Let $\Psi\in L^2(\RR^\nu,\sF)$, $t>0$, and $H:=H_{\RR^\nu}$. Then
\begin{align}\label{FK-reg}
e^{-tH}\Psi&=T_t\Psi.
\end{align}
\end{prop}

\begin{proof}
We pick $f\in C_0^\infty(\RR^\nu,\RR)$, $\psi\in\dom(\Id\Gamma(\omega))$, 
scalar-multiply \eqref{richard1} with $\phi\in\dom(\Id\Gamma(\omega))$, and use the fact that 
$\SPn{\phi}{\sM(\V{x})}$ is a martingale starting at zero to get
\begin{align}\nonumber
&\SPb{(T_t(f\phi))(\V{x})}{\psi}-\SPn{f(\V{x})\phi}{\psi}+t\SPb{\big(H(f\phi)\big)({\V{x}})}{\psi}
\\\label{richard2}
&=\int_0^t\EE\Big[\SPB{(H(f\phi))(\V{x})-W_s(\V{x})^*(H(f\phi))(\V{B}_s^{\V{x}})}{\psi}\Big]\Id s
=:I_\psi(t,\V{x}),
\end{align}
for all $t\ge0$ and $\V{x}\in\RR^\nu$. Here
\begin{align*}
\frac{1}{t^2}\int_{\RR^\nu}\sup_{{\psi\in\dom(\Id\Gamma(\omega)):\|\psi\|=1}}
|I_{\psi}(t,\V{x})|^2\Id\V{x}
&\le\frac{1}{t}\int_0^t\|(T_s-\id)H(f\phi)\|^2\Id s\xrightarrow{\;\;t\downarrow0\;\;}0,
\end{align*}
because $(T_s)_{s\ge0}$ is strongly continuous. This shows that 
$$
\frac{1}{t}\big(T_t(f\phi)-f\phi\big)\xrightarrow{\;\;t\downarrow0\;\;}H(f\phi)
\quad\text{in $L^2(\RR^\nu,\sF)$.}
$$
Hence, $\sD(\RR^\nu,\dom(\Id\Gamma(\omega)))$ is contained
in the domain of the selfadjoint generator of $(T_t)_{t\ge0}$ and the restriction of this generator to 
$\sD(\RR^\nu,\dom(\Id\Gamma(\omega)))$ is equal to
$H\restr_{\sD(\RR^\nu,\dom(\Id\Gamma(\omega)))}$. Since $H$ is essentially selfadjoint on 
$\sD(\RR^\nu,\dom(\Id\Gamma(\omega)))$ (see, e.g., \cite[Thm.~5.5]{Matte2017}), this
implies that $(T_t)_{t\ge0}$ is generated by $H$.
\end{proof}

%%%%%%%%%%%%%%%%%%%%%%%%%%%%%%%%%%%%%%%%%%%%%%%%%
%%%%%%%%%%%%%%%%%%%%%%%%%%%%%%%%%%%%%%%%%%%%%%%%%
%%%%%%%%%%%%%%%%%%%%%%%%%%%%%%%%%%%%%%%%%%%%%%%%%

\section{Feynman-Kac formulas for singular coefficients}\label{secFKsing}

\noindent
In the first two subsections of this final section we give a precise meaning to all stochastic integrals 
appearing in the formulas for our Feynman-Kac integrands and observe a useful dominated convergence
theorem for a particular class of stochastic integrals. After that we prove our main theorem
for the special choice $\Geb=\RR^\nu$ and continuous, bounded $V$ in Subsect.~\ref{ssecFKsingAG}.
Ultimately, we obtain the theorem in full generality in Subsect.~\ref{ssecFKsingAGGeb},
employing the results of Sect.~\ref{secFKforD} as well as an additional idea from \cite{Simon1978}.
Cor.~\ref{cornegpart} is proved in Subsect.~\ref{ssecFKsingAGGeb}, too.

\subsection{Existence and convergence of path integrals}\label{ssecexpathint}

\noindent
Let $\sK$ be a separable real or complex Hilbert space and
\begin{align*}
\V{f}\in L_\loc^2(\RR^\nu,\sK^\nu).
\end{align*}
More precisely, we assume that a representative of $\V{f}$ has been chosen so that
$\V{f}:\RR^\nu\to\sK^{\nu}$ is Borel measurable. Furthermore, we suppose that 
$\RR\ni s\mapsto J_s\in\LO(\sK,\hat{\sK})$ is a strongly continuous family of isometries from 
$\sK$ into another separable Hilbert space $\hat{\sK}$. Relevant examples are
$j_s:\HP\to\hat{\HP}$ and $\mathrm{id}_{\RR}:\RR\to\RR$. Finally, we fix $t>0$.

\begin{lem}\label{lemexstochint}
There exist Borel zero sets $N\subset\RR^\nu$ and $N'\subset\RR^{2\nu}$ such that
the two stochastic integral processes
\begin{align}\label{primus}
&\bigg(\int_0^\tau J_s\V{f}(\V{B}_s^{\V{x}})\Id\V{B}_s\bigg)_{\tau\in[0,t]},\quad
\bigg(\int_0^\tau J_{t-s}\V{f}(\V{B}_s^{t;\V{x}})\Id\V{B}^{t;\V{x}}_s\bigg)_{\tau\in[0,t]},
\end{align}
are well-defined semimartingales, for all $\V{x}\in\RR^\nu\setminus N$, and 
\begin{align}\label{secundus}
&\bigg(\int_0^\tau J_s\V{f}(\V{b}_s^{t;\V{y},\V{x}})\Id\V{b}^{t;\V{y},\V{x}}_s\bigg)_{\tau\in[0,t]},
\quad
\bigg(\int_0^\tau J_{t-s}\V{f}(\smash{\hat{\V{b}}}_s^{t;\V{x},\V{y}})
\Id\smash{\hat{\V{b}}}^{t;\V{x},\V{y}}_s\bigg)_{\tau\in[0,t]},
\end{align}
are well-defined semimartingales for all $(\V{x},\V{y})\in \RR^{2\nu}\setminus N'$. 
The zero sets $N$ and $N'$ can be chosen independently of the choice of representative
of $\V{f}$. If this has been done, then, for every $\V{x}\in\RR^\nu\setminus N$ and 
$(\V{x},\V{y})\in\RR^{2\nu}\setminus N'$,
the semimartingales in \eqref{primus} and \eqref{secundus}, respectively, 
change only up to indistinguishability, if we pick another representative of $\V{f}$.
\end{lem}

Notice that the first integral processes in \eqref{primus} and \eqref{secundus}
are defined and semimartingales with respect to the filtration $(\fF_s)_{s\in[0,t]}$,
while the second one in \eqref{primus} is constructed using $(\smash{\check{\fF}}_s)_{s\in[0,t]}$
and the second one in \eqref{secundus} by means of $(\smash{\hat{\fF}}_s)_{s\in[0,t]}$.

\begin{proof}
As we neither specify $(J_s)_{s\in[0,t]}$, $(\fF_s)_{s\in[0,t]}$, nor $\V{B}$,
we may ignore the second process in \eqref{secundus} in this proof.

Taking the strong continuity of $s\mapsto J_s$ into account we first observe that all integrands in
\eqref{primus} and \eqref{secundus} are predictable with respect to the corresponding
filtrations. In view of the stochastic differential equations solved by $\V{B}^{t;\V{x}}$ and
$\V{b}^{t;\V{y},\V{x}}$, we further have 
\begin{align*}
\int_0^\tau J_{t-s}\V{f}(\V{B}_s^{t;\V{x}})\Id\V{B}^{t;\V{x}}_s
&=\int_0^\tau J_{t-s}\V{f}(\V{B}_{t-s}^{\V{x}})\Id\check{\V{B}}_s
-\int_0^\tau J_{t-s}\V{f}(\V{B}_{t-s}^{\V{x}})\cdot\frac{\V{B}_{t-s}}{t-s}\Id s,
\\
\int_0^\tau J_s\V{f}(\V{b}_s^{t;\V{y},\V{x}})\Id\V{b}^{t;\V{y},\V{x}}_s
&=\int_0^\tau J_s\V{f}(\V{b}_s^{t;\V{y},\V{x}})\Id\V{B}_s
+\int_0^\tau J_s\V{f}(\V{b}_s^{t;\V{y},\V{x}})\cdot\frac{\V{x}-\V{b}_s^{t;\V{y},\V{x}}}{t-s}\Id s,
\end{align*}
for all $\tau\in[0,t]$. By the standard criterion for the existence of stochastic integrals along 
Brownian motions (see, e.g., \cite[\textsection4.2]{daPrZa2014}), the $\Id\check{\V{B}}$- and 
$\Id\V{B}$-integrals in the previous two formulas and the $\Id\V{B}$-integral to the 
left in \eqref{primus} are well-defined, if
\begin{align}\label{tertius}
\PP\bigg(\int_0^t\|J_s\V{f}(\V{B}_s^{\V{x}})\|^2_{\hat{\sK}^\nu}\Id s<\infty\bigg)=
\PP\bigg(\int_0^t\|J_{t-s}\V{f}(\V{B}_{t-s}^{\V{x}})\|^2_{\hat{\sK}^\nu}\Id s<\infty\bigg)&=1,
\\\label{quartus}
\PP\bigg(\int_0^t\|J_s\V{f}(\V{b}_s^{t;\V{y},\V{x}})\|^2_{\hat{\sK}^\nu}\Id s<\infty\bigg)&=1.
\end{align}
Furthermore, the pathwise defined Bochner-Lebegsue integrals in the above two formulas
exist and define processes having pathwise finite variation on $[0,t]$, $\PP$-a.s. at least, provided that
\begin{align}\label{quintus}
\PP\bigg(\int_0^t \|J_{s}\V{f}(\V{B}_{s}^{\V{x}})\|_{\hat{\sK}^\nu}
\frac{|\V{B}_{s}|}{s}\Id s<\infty\bigg)&=1,
\\\label{sextus}
\PP\bigg(\int_0^t\|J_s\V{f}(\V{b}_s^{t;\V{y},\V{x}})\|_{\hat{\sK}^\nu}
\frac{|\V{x}-\V{b}_s^{t;\V{y},\V{x}}|}{t-s}\Id s<\infty\bigg)&=1.
\end{align}

To verify \eqref{tertius} through \eqref{sextus}, we may obviously ignore the isometries $J_s$.
Since $\|\V{f}\|^2:=\|\V{f}\|^2_{{\sK}^\nu}$ is locally integrable on $\RR^\nu$, it follows from
\cite[Lem.~2]{FarisSimon1975} that \eqref{tertius} is satisfied for a.e. $\V{x}$. We shall, however,
re-obtain this result in the following arguments which elaborate on the ones in \cite{FarisSimon1975}.

Sets $C_n:=\{|\V{x}|\le n\}$, $n\in\NN$. Then a weighted Cauchy-Schwarz inequality yields
\begin{align*}
\EE\bigg[\int_0^t\|(1_{C_n}\V{f})(\V{B}_{s}^{\V{x}})\|\frac{|\V{B}_{s}|}{s}\Id s\bigg]
&\le\EE\bigg[\int_0^t\frac{\|(1_{C_n}\V{f})(\V{B}_{s}^{\V{x}})\|^2}{s^\eh}\Id s\bigg]^\eh
\EE\bigg[\int_0^t\frac{|\V{B}_s|^2}{s^{\nf{3}{2}}}\Id s\bigg]^\eh,
\end{align*}
where the rightmost expectation is a finite $(t,\nu)$-dependent constant and
\begin{align}\nonumber
&\int_{\RR^\nu}\EE\bigg[\int_0^t
\frac{\|(1_{C_n}\V{f})(\V{B}_{s}^{\V{x}})\|^2}{(s\wedge 1)^\eh}\Id s\bigg]\Id\V{x}
\\\label{jeppe0}
&\le\int_0^t\int_{\RR^\nu}\int_{\RR^\nu} p_{s}(\V{x},\V{y})
\frac{\|(1_{C_n}\V{f})(\V{y})\|^2}{s^\eh}\Id\V{x}\,\Id\V{y}\,\Id s
=2t^\eh\|\V{f}\|_{L^2(C_n,\sK^\nu)}^2.
\end{align}
Therefore, we find Borel zero sets $N_n\subset\RR^{\nu}$ such that
\begin{align}\label{igor0}
\EE\bigg[\int_0^t\|(1_{C_n}\V{f})(\V{B}_{s}^{\V{x}})\|^2\Id s\bigg]+
\EE\bigg[\int_0^t\|(1_{C_n}\V{f})(\V{B}_{s}^{\V{x}})\|\frac{|\V{B}_{s}|}{s}\Id s\bigg]<\infty,
\end{align}
for all $\V{x}\in\RR^{\nu}\setminus N_n$. Since the expectation in the first line of \eqref{jeppe0}
does not change when we pass to another representative of $\V{f}$, we can pick each $N_n$
independently of the choice of representative of $\V{f}$.
We set $N:=\bigcup_{n=1}^\infty N_n$. Since
every path of the continuous process $(\V{B}_{s}^{\V{x}})_{s\in[0,t]}$ must be contained
some $C_n$, it readily follows that \eqref{tertius} and \eqref{quintus} are
satisfied for all $\V{x}\in\RR^\nu\setminus N$.

Next, we define
\begin{align*}
c_{n}:=\sup_{\V{x},\V{y}\in C_n}\frac{1}{p_t(\V{x},\V{y})}=
(2\pi t)^{\nf{\nu}{2}}e^{2n^2/t},\quad n\in\NN,
\end{align*}
and recall that, for all $s\in(0,t)$, the law of $\V{b}_s^{t;\V{y},\V{x}}$ is given by
$$
L_{s;\V{y},\V{x}}(\V{z}):=\frac{p_{s}(\V{y},\V{z})p_{t-s}(\V{z},\V{x})}{p_t(\V{x},\V{y})},
\quad\V{z}\in\RR^\nu.
$$
Applying Fubini's theorem we find 
\begin{align}\nonumber
&\int_{C_n}\int_{C_n}\EE\bigg[\int_0^t
([t-s]\wedge1)^\mh\|(1_{C_n}\V{f})(\V{b}_s^{t;\V{y},\V{x}})\|^2\Id s\bigg]\Id\V{x}\,\Id\V{y}
\\\nonumber
&\le\int_0^t(t-s)^\mh\int_{C_n}\int_{C_n}
\int_{C_n}L_{s;\V{y},\V{x}}(\V{z})\|\V{f}(\V{z})\|^2\Id\V{z}\,\Id\V{x}\,\Id\V{y}\,\Id s
\\\nonumber
&\le c_n\int_0^t(t-s)^\mh\int_{C_n}\int_{\RR^\nu}\int_{\RR^\nu}
p_{s}(\V{y},\V{z})p_{t-s}(\V{z},\V{x})\|\V{f}(\V{z})\|_{\sK^\nu}^2\Id\V{x}\,\Id\V{y}\,\Id\V{z}\,\Id s
\\\label{jeppe1}
&=2c_{n}t^\eh\|\V{f}\|_{L^2(C_n,\sK^{\nu})}^2<\infty,\quad n\in\NN.
\end{align}
Also employing the bound (see, e.g., \cite[Lem.~15.2]{GMM2017}) 
\begin{align*}
\EE\bigg[\Big|\frac{\V{x}-\V{b}_s^{t;\V{y},\V{x}}}{t-s}\Big|^2\bigg]
&\le c_{\nu,t}\frac{1+|\V{x}-\V{y}|}{(t-s)\wedge1},\quad s\in(0,t),
\end{align*}
we thus find zero sets $N_n'\subset C_n\times C_n$ such that
\begin{align*}
&\EE\bigg[\int_0^t\|(1_{C_n}\V{f})(\V{b}_s^{t;\V{y},\V{x}})\|
\frac{|\V{x}-\V{b}_s^{t;\V{y},\V{x}}|}{t-s}\Id s\bigg]
\\
&\le\EE\bigg[\int_0^t\frac{\|(1_{C_n}\V{f})(\V{b}_s^{t;\V{y},\V{x}})\|^2}{([t-s]\wedge1)^\eh}
\Id s\bigg]^\eh\bigg(\int_0^t(t-s)^\eh\EE\bigg[\Big|\frac{\V{x}-\V{b}_s^{t;\V{y},\V{x}}}{t-s}\Big|^2
\bigg]\Id s\bigg)^\eh<\infty,
\end{align*}
for all $(\V{x},\V{y})\in(C_n\times C_n)\setminus N_n'$ and $n\in\NN$. Since $\V{b}^{t;\V{y},\V{x}}$ 
is continuous, we conclude that \eqref{quartus} and \eqref{sextus} are satisfied for all
$(\V{x},\V{y})\in\RR^{2\nu}\setminus N'$ with $N':=\bigcup_{n=1}^\infty N_n'$.
Again we can pick each $N_n'$ independently of the representative of $\V{f}$, since all
representatives lead to the same integrand under the $(\Id\V{x}\,\Id\V{y})$-integration in the
first line of \eqref{jeppe1}.

The last assertion is an easy consequence of It\^{o}'s isometry for the $\Id\V{B}$- and
$\Id\check{\V{B}}$-integrals, the continuity of 
stochastic integral processes, the isometry of $J_s$, and the fact that the laws of $\V{B}_s^{\V{x}}$ 
and $\V{b}_s^{t;\V{y},\V{x}}$ with $s\in(0,t)$ are absolutely continuous with respect to the
Lebesgue measure.
\end{proof}

We continue with a particular case of the dominated convergence theorem for stochastic integrals:

\begin{thm}\label{thmdomconvstoch}
Let $\V{f}^n\in L^2_\loc(\RR^\nu,\sK^\nu)$, $n\in\NN\cup\{\infty\}$, and 
$\alpha\in L_\loc^2(\RR^\nu)$. As a consequence of Lem.~\ref{lemexstochint}
we find Borel zero sets $N\subset\RR^\nu$ and $N'\subset\RR^{2\nu}$ such that
all processes  in \eqref{primus} and \eqref{secundus} are well-defined, for 
$\V{x}\in\RR^\nu\setminus N$ and $(\V{x},\V{y})\in\RR^{2\nu}\setminus N'$, respectively,
when any pair $(J_s,\V{f}^n)$ with $n\in\NN\cup\{\infty\}$ or $(\mathrm{id}_{\RR},\alpha)$ is put in
place of $(J_s,\V{f})$. Now, let 
$(I_\tau^\infty)_{\tau\in[0,t]}$ be any of the processes in \eqref{primus} or \eqref{secundus} 
defined by means of $(J_s,\V{f}^\infty)$ for some permitted value of $\V{x}$ (resp. $(\V{x},\V{y})$)
an let $(I_\tau^n)_{\tau\in[0,t]}$ denote the corresponding process defined by means of $(J_s,\V{f}^n)$.
 Assume that $\|\V{f}^n\|_{\sK^\nu}\le\alpha$ a.e. on $\RR^\nu$, for each $n\in\NN$, and
 $\V{f}^n\to\V{f}^\infty$ a.e. on $\RR^\nu$, as $n\to\infty$. Then
\begin{align}\label{cornelia1}
\sup_{\tau\in[0,t]}\|I_\tau^n-I_\tau^\infty\|_{\hat{\sK}}\xrightarrow{\;\;n\to\infty\;\;}0\quad
\text{in probability.}
\end{align}
\end{thm}

\begin{proof}
By the last assertion in Lem.~\ref{lemexstochint} we do not loose generality by assuming the bounds
$\|J_s\V{f}^n\|_{\hat{\sK}^\nu}=\|\V{f}^n\|_{\sK^\nu}\le\alpha$, $n\in\NN$, and the convergence 
$J_s\V{f}^n\to J_s\V{f}^\infty$ to hold {\em everywhere} on $\RR^\nu$.
If we do so, then \eqref{cornelia1} follows from the first assertion in Lem.~\ref{lemexstochint}
and the dominated convergence theorem for stochastic integrals; see, e.g., \cite[Thm.~26.3]{Me1982} 
and the complementing remarks in the proof of \cite[Thm.~2.13]{GMM2017}.
\end{proof}

We shall apply the preceding theorem in conjunction with the following, presumably well-known 
observation, whose proof we include for the reader's convenience:

\begin{lem}\label{lemdomfct}
Let $\V{f}^n\in L^2_\loc(\RR^\nu,\sK^\nu)$, $n\in\NN\cup\{\infty\}$ and assume that
$\V{f}^n\to\V{f}^\infty$ in $L_\loc^2(\RR^\nu,\sK^\nu)$, as $n\to\infty$.
Then there exist integers $1\le m_1<m_2<\ldots$ and some nonnegative 
$\alpha\in L_\loc^2(\RR^\nu)$ such that
$\|\V{f}^{m_\ell}\|_{\sK^\nu}\le\alpha$ a.e. on $\RR^\nu$, for each $\ell\in\NN$, and
$\V{f}^{m_\ell}\to\V{f}^\infty$ a.e. on $\RR^\nu$, as $\ell\to\infty$.
\end{lem}

\begin{proof}
Let $r\in\NN_0$  and abbreviate $\sS_r:=\{r<|\V{x}|\le r+1\}$, if $r\ge1$, and 
$\sS_0:=\{|\V{x}|\le1\}$.
Then, given any subsequence of $\{\V{f}^{n}\}_{n\in\NN}$, call it 
$\{\V{f}^{n_{r-1,\ell}}\}_{\ell\in\NN}$, we can single out another subsequence, call it 
$\{\V{f}^{n_{r,\ell}}\}_{\ell\in\NN}$, such that $\V{f}^{n_{r,\ell}}\to\V{f}^\infty$ a.e. on $\sS_r$ 
as $\ell\to\infty$. Furthermore, we find a dominating function
$\alpha_r\in L^2(\sS_r)$ such that $\|\V{f}^{n_{r,\ell}}\|_{\sK^\nu}\le\alpha_r$ 
a.e. on $\sS_r$, for each $\ell\in\NN$.
(These assertions, including the existence of the dominating function, follow from the Riesz-Fischer
theorem for $L^2(\sS_r,\sK^\nu)$.)
We employ this remark inductively with $n_{0-1,\ell}:=\ell$ and define
$\alpha:=\sum_{r=0}^\infty\alpha_r$, where every $\alpha_r$ is extended to a function on 
$\RR^\nu$ by setting it equal to $0$ outside $\sS_r$. Then $\alpha\in L^2_\loc(\RR^\nu)$ 
and the diagonal sequence $\{\V{f}^{m_\ell}\}_{\ell\in\NN}:=\{\V{f}^{n_{\ell,\ell}}\}_{\ell\in\NN}$ 
has all desired properties.
\end{proof}

%%%%%%%%%%%%%%%%%%%%%%%%%%%%%%%%%%%%%%%%%%%%%%%%

\subsection{The Feynman-Kac integrand for singular vector potentials}\label{ssecFKIsingAG}

\noindent
Next, we explain how the observations of the preceding subsection can be used to make sense out
of the stochastic integrals in \eqref{defSxintro}, \eqref{defKxintro}, \eqref{defSxyintro}, and 
\eqref{defKxyintro}, although $\V{A}$ and $\V{G}$ satisfying \eqref{introA} and \eqref{introGR},
respectively, might not have locally square-integrable extensions to the whole $\RR^\nu$.

Let $\Geb_n\subsetneq\Geb$ be open, proper subsets exhausting $\Geb$ in the sense that
$\ol{\Geb_n}\subset{\Geb}_{n+1}$ for all $n\in\NN$ and $\bigcup_{n=1}^\infty \Geb_n=\Geb$. 
Then $1_{\Geb_n}\V{A}\in L_\loc^2(\RR^\nu,\RR^\nu)$ and
$1_{\Geb_n}\V{G}\in L_\loc^2(\RR^\nu,\HP^\nu)$, after $\V{A}$ and
$\V{G}$ have been extended to functions on $\RR^\nu$ by setting them equal to zero outside $\Geb$.

Let $t>0$. According to the remarks in Subsect.~\ref{ssecexpathint} we may pick zero sets
$N\subset\RR^\nu$ and $N'\subset\RR^{2\nu}$ such that, for all $\V{x}\in\RR^\nu\setminus N$ and
$(\V{x},\V{y})\in\RR^{2\nu}$, respectively, we obtain linear combinations of well-defined, 
$\PP$-a.s. uniquely determined stochastic integrals,
\begin{align*}
K_t^n(\V{x})&:=\frac{1}{2}\int_0^tj_s(1_{\Geb_n}\V{G})_{\V{B}_s^{\V{x}}}\Id\V{B}_s^{\V{x}}
-\frac{1}{2}\int_0^tj_{t-s}(1_{\Geb_n}\V{G})_{\V{B}^{t;\V{x}}_s}\Id\V{B}^{t;\V{x}}_s,
\\
K_t^n(\V{x},\V{y})&:=\frac{1}{2}\int_0^tj_s(1_{\Geb_n}\V{G})_{\V{b}_s^{t;\V{y},\V{x}}}
\Id\V{b}_s^{t;\V{y},\V{x}}
-\frac{1}{2}\int_0^tj_{t-s}(1_{\Geb_n}\V{G})_{\smash{\hat{\V{b}}}^{t;\V{x},\V{y}}_s}
\Id\smash{\hat{\V{b}}}^{t;\V{x},\V{y}}_s,
\end{align*}
for every $n\in\NN$. 
From the pathwise uniqueness property of stochastic integrals 
(see, e.g., Kor.~1 on page 188 of \cite{HackenbrochThalmaier1994}, whose proof extends to the
Hilbert space-valued setting) we now infer that, for all natural numbers $m,n$ with $m>n$,
\begin{align*}
K_t^n(\V{x})=K_t^m(\V{x}),\;\text{$\PP$-a.s. on}\;
\{\tau_{\Geb_n}(\V{x})>t\}=\big\{\forall s\in[0,t]:\V{B}_s^{\V{x}}\in\Geb_n\big\},
\end{align*}
as well as
\begin{align*}
K_t^n(\V{x},\V{y})=K_t^m(\V{x},\V{y}),\;\text{$\PP$-a.s. on}\;
\{\tau_{\Geb_n}(t;\V{y},\V{x})=\infty\}=\big\{\forall s\in[0,t]:\V{b}_s^{t;\V{y},\V{x}}\in\Geb_n\big\}.
\end{align*}
Modulo changes on $\PP$-zero sets, we thus obtain well-defined random functions 
$K_t(\V{x})$ and $K_t(\V{x},\V{y})$ defined on
\begin{align*}
\{\tau_\Geb(\V{x})>t\}&=\bigcup_{n=1}^\infty\{\tau_{\Geb_n}(\V{x})>t\} \ \ \text{and} \ \ 
\{\tau_\Geb(t;\V{y},\V{x})=\infty\}=\bigcup_{n=1}^\infty\{\tau_{\Geb_n}(t;\V{y},\V{x})=\infty\},
\end{align*}
respectively, by setting
\begin{align}\label{Kxtaun}
K_t(\V{x})&:=K_t^n(\V{x})\;\,\quad\text{on $\{\tau_{\Geb_n}(\V{x})>t\}$,} 
\\\label{Kxytaun}
K_t(\V{x},\V{y})&:=K_t^n(\V{x},\V{y})\;\,\text{on $\{\tau_{\Geb_n}(t;\V{y},\V{x})=\infty\}$,}
\end{align}
for all $n\in\NN$. It is routine to check the independence of these definitions of the choice of the
exhausting sequence of open proper subsets $\{\Geb_n\}_{n\in\NN}$.

This gives a precise meaning to the random functions in \eqref{defKxintro} and \eqref{defKxyintro}. 
Quite obviously, 
they are indeed differences of two stochastic integrals individually defined in the above fashion.

The stochastic integrals in \eqref{defSxintro} and \eqref{defSxyintro} are defined in complete 
analogy; just replace $\HP$ by $\RR$ and ignore the isometries $j_s$ in the above construction.
Furthermore, it is well-known (see \cite[Lem.~2]{FarisSimon1975} and the estimations
\eqref{jeppe0} and \eqref{jeppe1}) that the path integrals of $V$ in \eqref{defSxintro} and 
\eqref{defSxyintro} are well-defined for a.e. $\V{x}$ and a.e. $(\V{x},\V{y})$, respectively.

Altogether, this gives a clear, canonical meaning to all terms in the Feynman-Kac integrands in 
\eqref{defWWintroscal} and \eqref{introdefWtxy}, which in the notation for the Weyl representation 
introduced in Subsect.~\ref{ssecFock} read
\begin{align}\label{fumio3}
W_t(\V{x})^*&=e^{-\ol{S_t(\V{x})}}\Gamma(j_0^*)\sW(-iK_t(\V{x}))\Gamma(j_t),
\\\label{fumio3b}
W_t(\V{x},\V{y})&=e^{-S_t(\V{x},\V{y})}\Gamma(j_t^*)\sW(iK_t(\V{x},\V{y}))\Gamma(j_0).
\end{align}

%%%%%%%%%%%%%%%%%%%%%%%%%%%%%%%%%%%%%%%%%%%%%%

\subsection{Feynman-Kac formulas for singular vector potentials and $\Geb=\RR^\nu$}\label{ssecFKsingAG}

In the next proof we shall work with the formulas \eqref{fumio3} and \eqref{fumio3b}, 
exploiting that 
\begin{align}\label{fumio1}
\HP\ni f\longmapsto&\Gamma(j_s^*)\sW(f)\Gamma(j_t) \ \text{is strongly continuous}. 
\\\label{fumio2}
\|&\Gamma(j_s^*)\sW(f)\Gamma(j_t)\|_{\LO(\sF)}\le1,\quad f\in\HP,
\end{align}
for all $s,t\ge0$. These two statements follow from the remarks in Subsect.~\ref{ssecFock}.

\begin{prop}\label{propFKsingAG}
Let $\V{A}\in L_\loc^2(\RR^\nu,\RR^\nu)$, $\V{G}\in L_\loc^2(\RR^\nu,\HP_{\RR}^\nu)$, and let
$V\ge0$ be in $C_b(\RR^\nu,\RR)$. Pick some $t>0$ and $\Psi\in L^2(\RR^\nu,\sF)$. Then
\begin{align}\label{FKregV}
(e^{-tH}\Psi)(\V{x})&=\EE\big[W_t(\V{x})^*\Psi(\V{B}_t^{\V{x}})\big],\quad
\text{a.e. $\V{x}\in\RR^\nu$.}
\end{align}
Furthermore,
\begin{align}\label{FKregVxy}
(e^{-tH}\Psi)(\V{x})&=\int_{\RR^\nu}p_t(\V{x},\V{y})
\EE\big[W_t(\V{x},\V{y})\Psi(\V{y})\big]\Id\V{y},\quad\text{a.e. $\V{x}\in\RR^\nu$.}
\end{align}
\end{prop}

In \eqref{FKregV} and \eqref{FKregVxy} we again drop the subscript $\RR^\nu$ in the notation for 
Pauli-Fierz operators on $\RR^\nu$; recall the remarks preceding Thm.~\ref{thmstrresconv}. The
completely real subspace $\HP_\RR\subset\HP$ has the properties mentioned below \eqref{introGR}.

\begin{proof}
{\em Step~1: Construction of approximating vector potentials.}
Define the standard mollifier $\rho_n$ as in \eqref{klausi1} and \eqref{klausi2}. Pick some
$\chi\in C^\infty(\RR,\RR)$ with $0\le\chi\le1$, $\chi=1$ on $(-\infty,1]$ and $\chi=0$ on 
$[2,\infty)$. For every $n\in\NN$, define $\chi_n(\V{x}):=\chi(|\V{x}|/n)$, $\V{x}\in\RR^\nu$, and
\begin{align*}
\V{A}^n&:=\rho_n*(\chi_n\V{A}),\quad\V{G}^n:=\rho_n*(\chi_n1_{[\nf{1}{n},n]}(\omega)\V{G}).
\end{align*}
 Then $\V{A}^n\in C_0^\infty(\RR^\nu,\RR^\nu)$ and every 
 $\V{G}^n\in C_0^\infty(\RR^\nu,\HP^\nu)$ with  $n\in\NN$ fulfills Hyp.~\ref{hypGreg} and 
 Hyp.~\ref{hypGC}. Defining $W_t^n(\V{x})$, $W^n_t(\V{x},\V{y})$, and $H^n$ by putting the pair
$(\V{A}^n,\V{G}^n)$ in place of $(\V{A},\V{G})$ in the construction of $W_t(\V{x})$, 
$W_t(\V{x},\V{y})$, and $H$, respectively, we therefore have the following Feynman-Kac formulas
for every $n\in\NN$,
\begin{align}\label{FKAnGn}
(e^{-tH^n}\Psi)(\V{x})&=\EE\big[W_t^n(\V{x})^*\Psi(\V{B}_t^{\V{x}})\big],
\quad\text{a.e. $\V{x}\in\RR^\nu$,}
\end{align}
as well as
\begin{align}\label{FKAnGnb}
(e^{-tH^n}\Psi)(\V{x})&=\int_{\RR^\nu}p_t(\V{x},\V{y})\EE\big[W_t^n(\V{x},\V{y})\Psi(\V{y})\big]
\Id\V{y},\quad\text{a.e. $\V{x}\in\RR^\nu$.}
\end{align}
Furthermore, the following limit relations hold as $n\to\infty$,
\begin{align}\label{limitAnGn}
&\V{A}^n\to\V{A} \ \text{in $L^2_\loc(\RR^\nu,\RR^\nu)$,} \qquad
\V{G}^n\to\V{G} \ \text{in $L^2_\loc(\RR^\nu,\HP^\nu)$.}
\end{align}
Here the first one is standard, while the second one follows from the following remarks:

Let $C\subset\RR^\nu$ be compact and choose $n_0\in\NN$ so large that
$$
C_1:=\big\{\V{x}\in\RR^\nu\big|\,\dist(\V{x},C)\le1\big\}\subset\{\chi_{n_0}=1\}.
$$
For every $\V{x}$, we have 
$\|(1-1_{[\nf{1}{n},n]}(\omega))\V{G}_{\V{x}}\|_{\HP^\nu}\to0$, $n\to\infty$,
by dominated convergence. Therefore, the generalized Minkowski inequality and the dominated
convergence theorem further imply
\begin{align}\nonumber
&\bigg(\int_C\Big\|\int_{\RR^\nu}\rho_n(\V{x}-\V{y})\chi_n(\V{y})(1-1_{[\nf{1}{n},n]}(\omega))
\V{G}_{\V{y}}\Id\V{y}\Big\|_{\HP^\nu}^2\Id\V{x}\bigg)^\eh
\\\nonumber
&\le\bigg(\int_C\Big(\int_{\RR^\nu}\rho_n(\V{z})\|(1-1_{[\nf{1}{n},n]}(\omega))\V{G}_{\V{x}-\V{z}}
\|_{\HP^\nu}\Id\V{z}\Big)^2\Id\V{x}\bigg)^\eh
\\\nonumber
&\le\int_{\RR^\nu}\rho_n(\V{z})\bigg(\int_{C-\V{z}}
\|(1-1_{[\nf{1}{n},n]}(\omega))\V{G}_{\V{x}}\|_{\HP^\nu}^2\Id\V{x}\bigg)^\eh\Id\V{z}
\\\label{limitGn2}
&\le\bigg(\int_{C_1}
\|(1-1_{[\nf{1}{n},n]}(\omega))\V{G}_{\V{x}}\|_{\HP^\nu}^2\Id\V{x}\bigg)^\eh
\xrightarrow{\;\;n\to\infty\;\;}0,
\end{align}
Here we also used that every $\rho_n$ is supported in the unit ball, which permitted to drop
$\chi_n$ for all $n\ge n_0$ in the first step and to replace $C-\V{z}$ by the larger set $C_1$
in the last step. Likewise,
\begin{align}\label{limitGn3}
\int_C\|\rho_n*(\chi_n\V{G})-\V{G}\|_{\HP^\nu}^2\Id\V{x}
&=\int_C\|\rho_n*\V{G}-\V{G}\|_{\HP^\nu}^2\Id\V{x}\xrightarrow{\;\;n\to\infty\;\;}0,
\end{align}
where the equality holds for $n\ge n_0$ and the convergence is a special case of \eqref{klausi3}.
Now the second relation in \eqref{limitAnGn} follows from \eqref{limitGn2} and \eqref{limitGn3}.

{\em Step~2. Convergence of the left hand side of the Feynman-Kac formulas.}
Fix $t>0$ in the rest of this proof.
Thm.~\ref{thmstrresconv} shows that $H^n\to H$, $n\to\infty$, in strong resolvent sense, which implies 
the strong convergence $e^{-tH^n}\to e^{-tH}$. Therefore, there exist integers
$1\le n_1<n_2<\ldots$ such that
\begin{align}\label{FKAnGn2}
(e^{-tH^{n_\ell}}\Psi)(\V{x})&\xrightarrow{\;\;\ell\to\infty\;\;}
(e^{-tH}\Psi)(\V{x}),\quad\text{for a.e. $\V{x}\in\RR^\nu$.}
\end{align}

{\em Step~3. Application of the dominated convergence theorem.}
Define $K^n_t(\V{x})$ and $K^n_t(\V{x},\V{y})$ by putting $\V{G}^n$ in place of $\V{G}$ in 
the formulas for $K_t(\V{x})$ and $K_t(\V{x},\V{y})$, respectively. Likewise, define
$S^n_t(\V{x})$ and $S^n_t(\V{x},\V{y})$ by substituting $\V{A}^n$ for $\V{A}$ in 
the expressions for $S_t(\V{x})$ and $S_t(\V{x},\V{y})$, respectively. According to
Lem.~\ref{lemexstochint} we may in fact fix zero sets $N\subset\RR^\nu$ and $N'\subset\RR^{2\nu}$ 
in the rest of this proof such that these random functions are well-defined, for all
$\V{x}\in\RR^\nu\setminus N$ and $(\V{x},\V{y})\in\RR^{2\nu}\setminus N'$, respectively.
Combining \eqref{limitAnGn},
Thm.~\ref{thmdomconvstoch}, and Lem.~\ref{lemdomfct} we now find a subsequence
$\{m_\ell\}_{\ell\in\NN}$ of the index sequence $\{n_\ell\}_{\ell\in\NN}$ such that, 
as $\ell\to\infty$,
\begin{align}\label{eskild2d}
&\text{$S_t^{m_\ell}(\V{x})\to S_t(\V{x})$ and $K_t^{m_\ell}(\V{x})\to K_t(\V{x})$  in probability},
\\\label{eskild2e}
&\text{$S_t^{m_\ell}(\V{x},\V{y})\to S_t(\V{x},\V{y})$
and $K_t^{m_\ell}(\V{x},\V{y})\to K_t(\V{x},\V{y})$  in probability,}
\end{align}
for all $\V{x}\in\RR^\nu\setminus N$ in the first line and all
$(\V{x},\V{y})\in\RR^{2\nu}\setminus N'$ in the second.

{\em Step~4. Convergence along a subsequence of the right hand side of \eqref{FKAnGn}.}
We fix $\V{x}\in\RR^\nu\setminus N$. Recall that convergence in probability implies 
$\PP$-a.s. convergence along a subsequence. By virtue of \eqref{eskild2d}
we therefore find a subsequence $(i_\ell)_{\ell\in\NN}$ of the index sequence 
$(m_\ell)_{\ell\in\NN}$ such that, $\PP$-a.s.,
\begin{align*}
K_t^{i_\ell}(\V{x})\xrightarrow{\;\;\ell\to\infty\;\;}K_t(\V{x}) \ \text{in $\hat{\HP}$,}\quad\text{and}
\quad S_t^{i_\ell}(\V{x})&\xrightarrow{\;\;\ell\to\infty\;\;}S_t(\V{x}) \ \text{in $\CC$.}
\end{align*}
Picking a representative $\Psi(\cdot)$ of $\Psi\in L^2(\RR^\nu,\sF)$ and taking
\eqref{fumio3}, \eqref{fumio1}, and \eqref{fumio2} into account, we deduce that
$W_t^{i_\ell}(\V{x})^*\Psi(\V{B}_t^{\V{x}})\to
W_t(\V{x})^*\Psi(\V{B}_t^{\V{x}})$, as $\ell\to\infty$, $\PP$-a.s.,
with the pointwise domination
$\|W_t^{i_\ell}(\V{x})^*\Psi(\V{B}_t^{\V{x}})\|_{\sF}
\le \|\Psi(\V{B}_t^{\V{x}})\|_{\sF}\in L^1(\PP)$, for every $\ell\in\NN$.
Thus, by dominated convergence,
\begin{align}\label{FKregVconv}
\EE\big[W_t^{i_\ell}(\V{x})^*\Psi(\V{B}_t^{\V{x}})\big]
\xrightarrow{\;\;\ell\to\infty\;\;}\EE\big[W_t(\V{x})^*\Psi(\V{B}_t^{\V{x}})\big]
\quad\text{in $\sF$,}
\end{align}
where $\V{x}\in\RR^\nu\setminus N$ was arbitrary. Combining \eqref{FKAnGn}, \eqref{FKAnGn2}, 
and \eqref{FKregVconv} we arrive at the Feynman-Kac formula \eqref{FKregV}.

{\em Step~5. Convergence along a subsequence of the right hand side of \eqref{FKAnGnb}.}
In this step we cannot just mimic the argument of the preceding
one because any choice of subsequence along which the convergences in \eqref{eskild2e} 
 hold $\PP$-a.s. would not only depend on $\V{x}$ but also on $\V{y}$.

Let us fix a representative $\Psi(\cdot):\RR^\nu\to\sF$ of $\Psi\in L^2(\RR^\nu,\sF)$ in the rest of this
proof. We also fix $(\V{x},\V{y})\in\RR^{2\nu}\setminus N'$ for the moment. 
Then the following  map is continuous,
\begin{align*}
\CC\times\hat{\HP}\ni(z,f)\longmapsto F(z,f):=p_{t}(\V{x},\V{y})
e^{-z}\Gamma(j_t^*)\sW(if)\Gamma(j_0)\Psi(\V{y})\in\sF.
\end{align*}
Since $W^n_{t}(\V{x},\V{y})\Psi(\V{y})=F(S^n_t(\V{x},\V{y}),K^n_t(\V{x},\V{y}))$, 
for every $n\in\NN$, and similarly for the limit processes, this permits to get
\begin{align*}
W^{m_\ell}_{t}(\V{x},\V{y})\Psi(\V{y})\xrightarrow{\;\;\ell\to\infty\;\;}
W_{t}(\V{x},\V{y})\Psi(\V{y})\quad\text{in $\sF$ and in probability,}
\end{align*}
employing \eqref{eskild2e}. We further have the uniform bounds
\begin{align*}
\|W^{m_\ell}_{t}(\V{x},\V{y})\Psi(\V{y})\|_{\sF}&\le \|\Psi(\V{y})\|_{\sF},\quad\ell\in\NN,
\end{align*}
showing in particular that the sequence $\{W^{m_\ell}_{t}(\V{x},\V{y})\Psi(\V{y})\}_{\ell\in\NN}$ in 
$L^1(\Omega,\sF;\PP)$ is uniformly integrable. Hence, by Vitali's theorem in its vector-valued version,
\begin{align}\label{eskild5}
\EE\big[W^{m_\ell}_{t}(\V{x},\V{y})\Psi(\V{y})\big]\xrightarrow{\;\;\ell\to\infty\;\;}
\EE\big[W_{t}(\V{x},\V{y})\Psi(\V{y})\big].
\end{align}

Now, for a.e. $\V{x}$, the cut $N_{\V{x}}':=\{\V{y}\in\RR^\nu|(\V{x},\V{y})\in N'\}$ has
Lebesgue measure zero. Let us fix some $\V{x}\in\RR^\nu$ for which this is the case in the rest of the
proof. Then \eqref{eskild5} holds for a.e. $\V{y}$ and we have the dominations
\begin{align*}
p_{t}(\V{x},\V{y})\big\|\EE\big[W^{m_\ell}_{t}(\V{x},\V{y})\Psi(\V{y})\big]\big\|_\sF
&\le p_{t}(\V{x},\V{y})\|\Psi(\V{y})\|_\sF,\quad\text{a.e. $\V{y}\in\RR^\nu$,}\,\ell\in\NN,
\end{align*}
where $\V{y}\mapsto p_{t}(\V{x},\V{y})\|\Psi(\V{y})\|_{\sF}$ is in $L^1(\RR^\nu)$. The
dominated convergence theorem, \eqref{FKAnGnb}, and \eqref{FKAnGn2} now imply the 
desired formula \eqref{FKregVxy}. 
\end{proof}

%%%%%%%%%%%%%%%%%%%%%%%%%%%%%%%%%%%%%%%%%%%%%%%%%

\subsection{Feynman-Kac formulas for singular coefficients and general open $\Geb$}\label{ssecFKsingAGGeb}

\noindent
We are now in a position to prove our main theorem. We start by applying the results of
Sect.~\ref{secFKforD}, which is possible when $\V{A}$ and $\V{G}$ have
locally square-integrable extension to the whole $\RR^\nu$.

\begin{prop}\label{propFKsing}
Let $\V{A}\in L_\loc^2(\RR^\nu,\RR^\nu)$, $\V{G}\in L_\loc^2(\RR^\nu,\HP^\nu)$,
and $V\in C_b(\RR^\nu,\RR)$. Pick some $t>0$ and $\Psi\in L^2(\RR^\nu,\sF)$. Then,
for a.e. $\V{x}\in\RR^\nu$,
\begin{align}\nonumber
(e^{-tH_\Geb}\Psi)(\V{x})&=\EE\big[1_{\{\tau_\Geb(\V{x})>t\}}
W_t(\V{x})^*\Psi(\V{B}_t^{\V{x}})\big]
\\\label{FKregV2}
&=\int_{\Geb}p_t(\V{x},\V{y})\EE\big[1_{\{\tau_\Geb(t;\V{y},\V{x})=\infty\}}
W_t(\V{x},\V{y})\Psi(\V{y})\big]\Id\V{y}.
\end{align}
\end{prop}

\begin{proof}
It suffices to check the postulates in Sect.~\ref{secFKforD} when we set
$\mathfrak{q}_{\RR^\nu}:=\mathfrak{h}_{\RR^\nu}$ and
$\mathfrak{q}_{\Geb}:=\mathfrak{h}_{\Geb,\mathrm{D}}'$. That these two forms fulfill
Hyp.~\ref{hypabD} has, however, already been observed in Prop.~\ref{propDFKhyp}.
The validity of Hyp.~\ref{hypDFKgen} and Hyp.~\ref{hypDFKbridgegen} follows from
Prop.~\ref{propFKsingAG}.
\end{proof}

\begin{proof}[Proof of Thm.~\ref{thmFKintro1}.]
First, we additionally assume that $V\in C_b(\RR^\nu,\RR)$.
To infer our main theorem from Prop.~\ref{propFKsing} in this case,
we apply an idea from \cite[\textsection4]{Simon1978}:
Set $\Geb_n:=\{\V{x}\in\Geb|\dist(\V{x},\Geb^c)>1/n\}$ and
$\V{A}^n:=1_{\Geb_n}\V{A}$, $\V{G}^n:=1_{\Geb_n}\V{G}$, for all $n\in\NN$.
Extend $\V{A}^n$ and $\V{G}^n$ to functions on $\RR^\nu$ by setting then equal to zero
on $\Geb^c$. Then $\V{A}^n\in L^2_\loc(\RR^\nu,\RR^\nu)$ and 
$\V{G}\in L^2_\loc(\RR^\nu,\HP^\nu)$. Let $\mathfrak{h}_n$ denote the minimal
Pauli-Fierz form on $\Geb_n$ defined by means of $\V{A}^n$ and $\V{G}^n$. Then it is clear that
$\dom(\mathfrak{h}_n)\subset\dom(\mathfrak{h}_{m})\subset
\mathfrak{h}_{\Geb,\mathrm{D}}$, $m>n$, and
$\mathfrak{h}_{\Geb,\mathrm{D}}[\Phi]=\lim_{n<m\to\infty}\mathfrak{h}_m[\Phi]$,
for all $\Phi\in\dom(\mathfrak{h}_n)$ and $n\in\NN$, where functions on $\Geb_n$ are
tacitly extended by $0$ to larger subsets of $\Geb$.
Thus, \cite[Thm.~4.1 and Thm.~4.2]{Simon1978} imply that
\begin{align}\label{maiken99}
e^{-tH_{\Geb_n}}(\Psi\restr_{\Geb_n})\xrightarrow{\;\;n\to\infty\;\;}e^{-tH_{\Geb}}\Psi\quad
\text{in $L^2(\Geb,\sF)$,}
\end{align}
where the $e^{-tH_{\Geb_n}}(\Psi\restr_{\Geb_n})$ are interpreted as functions on $\Geb$
that equal $0$ on $\Geb\setminus\Geb_n$. Along a suitable subsequence, the convergence
in \eqref{maiken99} also holds pointwise a.e. on $\Geb$. 
On the other hand, Prop.~\ref{propFKsing} in conjunction with \eqref{Kxtaun},
\eqref{Kxytaun}, and analogous relations for the complex actions $S_t(\V{x})$ and
$S_t(\V{x},\V{y})$ implies
\begin{align}\nonumber
(e^{-tH_{\Geb_n}}(\Psi\restr_{\Geb_n}))(\V{x})
&=\EE\big[1_{\{\tau_{\Geb_n}(\V{x})>t\}}W_t(\V{x})^*\Psi(\V{B}_t^{\V{x}})\big]
\\\label{maiken100}
&=\int_{\Geb_n}p_t(\V{x},\V{y})
\EE\big[1_{\{\tau_{\Geb_n}(t;\V{x},\V{y})=\infty\}}W_t(\V{x},\V{y})\Psi(\V{y})\big]\Id\V{y},
\end{align}
for a.e. $\V{x}\in\Geb_n$ and all $n\in\NN$.
Here $1_{\{\tau_{\Geb_n}(\V{x})>t\}}\to 1_{\{\tau_{\Geb}(\V{x})>t\}}$ and
$1_{\{\tau_{\Geb_n}(t;\V{x},\V{y})=\infty\}}\to1_{\{\tau_{\Geb}(t;\V{x},\V{y})=\infty\}}$
pointwise on $\Omega$, as $n\to\infty$, for all $\V{x},\V{y}\in\Geb$. Hence, by dominated
convergence, the expectation in the first line of \eqref{maiken100} and the member in the second
line of \eqref{maiken100} converge to the corresponding terms in \eqref{FKregV2}, 
for every $\V{x}\in\Geb$.

For merely measurable, bounded $V\ge0$, all statements of Thm.~\ref{thmFKintro1} now follow from a 
standard mollifying procedure and, after that, they can be extended to locally integrable $V\ge0$
by approximation with $V\wedge n$, $n\in\NN$; see, e.g., the proof of
\cite[Thm.~11.3]{GMM2017} for more details.
\end{proof}

\begin{proof}[Proof of Cor.~\ref{cornegpart}]
The first assertion in the corollary follows from the discussion in \cite[\textsection4]{Matte2017}. 
To prove the second one,
we start by observing that Thm.~\ref{thmFKintro1} and Rem.~\ref{remkernig} extend trivially to locally 
integrable potentials that are bounded from below and in particular to every $V-U\wedge n$ 
with $n\in\NN$. Furthermore, a monotone convergence theorem
for quadratic forms \cite[Thm.~VIII.3.11]{Kato}
implies that $H_{\Geb}^{U\wedge n}\to H_{\Geb}^{U}$ in strong resolvent
sense, as $n\to\infty$. Let $t>0$ and $\Psi\in L^2(\Geb,\sF)$. Then
$$
e^{-tH_{\Geb}^{U\wedge n_\ell}}\Psi\xrightarrow{\;\;\ell\to\infty\;\;} 
e^{-tH_{\Geb}^{U}}\Psi\quad\text{ a.e. on $\Geb$,}
$$
for a suitable subsequence $\{n_\ell\}_{\ell\in\NN}$ of $\{n\}_{n\in\NN}$. In view of 
\eqref{WWWbd} and the dominated convergence theorem it therefore remains to verify the inequality in
\begin{align*}%\label{tatjana0}
\int_\Geb p_t(\V{x},\V{y})\EE\Big[1_{\{\tau_\Geb(\V{y},\V{x})=\infty\}}
e^{\int_0^t U(\V{b}_s^{t;\V{y},\V{x}})\Id s}\Big]\|\Psi(\V{y})\|_{\sF}\Id\V{y}
&=\EE\Big[e^{\int_0^t U(\V{B}_s^{\V{x}})\Id s}\eta_{\Geb,t}^{\V{x}}\Big]<\infty,
\end{align*}
for a.e. $\V{x}\in\Geb$, where 
$\eta_{\Geb,t}^{\V{x}}:=1_{\{\tau_\Geb(\V{x})>t\}}\|\Psi(\V{B}_t^{\V{x}})\|_{\sF}$.
(The equality in the previous relation is true for every $\V{x}\in\Geb$
and follows upon substituting $U$ by $U\wedge n$ and applying the monotone convergence theorem.)

We now argue similarly as in \cite{Voigt1986}:
Denoting the Dirichlet-Laplacian on $\Geb$ by $\Delta_\Geb$, we know 
\cite[Thm.~VIII.3.11]{Kato} that the operators $-\Delta_\Geb/2-U\wedge n$
have a limit in the strong resolvent sense. Denoting this limit by $L$,
we find a subsequence $\{m_\ell\}_{\ell\in\NN}$ of the index
sequence $\{n_\ell\}_{\ell\in\NN}$ such that, for a.e. $\V{x}\in\Geb$,
\begin{align}\label{tatjana1}
(e^{-t(-\Delta_\Geb/2-U\wedge m_\ell)}\|\Psi\|_{\sF})(\V{x})
\xrightarrow{\;\;\ell\to\infty\;\;}(e^{-tL}\|\Psi\|_{\sF})(\V{x})<\infty.
\end{align}
The monotone convergence theorem now implies that
\begin{align}\label{tatjana2}
\EE\Big[e^{\int_0^t U(\V{B}_s^{\V{x}})\Id s}\eta_{\Geb,t}^{\V{x}}\Big]
&=\lim_{\ell\to\infty}
\EE\Big[e^{\int_0^t(U\wedge m_\ell)(\V{B}_s^{\V{x}})\Id s}\eta_{\Geb,t}^{\V{x}}\Big]<\infty,
\end{align}
 for a.e. $\V{x}\in\Geb$, since, again for a.e. $\V{x}\in\Geb$, the expectations to the right in 
\eqref{tatjana2} are equal to the vectors to the left in \eqref{tatjana1}.
\end{proof}

%%%%%%%%%%%%%%%%%%%%%%%%%%%%%%%%%%%%%%%%%%%%%%%%%
%%%%%%%%%%%%%%%%%%%%%%%%%%%%%%%%%%%%%%%%%%%%%%%%%
%%%%%%%%%%%%%%%%%%%%%%%%%%%%%%%%%%%%%%%%%%%%%%%%%

\appendix

\section{A useful rule for vector-valued conditional expectations}\label{appusefulrule}

\noindent
The following lemma should be well-known, also in the infinite dimensional setting, but we could not
find an appropriate reference. Therefore, we prove it for the convenience of the reader.

\begin{lem}\label{lemusefulrule}
Let $(X,\fA,P)$ be a probability space, $\fC$ be a sub-$\sigma$-algebra of $\fA$, 
and $Y$ and $Z$ separable Banach spaces equipped with their Borel $\sigma$-algebras 
$\fB(Y)$ and $\fB(Z)$, respectively.
Let $f:X\times Y\to Z$ be a function such that $f(\cdot,y):X\to Z$ is 
Bochner-Lebesgue integrable (in particular
$\fA$-$\fB(Z)$-measurable) and $\fC$-independent for every $y\in Y$, and such that
$f(x,\cdot):Y\to Z$ is continuous for every $x\in X$. 
(This implies that $f$ is $(\fA\otimes\fB(Y))$-$\fB(Z)$-measurable.) Define
\begin{align*}
\phi(y):=E[f(\cdot,y)]:=\int_X f(x,y)\Id P(x),\quad y\in Y.
\end{align*}
(Then $\phi:Y\to Z$ is in any case Borel measurable.) Finally, let $g:X\to Y$ be
$\fC$-$\fB(Y)$-measurable and assume that 
$$
X\ni x\longmapsto h(x):=f(x,g(x))\in Z
$$ 
is Bochner-Lebesgue integrable. Then
\begin{align*}
E^{\fC}[h]&=\phi(g),\quad \text{$P$-a.s.,}
\end{align*}
where $E^{\fC}$ denotes a version of the $Z$-valued conditional expectation with respect to $P$ given 
the hypothesis $\fC$.
\end{lem}

\begin{proof}
Let $\chi\in C(\RR,\RR)$ be such that $0\le\chi\le1$ on $\RR$,
$\chi=1$ on $(-\infty,1]$, and $\chi=0$ on $[2,\infty)$.
Put $f_n:=\chi(\|f\|_Z/n)f$, $n\in\NN$, so that each $f_n$ enjoys all properties of 
$f$ mentioned in the statement as well, and so that $\|f_n\|_Z\le2n$
and $f_n\to f$, $n\to\infty$, pointwise on $X\times Y$. Set $h_n(x):=f_n(x,g(x))$, $x\in X$, and
$\phi_n(y):=E[f_n(\cdot,y)]$, $y\in Y$. Then we have the dominations
$\|f_n(\cdot,y)\|\le\|f(\cdot,y)\|$ and $\|h_n\|_Z\le\|h\|_Z$ on $X$, for every $n\in\NN$.
Hence, the dominated convergence theorem for the Bochner-Lebesgue integral implies
that $\phi_n(y)\to\phi(y)$, $n\to\infty$, for every $y\in Y$, while
the dominated convergence theorem for $Z$-valued conditional expectations implies that 
$\EE^{\fC}[h_n]\to\EE^{\fC}[h]$, $n\to\infty$, $P$-a.s. Therefore, it only remains to show that
$\EE^{\fC}[h_n]=\phi_n(g)$ holds $P$-a.s., for each fixed $n\in\NN$. Or, put differently,
we may assume without loss of generality that $f$ is bounded, which we shall do in the rest of this proof.

There exists a sequence of $\fC$-$\fB(Y)$-measurable functions $(g_n)_{n\in\NN}$ such that
the image $g_n(X)$ is finite, for every $n\in\NN$, and such that $g_n\to g$, $n\to\infty$,
pointwise on $X$. Let $n\in\NN$. Then $g_n$ has a standard
representation $g_n=\sum_{i=1}^{k_n}1_{A^n_i}y_i^n$ for suitable $k_n\in\NN$,
$y_1^n,\ldots,y_{k_n}^n\in Y$, and suitable {\em disjoint} $A_1^n,\ldots,A_{k_n}^n\in\fC$
such that $A_1^n\cup\dots\cup A_{k_n}^n=X$.  Then
\begin{align*}
\tilde{h}_n(x):=f(x,g_n(x))=\sum_{i=1}^{k_n}1_{A_i^n}(x)f(x,y_i^n),\quad x\in X.
\end{align*}
Since $A_i^n\in\fC$ and since $f(\cdot,y_i^n):X\to Z$ is $\fC$-independent, 
well-known computation rules for the conditional expectation now imply
\begin{align*}
E^{\fC}[\tilde{h}_n]=\sum_{i=1}^{k_n}1_{A_i^n}\phi(y_i^n)=\phi(g_n),\quad\text{$P$-a.s.,}
\end{align*}
where we  again used that $y_i^n=g_n$ on $A_i^n$ in the second equality. Furthermore,
by our present assumptions on $f$, the functions $\tilde{h}_n$, ${n\in\NN}$, are uniformly bounded,
and thanks to the continuity of $y\mapsto f(x,y)$ for each $x$, we know that 
$\tilde{h}_n\to h$, $n\to\infty$, pointwise on $X$. Hence,
$E^{\fC}[\tilde{h}_n]\to E^{\fC}[h]$, $n\to\infty$, $P$-a.s., by 
the dominated convergence theorem for $Z$-valued conditional expectations.
Finally, we observe that $\phi:Y\to Z$ is continuous by the boundedness of $f$ and
dominated convergence. Thus, $\phi(g_n)\to\phi(g)$, $n\to\infty$, pointwise on $X$.
\end{proof}

\begin{ex}\label{exusefulrule}
Let $(X,\fA,P)$ and $\fC$ be as in Lem.~\ref{lemusefulrule}.
Let $Z$ be a separable Hilbert space, $A(\V{y}):X\to\LO(Z)$ be measurable and
separably valued, for every $\V{y}\in\RR^\nu$, such that 
$\RR^\nu\ni\V{y}\mapsto(A(\V{y}))(x)$ is strongly continuous for all $x\in X$.
Suppose that $A(\V{y})$ is $\fC$-independent and
let $g:=(\V{q},\Psi):X\to\RR^\nu\times Z$ be $\fC$-measurable
with $\int_X\|\Psi\|_{Z}\Id P<\infty$. Finally, assume there exists $C>0$ such that
$\|A(y)\|\le C$, $\PP$-a.s., for every $y\in\RR^\nu$.
Then we can apply Lem.~\ref{lemusefulrule} to the function $f$ given by
$f(x,\V{y},\psi):=A(y)\psi$, $(\V{y},\psi)=\RR^\nu\times Z(=:Y)$ with
$\phi(\V{y},\psi)=E[f(\cdot,\V{y},\psi)]=E[A(\V{y})]\psi$. That is,
\begin{align*}
E^{\fC}[A(\V{q})\Psi]&=E[A(\V{y})]\big|_{\V{y}=\V{q}}\Psi.
\end{align*}
\end{ex}

%%%%%%%%%%%%%%%%%%%%%%%%%%%%%%%%%%%%%%%%%%%%%%%%%
%%%%%%%%%%%%%%%%%%%%%%%%%%%%%%%%%%%%%%%%%%%%%%%%%
%%%%%%%%%%%%%%%%%%%%%%%%%%%%%%%%%%%%%%%%%%%%%%%%%

\subsection*{Acknowledgement}
The author is grateful for support by the Independent Research Fund Denmark via the 
project grant ``Mathematical Aspects of Ultraviolet Renormalization'' (8021-00242B).

%%%%%%%%%%%%%%%%%%%%%%%%%%%%%%%%%%%%%%%%%%%%%%%%%
%%%%%%%%%%%%%%%%%%%%%%%%%%%%%%%%%%%%%%%%%%%%%%%%%
%%%%%%%%%%%%%%%%%%%%%%%%%%%%%%%%%%%%%%%%%%%%%%%%%


\begin{thebibliography}{42}

\bibitem{BHL2000}
Broderix, K., Hundertmark, D., Leschke, H.:
\newblock Continuity properties of Schr\"odinger semigroups with magnetic fields.
\newblock {\em Rev. Math. Phys.} \textbf{12}, 181--225 (2000)

\bibitem{BroderixLeschkeMueller2004}
{Broderix, K., Leschke, H., M\"{u}ller, P.:} {Continuous integral kernels for unbounded {S}chr\"{o}dinger
semigroups and their spectral projections.} {\em J. Funct. Anal.} \textbf{212}, {287--323} (2004)

\bibitem{CostabelDauge2000}
Costabel, M., Dauge, M.: Singularities of electromagnetic fields in polyhedral domains.
{\em Arch. Rational Mech. Anal.} \textbf{151},  221--276 (2000)

\bibitem{daPrZa2014}
Da Prato, G., Zabczyk, J.: {Stochastic equations in infinite dimensions.}
Second Edition. Encyclopedia of Mathematics and Its Applications, vol. 152,
Cambridge University Press, Cambridge, 2014.

\bibitem{Dutra2005}
Dutra, S.M.: {\em Cavity quantum electrodynamics. The strange theory of light in a box.}
John Wiley \& Sons, Hoboken, New Jersey, 2005.

\bibitem{FarisSimon1975}
Faris, W., Simon, B.: Degenerate and non-degenerate ground states for Schr\"odinger operators.
{\em Duke Math. J.} \textbf{42}, 559--581 (1975)

\bibitem{Gueneysu2017}
{G\"{u}neysu, B.:} {\em Covariant {S}chr\"{o}dinger semigroups on {R}iemannian manifolds.}
{Operator Theory: Advances and Applications}, {vol. 264}, {Birkh\"{a}user/Springer, Cham}, {2017}.

\bibitem{GMM2017}
G\"{u}neysu, B., Matte, O., M{\o}ller, J.S.:
Stochastic differential equations for models of non-relativistic matter interacting
with quantized radiation fields.
{\em Probab. Theory Relat. Fields} \textbf{167}, 817--915 (2017)

\bibitem{HackenbrochThalmaier1994}
Hackenbroch, W., Thalmaier, A.: {\em Stochastische Analysis.} Teubner, Stuttgart, 1994.

\bibitem{HaussmannPardoux1986}
Haussmann, U.G., Pardoux, E.: {Time reversal of diffusions.}
 {\em Ann. Prob.} \textbf{14}, 1188--1205 (1986) 
 
 \bibitem{Hinz2015}
 Hinz, M.: Magnetic energies and Feynman-Kac-It\^{o} formulas for symmetric Markov processes.
 {\em Stoch. Anal. Appl.} \textbf{33}, 1020--1049  (2015)

\bibitem{Hiroshima1996}
Hiroshima, F.:
\newblock {Diamagnetic inequalities for systems of nonrelativistic particles with a quantized field.}
\newblock {\em Rev. Math. Phys.} \textbf{8}, 185--203 (1996)

\bibitem{Hiroshima1997}
{Hiroshima, F.:}
\newblock {Functional integral representation of a model in quantum electrodynamics}.
\newblock {\em Rev. Math. Phys.} \textbf{9}, {489--530} (1997)

\bibitem{HiroshimaLorinczi2008}
Hiroshima, F., L\H{o}rinczi, J.:
\newblock Functional integral representations of the 
Pauli-Fierz model with spin $1/2$.
\newblock {\em J. Funct. Anal.} \textbf{254}, 2127--2185 (2008)

\bibitem{HiroshimaMatte2019}
Hiroshima, F., Matte, O.: Ground states and associated path measures in the renormalized Nelson 
model. Rev. Math. Phys. \textbf{33},  2250002, 84 pages (2022)

\bibitem{HKNSV2006}
Hundertmark, D., Killip, R., Nakamura, S., Stollmann, P., Veseli\'{c}, I.:
Bounds on the spectral shift function and the density of states.
{\em Commun. Math. Phys.} \textbf{262}, 489--503 (2006)

\bibitem{HundertmarkSimon}
Hundertmark, D., Simon, B.:
\newblock {A diamagnetic inequality for semigroup differences.}
\newblock {\em J. Reine Angew. Math.} \textbf{571}, 107--130 (2004)

\bibitem{Kato1978}
Kato, T.: {Remarks on Schr\"{o}dinger operators with vector potentials.}
{\em Integr. Equat. Oper. Theory} \textbf{1}, 103--113 (1978)

\bibitem{Kato}
Kato, T.: {\em Perturbation theory of linear operators.}
Classics in Mathematics, Springer, Berlin-Heidelberg, 1995. Reprint of the 1980 edition.

\bibitem{KMS2013}
K\"{o}nenberg, M., Matte, O., Stockmeyer, E.: Hydrogen-like atoms in relativistic QED.
\newblock In: Siedentop, H. (Editor):
{\em Complex Quantum Systems: Theory of Large Coulomb Systems.} 
Singapore, February 2010. Lecture Note Series, Institute for Mathematical Sciences,
National University of Singapore, vol. {24}, World Scientific, Singapore, 2013, pp. 219--290.

\bibitem{Leinfelder1983}
Leinfelder, H.: Gauge invariance of Schr\"{o}dinger operators and related spectral properties.
{\em J. Operator Theory} \textbf{9}, 163--179 (1983)

\bibitem{LeinfelderSimader1981}
Leinfelder, H., Simader, C.G.:
Schr\"{o}dinger operators with singular magnetic vector potentials.
{\em Math. Z.} \textbf{176}, 1--19 (1981)

\bibitem{LiebLoss2001}
Lieb, E.H., Loss, M.: {\em Analysis.} Second edition. Graduate Studies in Mathematics, vol.~14,
American Mathematical Society, Providence, Rhode Island, 2001.

\bibitem{LiskevichManavi1997}
Liskevich, V., Manavi, A.: {Dominated semigroups with singular complex potentials.}
{\em J. Funct. Anal.} \textbf{151}, {281--305} (1997)

\bibitem{LHB2011}
L\H{o}rinczi, J., Hiroshima F., Betz, V.:
\newblock {\em Feynman-Kac-type theorems and Gibbs measures on path space.}
\newblock De Gruyter Studies in Mathematics, vol. {34}, Walter de Gruyter \& Co., Berlin, 2011.

\bibitem{Matte2016}
Matte, O.: Continuity properties of the semi-group and its integral kernel in non-relativistic QED.
{\em Rev. Math. Phys.} \textbf{28}, 1650011, 90 pp. (2016)

\bibitem{Matte2017}
Matte, O.: Pauli-Fierz type operators with singular electromagnetic potentials on general domains.
{\em Math. Phys. Anal. Geom.} \textbf{20}, Art.~18, 41 pp. (2017)

\bibitem{MatteMoeller2018}
Matte, O., M{\o}ller, J.S.: {Feynman-Kac formulas for the ultra-violet renormalized Nelson model.}
{\em Ast\'{e}risque} \textbf{404}, vi+110 pp. (2018)

\bibitem{Me1982}
M\'{e}tivier, M.: {\em Semimartingales. A course on stochastic processes.}
De Gruyter Studies in Mathematics, vol. {2}, Walter de Gruyter \& Co.,  Berlin, 1982.

\bibitem{Nelson1973}
{Nelson, E.}: {The free {M}arkoff field.} {\em J. Funct. Anal.} \textbf{12}, {211--227} (1973)

\bibitem{PardouxLNM1204}
Pardoux, E.: {Grossissement d'une filtration et retournement du temps d'une diffusion.}
\newblock In: Az\'{e}ma, J., Yor, M. (Editors): {\em S\'{e}minaire de Probabilit\'{e}s XX, 1984/85.}  
Lecture Notes in Mathematics, vol. 1204, Springer, Berlin-Heidelberg, 1986. pp. 48--55.

\bibitem{Parthasarathy1992}
Parthasarathy, K.R.: {\em An introduction to quantum stochastic calculus.}
Monographs in Mathematics, vol. 85, Birkh\"{a}user, Basel, 1992.

\bibitem{PerelmuterSemenov1981}
Perelmuter, M.A., Semenov, Ya.A.: {On decoupling of finite singularities in the scattering theory for the 
Schr\"{o}dinger operator with a magnetic field.} 
{\em J. Math. Phys.} \textbf{22}, 521--533 (1981)

\bibitem{Simon1978}
Simon, B.:
\newblock {A canonical decomposition for quadratic forms with applications to monotone 
convergence theorems.}
\newblock {\em J. Funct. Anal.} \textbf{28}, 377--385 (1978)

\bibitem{Simon1978Adv}
Simon, B.:
\newblock {Classical boundary conditions as a technical tool in modern mathematical physics.}
{\em Adv. Math.} \textbf{30}, 268--281 (1978)

\bibitem{SimonJOT1979}
Simon, B.: Maximal and minimal Schr\"{o}dinger forms. 
{\em J. Operator Theory} \textbf{1}, 37--47 (1979)

\bibitem{Voigt1986}
Voigt, J.: {Absorption semigroups, their generators, and Schr\"{o}dinger semigroups.}
{\em J. Funct. Anal.} \textbf{67}, 167--205 (1986)

\end{thebibliography}
\end{document}